\documentclass[10pt]{article}
\usepackage{amssymb,amsmath,amsthm}
\usepackage[hmargin=1.5in,vmargin=1.25in]{geometry}
\usepackage{setspace}
\usepackage{natbib} 
\usepackage[colorlinks,citecolor=magenta,hypertexnames=false]{hyperref}
\usepackage{graphicx}
\usepackage{enumerate}
\usepackage{etoolbox}
\usepackage{adjustbox}
\usepackage{multirow}
\usepackage{booktabs}
\usepackage{lscape}
\usepackage{multicol}
\usepackage{threeparttable}
\allowdisplaybreaks

\renewcommand{\qed}{\rule{2mm}{2mm}}
\newcommand{\indep}{\perp \!\!\! \perp}
\DeclareMathOperator{\var}{Var}
\DeclareMathOperator{\bias}{Bias}
\DeclareMathOperator{\mse}{MSE}
\newcommand{\dd}{\mathrm{d}}

\newtheorem{theorem}{Theorem}[section]
\newtheorem{lemma}{Lemma}[section]

\newtheorem{remark}{Remark}[section]
\newtheorem{assumption}{Assumption}[section]

\AtEndEnvironment{remark}{~\qed}
\AtEndEnvironment{example}{~\qed}
\AtEndEnvironment{proof}{~\qed}

\begin{document}

\author{
Yuehao Bai \\
Department of Economics \\
University of Michigan \\
\href{mailto:yuehaob@umich.edu}{\texttt{yuehaob@umich.edu}}
}

\bigskip

\title{Optimality of Matched-Pair Designs in Randomized Controlled Trials\thanks{Bai: University of Michigan (email: yuehaob@umich.edu). Isaiah Andrews was the coeditor for this article. I am deeply grateful for the encouragement and guidance from my advisors Azeem Shaikh, Stephane Bonhomme, Alex Torgovitsky, and Leonardo Bursztyn. I thank the coeditor and the referees for numerous comments that helped improve the paper greatly. I thank Marinho Bertanha, John Bound, Charlie Brown, Aibo Gong, Florian Gunsilius, Sara Heller, Wooyong Lee, Jonathan Roth, Joshua Shea, Max Tabord-Meehan, Dean Yang, Basit Zafar, and Qinyue Zhou for extensive feedback on the paper. I thank Rex Hsieh and Jiehan Xu for outstanding research assistance. I also thank Bobbie Goettler for excellent copy editing.}}

\maketitle

\begin{spacing}{1.2}
\begin{abstract}
In randomized controlled trials (RCTs), treatment is often assigned by stratified randomization. I show that among all stratified randomization schemes which treat all units with probability one half, a certain matched-pair design achieves the maximum statistical precision for estimating the average treatment effect (ATE). In an important special case, the optimal design pairs units according to the baseline outcome. In a simulation study based on datasets from 10 RCTs, this design lowers the standard error for the estimator of the ATE by 10\% on average, and by up to 34\%, relative to the original designs.
\end{abstract}
\end{spacing}



\thispagestyle{empty}
\newpage
\setcounter{page}{1}

This paper studies the optimality of matched-pair designs in randomized controlled trials (RCTs). Matched-pair designs are examples of stratified randomization, in which the researcher partitions a set of units into strata (groups) based on their observed covariates and assigns a fraction of units in each stratum to treatment. A matched-pair design is a stratified randomization scheme with two units in each stratum.

Stratified randomization is prevalent in economics. Among the 5,000 RCTs in the AEA RCT Registry, more than 800 are stratified. The schemes in these papers, however, differ vastly in terms of the covariates used to stratify and how fine the strata are. Among these 800 RCTs, around 50 use matched-pair designs. Moreover, 56\% of the researchers interviewed in \cite{bruhn2009pursuit} have used matched-pair designs at some point in their research. Yet, despite the frequency with which applied researchers make decisions about how to stratify, there are few general econometric results on whether matched-pair designs lead to better precision of estimators of treatment effects than other stratified randomization schemes and the best way to pair units.

I derive the exact form of the stratified randomization scheme that has the maximum statistical precision for estimating the average treatment effect (ATE). The optimal scheme is a matched-pair design. In an important special case, the optimal design is to order the units according to the baseline values of the primary outcome variable of interest and then pair the adjacent units. When I simulate this simple design using data from 10 recent papers in the \textit{American Economic Journal: Applied Economics}, I find it lowers the standard error of the difference-in-means estimator by 10\% on average, and by up to 34\%, relative to the designs actually used in those studies. I also find some more complicated stratifications with strata of four units according to multiple covariates could further lower both the MSE and the standard error. Based on these findings, I make practical recommendations across a wide range of empirical settings.

In Section \ref{sec:oracle}, I study settings where the treated fractions are identically $\frac{1}{2}$ across strata. In such settings, a common estimator for the ATE is the difference in the means of the treated and control groups. The properties of the difference-in-means estimator, however, vary substantially with how the researcher stratifies. To begin, consider the thought experiment where we know the distributions of the potential outcomes. Let $Y(1)$ denote the potential outcome if a unit is treated and let $Y(0)$ denote the potential outcome if it is not treated. Let $X$ denote the observed, baseline covariates. I define an index function $E[Y(1) + Y(0) | X]$, the expected sum of the potential outcomes given the covariates. My first result shows the mean-squared error (MSE) of the difference-in-means estimator is minimized by a matched-pair design, where units are ordered according to this index function and paired adjacently. My optimality result holds at any sample size and without any distributional assumption beyond the existence of moments. In particular, my result does not rely on restrictions on treatment effects heterogeneity.

I describe a special case where the optimal stratification is feasible even without knowing the index function. Suppose $X$ contains a single covariate. Further suppose both $Y(1)$ and $Y(0)$ are higher in expectation when $X$ is higher, so that $E[Y(1) + Y(0) | X]$ is increasing in $X$. In this case, pairing units according to $X$ is optimal. An important example in empirical practice is when $X$ is the baseline value of the primary outcome variable of interest. For instance, in \cite{angrist2009effects}, the primary outcome variable of interest is a test score and the treatment is an educational program, so we expect a higher baseline test score ($X$) implies a higher endline test score ($Y(1)$ and $Y(0)$) in expectation.

If researchers are unsure about the monotonicity condition, or if multiple covariates are available, then the optimal stratification is generally unknown because the index function is generally unknown. As such, Section \ref{sec:practical} studies several feasible procedures. With multiple covariates, I study pairing units to minimize the (Mahalanobis) distances of the covariates. In settings with auxiliary data, such as data from pilot studies, I propose several matched-pair designs in which the index function is approximated by a proxy based on the auxiliary data.

In Section \ref{sec:asymptotic}, to compare the performance of these practical procedures, I study the asymptotic properties of the difference-in-means estimator. I show that relative to not stratifying, pairing according to any function of the covariates can only reduce the limiting variance of the difference-in-means estimator. Moreover, the limiting variance is lower if the stratifying variables explain a larger proportion of the variation in $Y(1) + Y(0)$.

In Section \ref{sec:sims}, I conduct a simulation study using data from a systematically selected set of 10 RCTs from recent issues of the \textit{American Economic Journal: Applied Economics}. Relative to the original stratifications used in those 10 papers, if the researchers had just paired the units according to their baseline outcomes, then the MSE of the difference-in-means estimator would be 24\% smaller on average and 56\% smaller in some cases. The standard error of the difference-in-means estimator would be 10\% smaller on average and 34\% smaller in some cases.

Among all methods in the simulation, pairing units to minimize the sum of the squared Mahalanobis distances of the covariates usually leads to the smallest MSEs. When the number of covariates is large, however, the standard error could be even larger than the standard error when pairing according to the baseline outcome alone. Intuitively, this is because the quality of the variance estimator is lower when the curse of dimensionality is more severe. An alternative that balances the MSE and the standard error is to match units into sets of four, instead of pairs, to minimize the sum of the squared Mahalanobis distance of the covariates. Such a method has both smaller MSEs and standard errors than pairing according to the baseline outcome alone while being computationally more intensive.

I conclude with recommendations for empirical practice in Section \ref{sec:conclusion}. I recommend different stratifications based on the availability of auxiliary datasets and whether one main outcome of interest clearly dominates the others. All of my recommended procedures are defined by pairing units or matching units into sets of four according to all or a subset of the available covariates.

\paragraph{Related Literature} This paper is most closely related to \cite{barrios2013optimal} and \cite{tabord-meehan2020stratification}. \cite{barrios2013optimal} studies minimizing the variance of the difference-in-means estimator. He is the first to show pairing units according to my index function is optimal among all matched-pair designs, albeit under the assumption of homogeneous treatment effects. My optimality result holds among all stratified randomization schemes and with heterogeneous treatment effects. \cite{tabord-meehan2020stratification} studies optimality within a class of stratification trees. Because the number of strata is fixed in his asymptotic framework, he can optimize over the treated fraction in each stratum. In a matched-pair design, the number of strata is half of the sample size and hence not fixed as the sample size increases, so matched-pair designs are precluded in his framework. In Section \ref{sec:unequal} of the supplement, I elaborate on the comparison between the two papers and further note that combining our procedures is straightforward.

The following papers also study matched-pair designs: \cite{greevy2004optimal} study pairing units to minimize the sum of the squared Mahalanobis distances of the covariates. \cite{imai2008variance} studies matched-pair designs, focusing on the sample ATE. The inference methods in this paper build on and extend those in \cite{bai2021inference}. In addition, inference under matched-pair designs has also been studied in \cite{abadie2008estimation}, who assume a different sampling framework, \cite{fogarty2018mitigating,fogarty2018regression-assisted}, who provides conservative estimators for the limiting variance, and \cite{de_chaisemartin2021at} in a finite-population setting.

\section{Setup and Notation} \label{sec:setup}
Let $Y_i$ denote the observed outcome of interest for the $i$th unit, let $D_i$ denote the treatment status for the $i$th unit, and let $X_i$ denote the observed, baseline covariates for the $i$th unit. Further denote by $Y_i(1)$ the potential outcome of the $i$th unit if treated and by $Y_i(0)$ if not treated. As usual, the observed outcome is related to the potential outcomes and treatment status by the relationship
\[ Y_i = Y_i(1) D_i + Y_i(0) (1 - D_i)~. \]
For ease of exposition, I assume the sample size is even and denote it by $2n$. I assume $((Y_i(1), Y_i(0), X_i): 1 \leq i \leq 2n)$ is an i.i.d. sequence of random vectors. Note the potential outcomes and the covariates are drawn from a population and hence are random instead of fixed. For any random vector indexed by $i$, $A_i$, define $A^{(n)} = (A_1, \dots, A_{2n})'$. The main parameter of interest is the average treatment effect (ATE):
\[ \theta = E[Y_i(1) - Y_i(0)]~. \]

In stratified randomization, I first partition the set of units into strata. Formally, I define a stratification $\lambda = \{\lambda_s: 1 \leq s \leq S\}$ as a partition of $\{1, \dots, 2n\}$:
\begin{enumerate}[\rm (a)]
\item $\lambda_s \bigcap \lambda_{s'} = \varnothing$ for all $s$ and $s'$ such that $1 \leq s \neq s' \leq S$.
\item $\bigcup\limits_{1 \leq s \leq S} \lambda_s = \{1, \dots, 2n\}$.
\end{enumerate}

\noindent Let $\Lambda_n$ denote the set of all stratifications of $2n$ units. Define $n_s = |\lambda_s|$ and $\tau_s$ as the treated fraction in stratum $\lambda_s$. A matched-pair design is simply a stratified randomization scheme with $S = n$ and $n_s = 2$ for $1 \leq s \leq S$. I define $\Lambda_n^{\rm pair} \subseteq \Lambda_n$ as the set of all matched-pair designs for $2n$ units.

I make the following assumption on the treatment assignment scheme:

\begin{assumption} \rm \label{as:half}
Given the covariates $X^{(n)}$, treatment status is determined as follows: independently for $1 \leq s \leq S$, uniformly at random choose $n_s \tau_s$ units in $\lambda_s$, and assign $D_i = 1$ to them and $D_i = 0$ to the other units in $\lambda_s$. Furthermore, $\tau_s = \frac{1}{2}$ for $1 \leq s \leq S$.
\end{assumption}

Assumption \ref{as:half} implies
\begin{equation} \label{eq:indep}
(Y^{(n)}(0), Y^{(n)}(1)) \indep D^{(n)} | X^{(n)}~.
\end{equation}
In other words, treatment status and potential outcomes are conditionally independent given the covariates. Assumption \ref{as:half} also implies $n_s$ has to be even because a unit cannot be cut in half. Note the distribution of the vector of treatment status $D^{(n)}$ depends on $\lambda$. Most results below can be extended to settings where $\tau_s, 1 \leq s \leq S$ are identical but not $\frac{1}{2}$, or where they are additionally allowed to vary across subpopulations. See Remark \ref{remark:tau} for details.

For all treatment assignment schemes in the main text, I estimate the ATE by the difference in the means of the treated and control groups. Formally, for $d \in \{0, 1\}$, define
\[ \hat \mu_n(d) = \frac{1}{n} \sum_{1 \leq i \leq 2n: D_i = d} Y_i~. \]
The difference-in-means estimator is defined as
\[ \hat \theta_n =  \hat \mu_n(1) - \hat \mu_n(0)~. \]
The difference-in-means estimator is widely used because it is simple and transparent. Under Assumption \ref{as:half}, it coincides with the OLS estimator for the coefficient in the linear regression of the outcome on treatment status and strata fixed effects and the OLS estimator from the fully saturated version of that regression, both of which are also widely used in analyses of RCTs. See, for example, \cite{duflo2007using}, \cite{glennerster2013running}, and \cite{crepon2015estimating}.

\section{Optimal Stratification} \label{sec:oracle}
This section studies the optimal stratification. To preview the results, define the index function
\begin{equation} \label{eq:g}
g(x) = E[Y_i(1) + Y_i(0) | X_i = x]~.
\end{equation} 
I show the optimal stratification is given by ordering the units according to $g_i = g(X_i)$ and then pairing the adjacent units. In the special case where $X_i$ is a scalar and $E[Y_i(1) | X_i = x]$ and $E[Y_i(0) | X_i = x]$ are both weakly increasing (or both weakly decreasing) in $x$, the optimal stratification is given by ordering the units according to $X_i$ and then pairing the adjacent units.

The analysis in this section is conditional on $X^{(n)}$. In this section only, instead of the population ATE, I focus on the ATE conditional on $X^{(n)}$:
\[ \theta_n = \frac{1}{2n} \sum_{1 \leq i \leq 2n} E[Y_i(1) - Y_i(0) | X_i]~. \]
Focusing on $\theta_n$ simplifies the discussion. Moreover, conditional on a fixed sample with covariates $X^{(n)}$, I can only hope to be unbiased for $\theta_n$ instead of $\theta$. The conclusions of the theorems in this section are the same regardless of whether the parameter of interest is $\theta_n$ or $\theta$.

My objective function is the MSE of $\hat \theta_n$ for $\theta_n$ conditional on $X^{(n)}$ under a stratification $\lambda \in \Lambda_n$:
\[ \mse (\lambda | X^{(n)}) = E_\lambda[(\hat \theta_n - \theta_n)^2 | X^{(n)}]~. \]
Here, the notation $E_\lambda$ indicates the distribution of the vector of treatment status $D^{(n)}$ depends on the stratification. I consider minimizing the conditional MSE over the set of all stratifications:
\begin{equation} \label{eq:min-mse}
\min_{\lambda \in \Lambda_n} ~ \mse (\lambda | X^{(n)})~.
\end{equation}
In what follows, I derive the optimal stratification as the solution to \eqref{eq:min-mse}. I emphasize that by a simple bias-variance decomposition, one can show \eqref{eq:min-mse} is equivalent to the problem where $\theta_n$ is replaced by $\theta$, so focusing on $\theta_n$ in this section is genuinely without loss of generality.

Solving \eqref{eq:min-mse} involves two intermediate results, each carrying additional insights into the problem. To describe the first intermediate result, I define the ex-ante bias of $\hat \theta_n$ for $\theta_n$ conditional on $X^{(n)}$ as
\[ \bias_{n, \lambda}^{\rm ante}(\hat \theta_n | X^{(n)}) = E_\lambda[\hat \theta_n | X^{(n)}] - \theta_n~, \]
and the ex-post bias of $\hat \theta_n$ for $\theta_n$ conditional on $X^{(n)}$ and $D^{(n)}$ as
\[ \bias_n^{\rm post}(\hat \theta_n | X^{(n)}, D^{(n)}) = E[\hat \theta_n | X^{(n)}, D^{(n)}] - \theta_n~. \]
Here, ex-ante bias refers to the bias conditional only on the covariates, before treatment status is realized; ex-post bias refers to the bias conditional on both the covariates and treatment status, after treatment status is realized. Note in the definition of the ex-post bias, the $\lambda$ subscript does not appear because $D^{(n)}$ is already given. Note from the definition of the difference-in-means estimator that
\begin{equation} \label{eq:diffim-d}
\hat \theta_n = \frac{1}{n} \sum_{1 \leq i \leq 2n} (Y_i(1) D_i - Y_i(0) (1 - D_i))~.
\end{equation}
By Assumption \ref{as:half}, the marginal treatment probability of each unit satisfies $E_\lambda[D_i | X^{(n)}] = \frac{1}{2}$, and together with the conditional independence assumption in \eqref{eq:indep}, they imply
\[ E_\lambda[\hat \theta_n | X^{(n)}] = \theta_n~. \]
Therefore, the ex-ante bias is identically zero across $\lambda \in \Lambda_n$, which is not surprising because the ex-ante bias should be zero if we run an experiment. By the law of iterated expectations,
\[ E_\lambda[\bias_n^{\rm post}(\hat \theta_n | X^{(n)}, D^{(n)}) | X^{(n)}] = \bias_{n, \lambda}^{\rm ante}(\hat \theta_n | X^{(n)}) = 0~, \]
so the mean of the ex-post bias over the distribution of treatment status equals the ex-ante bias, which is zero.

The first intermediate result is a decomposition of the conditional MSE in \eqref{eq:min-mse}. Because $E_\lambda[\hat \theta_n - \theta_n | X^{(n)}] = 0$, by the law of total variance,
\begin{multline} \label{eq:var}
\mse(\lambda | X^{(n)}) = \var[\hat \theta_n - \theta_n | X^{(n)}] \\
= E_\lambda[\var[\hat \theta_n - \theta_n| X^{(n)}, D^{(n)}] | X^{(n)}] + \var_\lambda[E[\hat \theta_n | X^{(n)}, D^{(n)}] - \theta_n | X^{(n)}]~.
\end{multline}
For any $\lambda \in \Lambda_n$, the first term on the right-hand side of \eqref{eq:var} equals
\begin{multline*}
E_\lambda \Big [ \frac{1}{n^2} \sum_{1 \leq i \leq 2n} (\var[Y_i(1) | X_i] D_i + \var[Y_i(0) | X_i] (1 - D_i)) \Big | X^{(n)} \Big ] \\
= \frac{1}{2n^2} \sum_{1 \leq i \leq 2n} (\var[Y_i(1) | X_i] + \var[Y_i(0) | X_i])~,
\end{multline*}
which is identical across all $\lambda \in \Lambda_n$. Note I used the conditional independence assumption in \eqref{eq:indep}, the facts that $\theta_n$ is a constant given $X^{(n)}$, that $D_i(1 - D_i) = 0$ for $1 \leq i \leq 2n$, and that $E_\lambda[D_i | X^{(n)}] = \frac{1}{2}$. Hence, \eqref{eq:min-mse} is further equivalent to minimizing the second term on the right-hand side of \eqref{eq:var}, which is the variance of the ex-post bias:
\[ \var_\lambda[\bias_n^{\rm post}(\hat \theta_n | X^{(n)}, D^{(n)}) | X^{(n)}]~. \]
The discussion so far leads to my first intermediate result:

\begin{lemma} \label{lem:post}
Suppose the treatment assignment scheme satisfies Assumption \ref{as:half}. Then, \eqref{eq:min-mse} is equivalent to
\[ \min_{\lambda \in \Lambda_n} ~ \var_\lambda[\bias_n^{\rm post}(\hat \theta_n | X^{(n)}, D^{(n)}) | X^{(n)}]~. \]
\end{lemma}

Next, I describe the second intermediate result in solving \eqref{eq:min-mse}. The result states any stratification is a convex combination of matched-pair designs. Formally, for $\lambda, \lambda' \in \Lambda_n^{\rm pair}$ and $\delta \in [0, 1]$, define $\delta \lambda \oplus (1 - \delta) \lambda'$ as the randomization between $\lambda$ and $\lambda'$ such that $\lambda$ is implemented with probability $\delta$. Define the convex hull formed by all convex combinations of any finite number of matched-pair designs as
\begin{multline*}
\mathrm{co}(\Lambda_n^{\rm pair}) = \Big \{ \bigoplus_{1 \leq j \leq J} \delta_j \lambda^j: \lambda^j \in \Lambda_n^{\rm pair}, \\
\delta_j \geq 0 \text{ for } 1 \leq j \leq J, \sum_{1 \leq j \leq J} \delta_j = 1, 1 \leq J < \infty \Big \}~.
\end{multline*}
In other words, a member of the convex hull is the ``mixing'' of $J$ matched-pair designs, where $J$ is finite.

For example, suppose $2n = 4$. Then, four stratifications are possible:
\begin{align*}
\lambda^0 & = \{\{1, 2, 3, 4\}\} \\
\lambda^1 & = \{\{1, 2\}, \{3, 4\}\} \\
\lambda^2 & = \{\{1, 3\}, \{2, 4\}\} \\
\lambda^3 & = \{\{1, 4\}, \{2, 3\}\}~.
\end{align*}
$\lambda^0$ puts four units in the same stratum. $\lambda^1$ pairs 1 and 2 together and 3 and 4 together. $\lambda^2$ and $\lambda^3$ are defined similarly. Note that implementing each of the three matched-pair designs with probability $1 / 3$ is equivalent to implementing $\lambda^0$, in the sense that the distributions of $(D_1, D_2, D_3, D_4)$ are the same under the two implementations. Indeed, under $\lambda^1$, $(D_1, D_2, D_3, D_4)$ takes the following four values each with probability $1/4$: $(1, 0, 1, 0)$, $(1, 0, 0, 1)$, $(0, 1, 1, 0)$, $(0, 1, 0, 1)$. Similarly, under $\lambda^2$, it takes the following four values each with probability $1/4$: $(1, 0, 0, 1)$, $(1, 1, 0, 0)$, $(0, 0, 1, 1)$, $(0, 1, 1, 0)$. Under $\lambda^3$, it takes the following four values each with probability $1/4$: $(1, 0, 1, 0)$, $(1, 1, 0, 0)$, $(0, 1, 0, 1)$, $(0, 0, 1, 1)$. Accordingly, under $\frac{1}{3} \lambda^1 \oplus \frac{1}{3} \lambda^2 \oplus \frac{1}{3} \lambda^3$, it takes the following six values each with probability $1/6$: $(1, 1, 0, 0)$, $(1, 0, 1, 0)$, $(1, 0, 0, 1)$, $(0, 1, 1, 0)$, $(0, 1, 0, 1)$, $(0, 0, 1, 1)$. This distribution is the same as that of $(D_1, D_2, D_3, D_4)$ under $\lambda^0$, where two out of four units are treated uniformly at random. As a result, $\lambda^0 \in \mathrm{co}(\{\lambda^1, \lambda^2, \lambda^3\})$, meaning $\lambda^0$ can be written as a convex combination of the three matched-pair designs.

I show in Section \ref{sec:proof} of the supplement that the result above holds in general and summarize it into the following lemma:

\begin{lemma} \label{lem:mixing}
If the treatment assignment scheme satisfies Assumption \ref{as:half}, then $\Lambda_n \subseteq \mathrm{co}(\Lambda_n^{\rm pair})$. In other words, any stratification is a convex combination of matched-pair designs.
\end{lemma}

Combining Lemmas \ref{lem:post}--\ref{lem:mixing} to minimize the MSE as in \eqref{eq:min-mse} is now straightforward. To state the result, I need an equivalent notation for matched-pair designs. Recall that a permutation of $\{1, \dots, 2n\}$ is a function that maps $\{1, \dots, 2n\}$ onto itself. Let $\Pi_n$ denote the group of all permutations of $\{1, \dots, 2n\}$. A matched-pair design is a stratified randomization scheme with
\[ \lambda = \{\{\pi(2s - 1), \pi(2s)\}: 1 \leq s \leq n\}~, \]
where $\pi \in \Pi_n$. Recall the definition of the index function $g$ in \eqref{eq:g} and order the units by defining $\pi^g \in \Pi_n$ that satisfies $g_{\pi^g(1)} \leq \dots \leq g_{\pi^g(2n)}$. Define the stratification
\begin{equation} \label{eq:oracle}
\lambda^g(X^{(n)}) = \{\{\pi^g(2s - 1), \pi^g(2s)\}: 1 \leq s \leq n\}~.
\end{equation}
The stratification in \eqref{eq:oracle} is given by ordering the units according to $g_i$ and then pairing the adjacent units. I now show it minimizes the MSE as in \eqref{eq:min-mse}.

For each $\lambda \in \Lambda_n$, define $V(\lambda)$ as the objective in Lemma \ref{lem:post}. Recall $g^{(n)} = (g_1, \ldots, g_n)'$. Then,
\begin{align}
\nonumber V(\lambda) & = \var_\lambda[E[\hat \theta_n | X^{(n)}, D^{(n)}] - \theta_n | X^{(n)}] \\
\nonumber & = \frac{1}{n^2} \var_\lambda \Big [ \sum_{1 \leq i \leq 2n} (D_i E[Y_i(1) | X_i] - (1 - D_i) E[Y_i(0) | X_i]) \Big | X^{(n)} \Big ] \\
\nonumber & = \frac{1}{n^2} \var_\lambda \Big [ \sum_{1 \leq i \leq 2n} D_i(E[Y_i(0) | X_i] + E[Y_i(1) | X_i]) \Big | X^{(n)} \Big ] \\
\nonumber & = \frac{1}{n^2} (g^{(n)})' \var_\lambda[D^{(n)} | X^{(n)}] g^{(n)}~,
\end{align}
where the first equality follows from the definition of the ex-post bias, the second equality follows from \eqref{eq:diffim-d} and the fact that $\theta_n$ is a constant given $X^{(n)}$, and the last two equalities follow by inspection. Recall the variance of $D_i$ is $\frac{1}{4}$. Also recall that for a matched-pair design, the covariance between treatment status of the two units in a pair is $- \frac{1}{4}$, and that of units across pairs is 0. Therefore, for any $\lambda = \{\{\pi(1), \pi(2)\}, \allowbreak \dots, \{\pi(2n - 1), \pi(2n)\}\} \in \Lambda_n^{\rm pair}$,
\[ V(\lambda) = \frac{1}{4n^2}\sum_{1 \leq s \leq n} (g_{\pi(2s-1)} - g_{\pi(2s)})^2~. \]
Therefore, $V(\lambda)$ is proportional to the sum of squared distances of $g$ within each pair. By Lemma \ref{lem:hardy-littlewood} in the supplement, which is a simple consequence of the Hardy-Littlewood-P\'olya rearrangement inequality, $V(\lambda^g(X^{(n)})) \leq V(\lambda)$ for any $\lambda \in \Lambda_n^{\rm pair}$. Therefore, $\lambda^g(X^{(n)})$ minimizes the MSE among $\Lambda_n^{\rm pair}$, the set of all matched-pair designs.

To conclude $\lambda^g(X^{(n)})$ is optimal among the set of all stratifications $\Lambda_n$, note each stratification is a mixing of matched-pair designs, and no ``mixed strategy'' has a better payoff than the optimal ``pure strategy.'' Formally, by Lemma \ref{lem:mixing}, any $\lambda \in \Lambda_n$ can be written as
\[ \lambda = \bigoplus_{1 \leq j \leq J} \delta_j \lambda^j~, \]
where $\lambda^j \in \Lambda_n^{\rm pair}$, $\delta_j \geq 0$ for $1 \leq j \leq J$, and $\sum_{1 \leq j \leq J} \delta_j = 1$. As a result,
\begin{align*}
\mse(\lambda | X^{(n)}) & = \sum_{1 \leq j \leq J} \delta_j \mse(\lambda^j | X^{(n)}) \\
& \geq \min_{1 \leq j \leq J} \mse(\lambda^j | X^{(n)}) \geq \mse(\lambda^g(X^{(n)}) | X^{(n)})~,
\end{align*}
where the equality follows from the definition of the MSE, the first inequality follows because any weighted average of a set of numbers is weakly larger than the minimum across them, and the last inequality follows because $\lambda^g(X^{(n)})$ minimizes $\mse(\lambda | X^{(n)})$ across $\Lambda_n^{\rm pair}$. Therefore, I have established my main theorem on the optimal stratification:

\begin{theorem} \label{thm:oracle}
Suppose the treatment assignment scheme satisfies Assumption \ref{as:half}. Then, the matched-pair design defined in \eqref{eq:oracle} minimizes the MSE as in \eqref{eq:min-mse}. In other words, the optimal stratification is given by ordering the units according to $g_i$ and then pairing the adjacent units.
\end{theorem}

\begin{remark} \rm
Note the optimal stratification does not depend on knowledge of the conditional variances of $Y_i(1)$ and $Y_i(0)$ given $X_i$.
\end{remark}

\begin{remark} \rm \label{remark:tau}
Theorem \ref{thm:tau-oracle} in the supplement examines settings where the treated fractions are identical across strata but not $\frac{1}{2}$. Formally, suppose $\tau_s = \tau = \frac{l}{k}$ for $1 \leq s \leq S$, where $l, k \in \mathbf N$, $0 < l < k$, and $l$ and $k$ are mutually prime. Define
\begin{equation} \label{eq:tau-g}
g^\tau(X_i) = \frac{E[Y_i(1) | X_i]}{\tau} + \frac{E[Y_i(0) | X_i]}{1 - \tau}~.
\end{equation}
$g^\tau$ adjusts for the treatment probability by inverse probability weighting. The optimal stratification is defined by the following algorithm:
\begin{enumerate}[(a)]
\item Order the units according to $g^\tau(X_i)$.
\item Put the first $k$ units in the first stratum, the second $k$ units in the second stratum, and so on.
\item Uniformly at random assign $l$ of the $k$ units in each stratum to treatment.
\end{enumerate}
In this case, the optimal design is not paired, but stratified randomization with the appropriate group size remains optimal. For examples in this spirit of small strata, see \cite{bold2018experimental} and \cite{brown2020inducing}.
\end{remark}

\begin{remark} \rm \label{remark:rerand}
Re-randomization, studied by \cite{morgan2012rerandomization,morgan2015rerandomization}, is an alternative to stratified randomization. Re-randomization takes random draws of treatment status until it falls in an admissible set. The admissible set is usually defined as the collection of treatment assignments under which the distance between the treated and control units is below a threshold. The notion of distance can be, for instance, the (Mahalanobis) distance in the covariates or the distance in $g$. In matched-pair designs, units are matched to \emph{minimize} the distance between treated and control units. As such, each possible realization of the vector of treatment status under a matched-pair design not only belongs to the admissible set but also attains the smallest distance within the admissible set. For example, suppose each distinct value of the covariate appears twice in the sample. Then, a matched-pair design is equivalent to re-randomization with the distance threshold set to zero.
\end{remark}

Note from \eqref{eq:g} that the index function $g_i$ is a scalar regardless of the dimension of $X_i$. Moreover, the optimal stratification depends not on the values but merely on the ordering of $g_i$. For instance, if $X_i$ is univariate and $g(x)$ is monotonic in $x$, then the optimal stratification in \eqref{eq:oracle} is given by ordering the units by $X_i$ and then pairing the adjacent units. This scenario arises in many settings, especially if $X_i$ is the baseline value of the primary outcome variable of interest, which is collected in the baseline survey before treatment is assigned. For instance, \cite{angrist2009effects} study the effect of an educational program on test scores. In their paper, $X_i$ is the baseline test score, so we expect both $E[Y_i(1) | X_i]$ and $E[Y_i(0) | X_i]$ are weakly increasing in $X_i$. I record this result as a theorem. Let $\pi^X \in \Pi_n$ be such that $X_{\pi^X(1)} \leq \dots \leq X_{\pi^X(2n)}$.

\begin{theorem} \label{thm:monotonic}
Suppose $X_i$ is univariate, the treatment assignment scheme satisfies Assumption \ref{as:half}, and $g(x)$ in \eqref{eq:g} is monotonic in $x$. Then,
\[ \lambda^g(X^{(n)}) = \{\{\pi^X(2s - 1), \pi^X(2s)\}: 1 \leq s \leq n\}~. \]
In other words, the optimal stratification is given by ordering the units according to their covariate values and then pairing the adjacent units.
\end{theorem}

\section{Feasible Procedures} \label{sec:practical}
The optimal stratification in Theorem \ref{thm:oracle} depends on the index function $g$, which is generally unknown, so the optimal stratification is also generally unknown. Therefore, researchers often need to approximate the index function with some proxies, possibly with the help of auxiliary data. This section studies a wide range of feasible stratification methods. Some procedures are based on data from pilot experiments, which are smaller-scale copies of the main experiment run on the same population. Depending on the availability of a pilot experiment and its sample size, different procedures are available. I switch the parameter of interest back to the population ATE $\theta$, recalling that all results in the previous section hold for both $\theta_n$ and $\theta$.

\subsection{Settings without Pilot Data} \label{sec:none}
According to Theorem \ref{thm:monotonic}, if $X_i$ is univariate and the index function $g(x)$ is monotonic in $x$, then the optimal stratification is given by pairing units according to $X_i$. A prominent example is where $X_i$ is the baseline value of the primary outcome variable of interest, and $E[Y_i(1) | X_i = x]$ and $E[Y_i(0) | X_i = x]$ are both weakly increasing or both weakly decreasing in $x$. 

Even if the monotonicity condition fails, units can still be paired according to their baseline outcomes. Theorem \ref{thm:limit} and Remark \ref{remark:loss} below study the limiting variance of the difference-in-means estimator. They reveal that if we need to choose a single covariate to pair on, the smallest limiting variance is attained by pairing units according to a covariate that explains the largest proportion of the variation in the potential outcomes. \cite{bruhn2009pursuit} note the baseline outcome is often such a covariate. Simulation evidence in Section \ref{sec:sims} further shows pairing units according to the baseline outcome performs better than the status-quo methods in terms of both the MSE and the standard error of the difference-in-means estimator.

Regardless of whether the baseline outcome is available, if $X_i$ is multivariate, then researchers can also pair units to minimize the sum of the squared Mahalanobis distances of the covariates:
\begin{equation} \label{eq:maha}
d(x_1, x_2) = (x_1 - x_2)' \hat \Sigma_n^{-1} (x_1 - x_2)~.
\end{equation}
Here, $\hat \Sigma_n$ is the sample variance matrix of $X$. \eqref{eq:maha} is simply the squared Euclidean distance if $\hat \Sigma_n$ is the identity matrix, and $\hat \Sigma_n^{-1}$ serves as a scale normalization because different covariates may be measured in different units or have different standard deviations. Note
\[ d(x_1, x_2) = \|\hat \Sigma_n^{-1/2} (x_1 - x_2)\|^2~, \]
where $\hat \Sigma_n^{-1/2}$ is the square root of $\hat \Sigma_n$. So the Mahalanobis distance between $x_1$ and $x_2$ equals the Euclidean distance between $\hat \Sigma_n^{-1/2} x_1$ and $\hat \Sigma_n^{-1/2} x_2$.

When the baseline outcome is unavailable but a large amount of auxiliary data is available, I can calculate a sample counterpart of \eqref{eq:oracle}. In general, the auxiliary data needs to come from pilot experiments, but in one special case, even observational data suffices. If the conditional ATEs are homogeneous, meaning
\begin{equation} \label{eq:hom}
E[Y_i(1) - Y_i(0) | X_i] = E[Y_i(1) - Y_i(0)] \text{ with probability one}~,
\end{equation}
then the ordering of $g_i$ is the same as that of $E[Y_i(0) | X_i]$. Suppose we have an observational dataset where the distribution of $(Y_i(0), X_i)$ is the same as that in the main experiment. As an example, suppose in an RCT to study the effect of educational program on test scores, the researcher has administrative data on the test scores of the previous cohort, and the distributions of $(Y_i(0), X_i)$ are the same across the two cohorts. Then, they can estimate $E[Y_i(0) | X_i = x]$ by a nonparametric regression using the data for the previous cohort and pair the units in the current cohort according to the predicted values in the regression. A key requirement is that the estimator for $E[Y_i(0) | X_i = x]$ is consistent in the sense of Assumption \ref{as:l2} and Theorem \ref{thm:l2} below. Then, as the sample sizes of the auxiliary data and the main experiment both increase, the limiting variance of $\hat \theta_n$ when units are paired according to the predicted values of $E[Y_i(0) | X_i]$ is the same as that under the optimal stratification in \eqref{eq:oracle}.

\subsection{Settings with Large Pilots} \label{sec:large}
Next, I consider settings with data from a pilot experiment. Let $m$ denote the sample size of the pilot experiment. I assume the pilot units are drawn from the same population as the main experiment. 

I start by investigating settings where the sample size of the pilot experiment is large. Formally, in the asymptotic framework, I allow both $m$ and $n$ to go to infinity. I pair units according to a suitable estimator $\tilde g_m$ of the index function $g$, where $\tilde g_m$ comes from a nonparametric regression using the pilot data. Again, a key requirement is that $\tilde g_m$ is consistent for $g$ in the sense of Assumption \ref{as:l2} and Theorem \ref{thm:l2} below. Then, as the sample sizes of the auxiliary data and the main experiment both increase, the limiting variance of $\hat \theta_n$ when units are paired according to $\tilde g_m$ is the same as that under the optimal stratification in \eqref{eq:oracle}.

If the pilot data is imperfect in the sense that it does not come from the same population as the main experiment, or if the estimation method for constructing $\tilde g_m$ is not flexible enough, then $\tilde g_m$ may not converge to $g$ but instead to another function $h$. In that case, the limiting variance of $\hat \theta_n$ is different from that under the optimal stratification in \eqref{eq:oracle} and depends on $h$, but it is still smaller than that under no stratification. See Theorem \ref{thm:l2} and Remark \ref{remark:compare} for details.

\subsection{Settings with Small Pilots} \label{sec:small}
In practice, even if pilot data is available, its sample size is often small. In those settings, pairing units according to $\tilde g_m$ generally does not ensure efficiency, unlike in settings with large pilots. We may be concerned that $\tilde g_m$ is a poor approximation of the index function $g$, and as a result, if units are paired according to $\tilde g_m$, then both the conditional MSE and the limiting variance of $\hat \theta_n$ are large.

Researchers could of course ignore the information in the small pilot and implement the procedures in Section \ref{sec:none}. If they would like to incorporate information from the pilot experiment, they can consider the following procedure. For $d \in \{0, 1\}$, let $\tilde \beta_m(d)$ denote the OLS estimators of the linear regression coefficients among the treated or untreated units in the pilot experiment and let $\tilde \Omega_m(d)$ denote the variance estimators in OLS assuming homoskedasticity (see Section \ref{sec:pen-detail} of the supplement for details). Further define
\begin{align*}
\tilde \beta_m & = \tilde \beta_m(1) + \tilde \beta_m(0) \\
\tilde \Omega_m & = \tilde \Omega_m(1) + \tilde \Omega_m(0)~.
\end{align*}
I pair the units to minimize the sum of the following distances of the covariates:
\begin{equation} \label{eq:pen-metric}
d^{\rm pen}(x_1, x_2) = (x_1' \tilde \beta_m - x_2' \tilde \beta_m)^2 + (x_1 - x_2)' \tilde \Omega_m (x_1 - x_2)~.
\end{equation}

To shed some light on the behavior of such a minimization problem, I consider two extreme cases. If $\tilde \Omega_m = 0$, which means $\tilde \beta_m$ is very precise, then the solution is given by pairing units according to $\tilde g_m = X_i' \tilde \beta_m$. If $\tilde \Omega_m$ is large, which means $\tilde \beta_m$ is very imprecise, then the second term on the right-hand side of \eqref{eq:pen-metric} dominates the first term, so the solution is close to a paired matching weighted by $\tilde \Omega_m$. Therefore, the solution can be viewed as penalizing the pairing according to $\tilde g_m(x) = x' \tilde \beta_m$, with the penalization determined by the variance estimator $\tilde \Omega_m$. I refer to the solution as the penalized matched-pair design. In Section \ref{sec:pen-detail} of the supplement, I show it is optimal in a Bayesian framework.

\subsection{Other Practical Considerations}
Each matched-pair design discussed in this section has a counterpart where units are matched into sets of four instead of pairs. Specifically, I first pair the units and then pair the pairs using the midpoints of all pairs, as in Section 4 of \cite{bai2021inference}. Such a design is also discussed by \cite{athey2017econometrics}. It often increases the MSE relative to its paired version but often improves inference for the ATE, especially with multiple covariates. In particular, simulation evidence in Section \ref{sec:sims} shows that, with multiple covariates, the test with matched sets of four usually has the correct size but the test with matched pairs often severely underrejects. I refer interested readers to Section \ref{sec:asymptotic-mult} of the supplement for a detailed discussion.

A frequent concern in experiments is attrition, meaning units in the baseline survey may drop out in the follow-up survey, so their covariates are available but outcomes are not. I emphasize that even when a unit attrites, the entire pair may not need to be dropped. If attrition happens, then I redefine the difference-in-means estimator using only non-attritors. If attrition is independent of treatment status conditional on the covariates, then this estimator is consistent for the ATE for non-attritors. I refer interested readers to Section \ref{sec:attrition} of the supplement for details. The case with differential attrition, as in the setting of \cite{lee2009training}, is an interesting topic for future work.

Another related question that frequently arises in the design of experiments is that some studies are implemented in multiple waves. Although a full-length discussion of such settings is beyond the scope of the paper, a possible solution is to implement the procedures discussed in this section repeatedly. For instance, in the first wave, researchers could pair the units according to their baseline outcomes. In the second wave, they could use the data from the first wave as pilot data, and implement the pilot-based procedures discussed earlier in this section. They can repeatedly implement the pilot-based procedures in the following waves. In Section \ref{sec:pooled} of the supplement, I discuss how to pool the data from multiple waves for estimation and inference.

\section{Asymptotic Results and Inference} \label{sec:asymptotic}
The optimality result in Section \ref{sec:oracle} pinpoints the optimal stratification but is silent on how the feasible procedures in Section \ref{sec:practical} compare with each other. To make such a comparison, this section studies the asymptotic properties of the difference-in-means estimator. I also provide inference methods for the ATE under different stratifications. The main difficulty in deriving the theoretical results is that under matched-pair designs, treatment status across units is heavily dependent; in fact, treatment status of the two units in a pair is perfectly correlated. I extend the results in \cite{bai2021inference} by allowing units to be paired according to functions of the covariates instead of the covariates themselves, and furthermore allowing the function to be random and dependent on auxiliary data. To begin, I make the following mild moment restriction on the distributions of potential outcomes:

\begin{assumption} \rm \label{as:moments}
$E[Y_i^2(d)] < \infty$ for $d \in \{0, 1\}$.
\end{assumption}

\subsection{Pairing on Nonrandom Functions}
I provide general results when units are paired according to a measurable function $h$ that maps from the support of $X_i$ into $\mathbf R$. The results can be easily specialized to the procedures in Section \ref{sec:practical}. Let $\pi^h \in \Pi_n$ be such that $h_{\pi^h(1)} \leq \dots \leq h_{\pi^h(2n)}$ and define the stratification that pairs units according to $h$ as
\[ \lambda^h(X^{(n)}) = \{\{\pi^h(2s - 1), \pi^h(2s)\}: 1 \leq s \leq n\}~. \]

To describe the requirements on $h$, define $\mathbf H$ to be the set of all measurable functions mapping from the support of $X_i$ into $\mathbf R$ such that the following three conditions hold:
\begin{enumerate}[\rm (a)]
\item $0 < E[\var[Y_i(d) | h(X_i)]]$ for $d \in \{0, 1\}$.
\item $E[Y_i^r(d) | h(X_i) = z]$ is Lipschitz in $z$ for $r = 1, 2$ and $d = 0, 1$.
\item $E[h^2(X_i)] < \infty$.
\end{enumerate}
\noindent (a) is a mild restriction to rule out degenerate situations and to permit the application of suitable laws of large numbers and central limit theorems, and (c) is another mild moment restriction to ensure the pairs are ``close'' in the limit. Some restrictive primitive conditions for (b) are provided in Section \ref{sec:geometry} of the supplement. I assume $h$ lies in the set $\mathbf H$:

\begin{assumption} \rm \label{as:H}
$h \in \mathbf H$.
\end{assumption}

The next theorem establishes the limiting distribution of $\hat \theta_n$ when units are paired according to $h$, where $h$ satisfies Assumption \ref{as:H}.

\begin{theorem} \label{thm:limit}
Suppose the treatment assignment scheme satisfies Assumption \ref{as:half}, the distribution of the data satisfies Assumption \ref{as:moments}, and $h$ satisfies Assumption \ref{as:H}. Then, when units are paired according to $h$, as $n \to \infty$,
\[ \sqrt n (\hat \theta_n - \theta) \stackrel{d}{\to} N(0, \varsigma_h^2)~, \]
where
\begin{equation} \label{eq:limit}
\varsigma_h^2 = \var[Y_i(1)] + \var[Y_i(0)] - \frac{1}{2} E[(E[Y_i(1) + Y_i(0) | h(X_i)] - E[Y_i(1) + Y_i(0)])^2]~.
\end{equation}
\end{theorem}

\begin{remark} \rm \label{remark:loss}
In Section \ref{sec:tau} of the supplement, I show the minimum of $\varsigma_h^2$ over $h \in \mathbf H$ occurs when $h = g$. The law of iterated expectations implies
\[ \varsigma_h^2 = \var[Y_i(1)] + \var[Y_i(0)] - \frac{1}{2} \var[g(X_i)] + \frac{1}{2} E[\var[g(X_i) | h(X_i)]]~, \]
so the increase in the limiting variance when pairing according to $h$ instead of $g$ is proportional to $E[\var[g(X_i) | h(X_i)]]$, the average conditional variance of $g(X_i)$ given $h(X_i)$. Therefore, among all functions $h \in \mathbf H$, choosing an $h$ that minimizes $E[\var[g(X_i) | h(X_i)]]$ is optimal. Intuitively, the optimal $h$ explains the largest proportion of the variation in $Y(1)$ and $Y(0)$.
\end{remark}

\begin{remark} \rm \label{remark:compare}
Theorem \ref{thm:limit} immediately leads to three insights on the comparison of different treatment assignment schemes:
\begin{enumerate}[(a)]
\item Stratifications with a small number of large strata can be characterized by a function $h$ mapping from the support of $X_i$ into $\{1, \dots, S\}$, such that unit $i$ is in stratum $s$ if and only if $h(X_i) = s$. \cite{bugni2018inference} show the limiting variance of $\hat \theta_n$ under such a stratification equals $\varsigma_h^2$. Therefore, $\varsigma_h^2 > \varsigma_g^2$ unless $g(X_i) = E[g(X_i) | h(X_i)]$ with probability one, which means $g(X_i)$ is constant within each stratum.
\item The stratification $\{\{1, \dots, 2n\}\}$ with all units in one stratum can be written as $\lambda^{h_c}(X^{(n)})$, where $h_c$ is a constant function. For any $h$ that satisfies Assumption \ref{as:H}, $\varsigma_{h_c}^2 > \varsigma_h^2$ unless $E[g(X_i) | h(X_i)]$ is constant with probability one. As a result, in terms of the limiting variance of $\hat \theta_n$, any stratification is weakly better than not stratifying at all.
\item It follows from straightforward calculation that for any $h \in \mathbf H$, $\varsigma_h^2$ is weakly less than and typically strictly less than the limiting variance of $\hat \theta_n$ when treatment status is determined by i.i.d.\ coin flips.
\end{enumerate}
Theorem \ref{thm:tabord} in the supplement studies a procedure that ``breaks up'' a stratification with a small number of large strata. I further allow the treated fractions to vary across strata. I show the limiting variance of $\hat \theta_n$ is weakly smaller if I implement small-strata designs similar to the ones described in Remark \ref{remark:tau} separately within each stratum.
\end{remark}

Next, I consider inference for the ATE when units are paired according to $h \in \mathbf H$. For any prespecified $\theta_0 \in \mathbf R$, I am interested in testing
\begin{equation} \label{eq:H0}
H_0: \theta = \theta_0 \text{ versus } H_1: \theta \neq \theta_0
\end{equation}
at level $\alpha \in (0, 1)$. To do so, it suffices to provide a consistent estimator for the limiting variance $\varsigma_h^2$ in \eqref{eq:limit}. To describe such an estimator, for $d \in \{0, 1\}$, define the variance estimator among units with $D = d$ as
\[ \hat \sigma_n^2(d) = \frac{1}{n} \sum_{1 \leq i \leq 2n: D_i = d} (Y_i - \hat \mu_n(d))^2~. \]
In addition, define
\begin{equation} \label{eq:correction}
\hat \rho_n = \frac{2}{n} \sum_{1 \leq j \leq \lfloor\frac{n}{2}\rfloor} (Y_{\pi^h(4j - 3)} + Y_{\pi^h(4j - 2)}) (Y_{\pi^h(4j - 1)} + Y_{\pi^h(4j)})
\end{equation}
and
\begin{equation} \label{eq:se}
\hat \varsigma_{h, n}^2 = \hat \sigma_n^2(1) + \hat \sigma_n^2(0) - \frac{1}{2} \hat \rho_n + \frac{1}{2} (\hat \mu_n(1) + \hat \mu_n(0))^2~.
\end{equation}
The calculation in \ref{sec:nonneg} of the supplement shows $\hat \varsigma_{h, n}^2$ is nonnegative. The correction term $\hat \rho_n$ is constructed by averaging the product of the sum of the outcomes of adjacent \emph{pairs of pairs}, as in \cite{bai2021inference}.

The following theorem shows the variance estimator in \eqref{eq:se} is consistent for the limiting variance in \eqref{eq:limit}.

\begin{theorem} \label{thm:se}
Suppose the treatment assignment scheme satisfies Assumption \ref{as:half}, the distribution of the data satisfies Assumption \ref{as:moments}, and $h$ satisfies Assumption \ref{as:H}. Then, when units are paired according to $h$, as $n \to \infty$, $\hat \varsigma_{h, n}^2$ defined in \eqref{eq:se} satisfies
\[ \hat \varsigma_{h, n}^2 \stackrel{P}{\to} \varsigma_h^2~. \]
\end{theorem}

\begin{remark} \rm \label{remark:se}
The correction term $\hat \rho_n$ in \eqref{eq:correction} is crucial for the consistency of $\hat \varsigma_{h, n}^2$ in \eqref{eq:se}. In commonly-used tests including the two-sample $t$-test \citep{riach2002field,gelman2006data,duflo2007using} and the ``matched pairs'' $t$-test \citep{moses2006matched,hsu2007paired,armitage2008statistical,imbens2015causal,athey2017econometrics}, the test statistics are studentized by variance estimators whose limits in probability are weakly greater than $\varsigma_h^2$, so these tests are asymptotically conservative in the sense that the limiting size is no greater than and typically strictly less than the nominal level. For instance, a 5\%-level test could have a size of 1\%. In fact, the limiting size of the ``matched pairs'' $t$-test is strictly less than the nominal level unless \eqref{eq:hom} holds. I refer interested readers to \cite{bai2021inference} for details.
\end{remark}

\begin{remark} \rm
Let $\tilde h_m$ be a function of the pilot data such that $\tilde h_m \in \mathbf H$ with probability one. Then, the proof of Theorems \ref{thm:limit}--\ref{thm:se} implies the conclusions therein hold for $h = \tilde h_m$ conditional on the pilot data with probability one. Because probabilities are bounded between 0 and 1 and hence are uniformly integrable, the same conclusions hold unconditionally too. In particular, the variance estimator in \eqref{eq:se} is valid even when the sample size of the pilot experiment is small and fixed.
\end{remark}

\subsection{Pairing on Random Functions}
The discussion in the last subsection applies to settings where units are paired according to a fixed function $h \in \mathbf H$ or a random function $\tilde h_m$ such that $\tilde h_m \in \mathbf H$ with probability one. Such settings are most relevant when the pilot sample size $m$ is small. Next, I consider settings where $\tilde h_m$ converges to a fixed function $h \in \mathbf H$ in a suitable sense as $m \to \infty$. Let $Q_X$ denote the marginal distribution of $X_i$.

\begin{assumption} \rm \label{as:l2}
$\tilde h_m$ is a random function depending on the auxiliary data that maps from the support of $X_i$ into $\mathbf R$, and satisfies
\[ \int |\tilde h_m(x) - h(x)|^2 Q_X(d x) \stackrel{P}{\to} 0 \]
as $m \to \infty$.
\end{assumption}

\noindent Assumption \ref{as:l2} is commonly referred to as the $L^2$-consistency of the $\tilde h_m$ for $h$. When the dimension of $X_i$ is fixed and suitable smoothness conditions hold, $L^2$-consistency is satisfied by series and sieves estimators \citep{newey1997convergence,chen2007large} and kernel estimators \citep{li2007nonparametric}. In some high-dimensional settings, when the dimension of $X_i$ increases with $n$ at suitable rates, it is satisfied by the LASSO estimator \citep{buhlmann2011statistics,belloni2014inference}, regression trees and random forests \citep{gyorfi2002distribution-free,wager2015adaptive}, neural nets \citep{white1990connectionist,farrell2018deep}, and support vector machines \citep{steinwart2008support}. The results therein are either exactly as stated in Assumption \ref{as:l2} or one of the following:
\begin{enumerate}[\rm (a)]
\item $\sup\limits_x |\tilde h_m(x) - h(x)| \stackrel{P}{\to} 0$ as $m \to \infty$.
\item $E[|\tilde h_m(x) - h(x)|^2] \to 0$ as $m \to \infty$.
\end{enumerate}
It is straightforward to see (a) implies Assumption \ref{as:l2}. Furthermore, (b) implies Assumption \ref{as:l2} by Markov's inequality.

The next theorem shows that if $\tilde h_m$ is $L^2$-consistent for $h$, then as the sample sizes of both the pilot and main experiments increase, the limiting variance of $\hat \theta_n$ when units are paired according to $\tilde h_m$ is the same as that when units are paired according to $h$.

\begin{theorem} \label{thm:l2}
Suppose the treatment assignment scheme satisfies Assumption \ref{as:half}, the distribution of the data satisfies Assumption \ref{as:moments}, $h$ satisfies Assumption \ref{as:H}, and $\tilde h_m$ satisfies Assumption \ref{as:l2}. Then, when units are paired according to $\tilde h_m$, as $m, n \to \infty$,
\[ \sqrt n (\hat \theta_n - \theta) \stackrel{d}{\to} N(0, \varsigma_h^2) \]
and
\[ \hat \varsigma_{\tilde h_m, n}^2 \stackrel{P}{\to} \varsigma_h^2~. \]
\end{theorem}

\begin{remark} \rm
Note the assumptions in Theorems \ref{thm:limit}--\ref{thm:se} and those in \ref{thm:l2} are non-nested and differ in whether the sample size of the pilot experiment stays fixed or goes to infinity in the asymptotic framework. Theorems \ref{thm:limit} and \ref{thm:se} do not require $\tilde h_m$ to be consistent for any fixed function and allow $m$ to be fixed asymptotically, but require $\tilde h_m \in \mathbf H$ with probability one. On the other hand, Theorem \ref{thm:l2} does not require $\tilde h_m \in \mathbf H$ but requires $m \to \infty$ and $\tilde h_m$ to be $L^2$-consistent for $h$. 
\end{remark}

\subsection{Pairing on Multiple Covariates} \label{sec:asymptotic-mult}
We briefly comment on inference when units are paired according to multiple covariates. For the settings in \eqref{eq:maha} and \eqref{eq:pen-metric}, the variance estimators are slightly more complicated than that in \eqref{eq:se} because the distances in \eqref{eq:maha} and \eqref{eq:pen-metric} cannot be written as distances between two scalars, but the correction term is similar in spirit to \eqref{eq:correction}. I defer the discussion to Section \ref{sec:pen-detail} of the supplement. In addition, note combining data from both the pilot and main experiments for estimation and inference is possible. I defer the discussion to Section \ref{sec:pooled} of the supplement.

When units are paired using multiple covariates, the simulation evidence in Section \ref{sec:sims} shows that when the sample size is not large enough relative to the number of covariates, the size of the test is often strictly smaller than the nominal level. The reason is that the asymptotic results rely on the assumption that units are ``close,'' in the sense that a suitable normalization of the sum of distances between the covariates within each pair is close to zero. When the sample size is not large enough relative to the number of covariates, the procedures here suffer from the curse of dimensionality, so the units paired together are not close enough in terms of their covariates, and hence, the asymptotic results do not approximate the finite-sample distribution of $\hat \theta_n$ very well. The problem is mitigated by matching units into sets of four instead of pairs. Specifically, I first pair the units and then pair the pairs using the midpoints of all pairs, as in Section 4 of \cite{bai2021inference}. In Section \ref{sec:four} of the supplement, I propose a valid test for \eqref{eq:H0} when units are matched into sets of four. Simulation evidence in Section \ref{sec:sims} shows the size of my proposed test is close to the nominal level in finite sample.

\section{Simulation} \label{sec:sims}
In this section, I examine the performance of the practical procedures in Section \ref{sec:practical} and the inference methods in Section \ref{sec:asymptotic} via a simulation study calibrated to a systematically selected set of 10 RCTs from recent issues of the \textit{American Economic Journal: Applied Economics}. I focus on settings with small or no pilots, because they are the most common settings in practice. I searched the 11 issues from October 2018 to April 2021 and collected 28 papers running RCTs. I exclude 11 papers for which the treatment is assigned at the cluster level instead of the unit level. I further exclude four papers with a network/spillover structure. I also exclude one paper for which the sample size is too small (less than 20). Finally, I exclude two papers for which the data is confidential. I end up with 10 papers, which are listed in Table \ref{table:papers}. For each paper, I list whether the baseline outcome is available, the original randomization method, and the number of additional covariates besides the baseline outcome in the main regression specification of the paper. The full details of the data are available in Section \ref{sec:sims-supp} of the supplement.

\begin{table}[ht]
\caption{Papers selected for the simulation study}
\begin{adjustbox}{max width=0.8\linewidth,center}
\begin{tabular}{llll}
\hline\hline
Paper & Baseline & Original Stratification & \# of Covariates\\ 
\hline
1.\ \cite{herskowitz2021gambling} & $\times$ & none & 4 \\
\addlinespace
2.\ \cite{lee2021poverty} & \checkmark & re-randomization & 4 \\
\addlinespace
3.\ \cite{abel2020value} & \checkmark & gender & 7\\
\addlinespace
4.\ \cite{gerber2020one} & \checkmark & state & 11\\
\addlinespace
5.\ \cite{deserranno2019leader} & \checkmark & none & 1\\
\addlinespace
6.\ \cite{barrera-osorio2019medium-} & $\times$ & baseline grade, gender & 16\\
\addlinespace
7.\ \cite{himmler2019soft} & $\times$ & GPA (4 strata)& 8\\
\addlinespace
8.\ \cite{abel2019bridging} & \checkmark & none & 10\\
\addlinespace
9.\ \cite{de_mel2019labor} & \checkmark & region, sector & 7\\
\addlinespace
10.\ \cite{lafortune2018role} & \checkmark & none & 8\\
\hline\hline
\end{tabular}
\end{adjustbox}
\begin{tablenotes}
\item For each paper, I list whether the baseline outcome is available, the original stratification method, and the number of covariates besides the baseline outcome in the main regression specification of the paper. \cite{lee2021poverty} assign treatment status by re-randomization, and the other papers use stratified randomization (possibly with only one stratum).
\end{tablenotes}
\label{table:papers}
\end{table}

For each paper, I denote the sample size by $2n$. I use the original sample except for \cite{barrera-osorio2019medium-}, where the original data contains 15,759 observations and 16 covariates, so one replication in the simulation takes almost six hours. For \cite{barrera-osorio2019medium-} only, I take half of the observations as the population to reduce the computational time for one replication to about an hour, which is about the same as that using the next largest dataset. I begin by imputing the unobserved potential outcomes. For the $i$th unit, I denote the original data by $(Y_i^\ast, D_i^\ast, X_{1i}^\ast, X_{2i}^\ast)$, where $Y_i^\ast$ denotes its observed outcome, $D_i^\ast$ denotes its treatment status, $X_{1i}^\ast$ denotes its baseline outcome if available, and $X_{2i}^\ast$ denotes the other covariates in the main regression specification of the paper. Let $Y_i^\ast(1), Y_i^\ast(0)$ denote the potential outcomes for the $i$th unit. For the $i$th unit, $Y_i^\ast(D_i)$ is observed, and I construct $Y_i^\ast(1 - D_i)$ according to the following models:
\begin{itemize}
	\item[] Model 1: $Y_i^\ast(1) = Y_i^\ast(0)$.
	\item[] Model 2: $Y_i^\ast(1 - D_i) = Y_{j(i)}^\ast$, where the $j(i)$-th unit is the closest unit to the $i$th unit in terms of the Mahalanobis distance of $(X_1^\ast, X_2^\ast)$ among units with $D_j \neq D_i$.
	\item[] Model 3: $Y_i^\ast(1 - D_i) = Y_{j(i)}^\ast$, where the $j(i)$-th unit is the closest unit to the $i$th unit in terms of the baseline outcome $X_1^\ast$, if available, among units with $D_j \neq D_i$.
\end{itemize}
In Model 1, treatment effects are homogeneous, so \eqref{eq:hom} holds. In Models 2 and 3, treatment effects are heterogeneous. Note the baseline outcome predicts the potential outcomes better in Model 3 than in Model 2.

For each replication, I simulate new data $((Y_i(1), Y_i(0), X_{i1}, X_{2i}): 1 \leq i \leq 2n)$ by drawing $2n$ units from the empirical distribution of $((Y_i^\ast(1), Y_i^\ast(0), X_{1i}^\ast, X_{2i}^\ast): 1 \leq i \leq 2n)$, so each unit in the original data is drawn with equal probability, with replacement.

For each paper, I implement several stratifications. Note the covariates may have been selected ex-post by authors on the basis of predictive power, while ideally I would like to include only the covariates specified in the pre-analysis plans. Unfortunately, only one paper includes such information in the AEA RCT Registry, and the pre-registered covariates are the same as those in the regression analysis. I stratify on the baseline outcome whenever it is available. To stratify based on data from pilot experiments, I reached out to the authors of all 10 papers to request pilot data. 9 out of 10 replied, 8 of whom said they did not run a pilot, and the last said they ran a pilot, but the data was lost. Therefore, I simulate pilot data by drawing with replacement from the empirical distribution at a sample size of $\lfloor 0.2 \cdot (2n) \rfloor$. To study the setting with a small and fixed pilot, I fix the pilot data throughout all replications. I also consider several stratifications with matched sets of four, as in \cite{athey2017econometrics}. Specifically, strata are constructed by first pairing the units and then pairing the pairs according to their midpoints, as in Section 4 of \cite{bai2021inference}. The complete list of stratifications are:
\begin{enumerate}[(a)]
\item MP X: Matched pairs to minimize the sum of the squared Mahalanobis distances in \eqref{eq:maha} of all covariates $X$.
\item MS X: Matched sets of four to minimize the sum of the squared Mahalanobis distances of $X$.
\item MP base: Matched pairs according to the baseline outcome, if available.
\item MS base: Matched sets of four according to the baseline outcome, if available.
\item MP X2: Matched pairs to minimize the sum of the squared Mahalanobis distances of $X_2$, namely, all covariates in the main regression specification except the baseline outcome.
\item MP pilot: Matched pairs according to $\tilde g_m$ from the pilot, where $\tilde g_m$ is given by the OLS.
\item MP pen: The penalized matched pairs given by minimizing the sum of the distances in \eqref{eq:pen-metric} of all covariates.
\item Origin: Stratification used in the original paper, if not one of (a) through (g).
\item None: No stratification, meaning all units are in one stratum and exactly half are treated.
\item[(i')] None-reg: No stratification with the estimator given by the OLS estimator of the coefficient on $D$ in the linear regression of $Y$ on a constant, $D$, and $X$.
\end{enumerate}

The original stratifications are listed in Table \ref{table:papers}. I do not consider re-randomization in \cite{lee2021poverty}, because the exact implementation is unclear from the original paper, and inference under re-randomization is complicated. See also Remark \ref{remark:rerand} for a comparison between re-randomization and matched-pair designs. Section \ref{sec:sims-supp} of the supplement contains the results for several additional stratifications. Although it is interesting to investigate the performance of regression adjustment with the stratifications in Origin, inference with regression adjustment under stratified randomization is still an open question.

I consider the following inference methods:
\begin{enumerate}
	\item Matched pairs: (1) (adj) the adjusted $t$-test with the variance estimator in \eqref{eq:se}; (2) (MPt) the test with the variance estimator in Theorem 10.1 of \cite{imbens2015causal}, which is equivalent to the ``matched pairs'' $t$-test in \cite{bai2021inference}.
	\item Matched sets of four: (adj4) the adjusted $t$-test with the variance estimator in \eqref{eq:four-se} in Section \ref{sec:sims-supp} of the supplement;
	\item Original: the test in (23) of \cite{bugni2018inference}, which is asymptotically exact under stratified randomization.
	\item No stratification: without regression adjustment, the two-sample $t$-test with the variance estimator given by $\hat \sigma_n^2(1) + \hat \sigma_n^2(0)$; with regression adjustment, White's heteroskedasticity-robust standard error.
\end{enumerate}
For matched sets of four, \cite{athey2017econometrics} propose a test in a sampling framework different from ours. I show in Section \ref{sec:four} of the supplement that because of the differences in sampling frameworks, the test in \cite{athey2017econometrics} does not control size in my setting unless the conditional ATEs are homogeneous. Therefore, I defer these simulation results to Section \ref{sec:sims-supp} of the supplement.

For each paper, each model, and each stratification, across $1,000$ replications, I calculate three metrics of performance: (1) the MSE of estimating $\theta$ using $\hat \theta_n$, reflecting the precision of the estimator; (2) the average rejection probability of testing \eqref{eq:H0} for $\theta_0 = \theta$, reflecting the size of the test; and (3) the average standard error, which directly determines the length of the confidence interval for the ATE.

The main results of the simulation study are summarized in Table \ref{table:sims-main}. I only report the summary statistics across all papers and models and defer the raw numbers to Section \ref{sec:sims-supp} of the supplement. In particular, for each stratification, I report the average and $[\min, \max]$ across all papers and models of
\begin{enumerate}
	\item the ratio between the MSE under the particular stratification and the MSE under no stratification,
	\item the size of the test, and
	\item the ratio between the average standard error under the particular stratification and the average standard error under no stratification.
\end{enumerate}
Rows are labeled according to the stratifications.

\begin{table}[ht!]
\caption{Summary statistics for MSEs, size, and standard errors for each stratification across all papers and models}
\begin{adjustbox}{max width=\linewidth,center}
\begin{tabular}{ccccccc}
\hline\hline
& Stratification & MSE (ratio vs.\ None) & \multicolumn{2}{c}{size (\%)} & \multicolumn{2}{c}{s.e. (ratio vs.\ None)} \\ 
\cmidrule(lr){4-5} \cmidrule(lr){6-7} 
& & & adj/adj4 & MPt & adj/adj4 & MPt \\
\midrule
(a) &  MP X     & 0.549 & 2.267 & 3.533 & 0.870 & 0.810 \\ 
& & [0.304, 0.830] & [0.300, 4.800] & [1.700, 6.400] & [0.720, 0.968] & [0.530, 1.092] \\ 
\addlinespace
(b) &  MS X     & 0.695 & 5.148 & - & 0.828 & - \\
& & [0.464, 0.927] & [3.600, 6.600] & - & [0.663, 0.946] & - \\          
\addlinespace
(c) &  MP base  & 0.762 & 4.771 & 4.781 & 0.886 & 0.885 \\
& & [0.404, 1.030] & [3.200, 5.800] & [2.800, 6.200] & [0.629, 0.989] & [0.633, 1.003] \\
\addlinespace
(d) &  MS base & 0.792 & 5.257 & - & 0.882 & -    \\
& & [0.404, 0.982] & [4.400, 6.300] & - & [0.629, 0.991] & - \\         
\addlinespace
(e) &  MP X2    & 0.685 & 3.126 & 3.852 & 0.919 & 0.874 \\
& & [0.362, 0.923] & [0.500, 6.900] & [1.300, 6.300] & [0.840, 0.979] & [0.608, 1.091] \\
\addlinespace
(f) &  MP pilot & 0.666 & 3.578 & 4.107 & 0.879 & 0.855 \\
& & [0.387, 0.873] & [2.100, 5.500] & [2.700, 5.600] & [0.653, 0.969] & [0.605, 1.063] \\
\addlinespace
(g) &  MP pen   & 0.542 & 2.296 & 3.330 & 0.862 & 0.806 \\
& & [0.280, 0.826] & [0.400, 4.800] & [1.300, 5.600] & [0.625, 0.965] & [0.506, 1.093] \\
\addlinespace
(h) &  Origin   & 1.007 & 5.292 &  - & 0.980 &  -    \\
& & [0.918, 1.114] & [3.700, 7.300] & - & [0.950, 0.999] & - \\      
\addlinespace
(i) &  None     & 1.000 & 5.089 & - & 1.000 & -   \\
& (benchmark) & [1.000, 1.000] & [3.600, 6.900] & - & [1.000, 1.000] & - \\
\addlinespace
(i') & None-reg & 0.948 & 4.900 & - & 0.980 & - \\ 
& & [0.775, 1.012] & [3.200, 6.700] & - & [0.880, 1.034] & - \\ 
\hline\hline
\end{tabular}
\end{adjustbox}
\label{table:sims-main}
\begin{tablenotes}
\item For each stratification, I report summary statistics across all papers and models of (1) the ratio between the MSE under the particular stratification and the MSE under no stratification, (2) the size of testing \eqref{eq:H0} for $\theta_0 = \theta$ at the 5\% level, in percentage, and (3) the ratio between the average standard error under the particular stratification and the average standard error under no stratification. The tests used in this table are as follows: for matched-pair designs, the adjusted $t$-test with the variance estimator in \eqref{eq:se} (adj) and the test in \cite{imbens2015causal} (MPt); for matched sets of four, the adjusted $t$-test with the variance estimator in \eqref{eq:four-se} in Section \ref{sec:four} of the supplement (adj4); for the original stratifications, the test in (23) of \cite{bugni2018inference}; for no stratification, the two-sample $t$-test; for the regression-adjusted estimator, the $t$-test with White's heteroskedasticity-robust standard error. For each metric, I show the average and $[\min, \max]$ across all papers and models. Rows are labeled according to the stratifications. Columns are labeled according to the metrics. For size and standard errors, the second column corresponds to MPt for matched-pair designs and the first column corresponds to the other tests. The definitions of the stratifications can be found in the main text.
\end{tablenotes}
\end{table}

\subsection{Statistical Precision} \label{sec:sims-main}
In this subsection, I discuss major takeaways about statistical precision (specifically, MSE) from Table \ref{table:sims-main}. I focus on five questions that are particularly relevant to empirical practice.

First, how much statistical precision are researchers leaving on the table with their current stratification methods? To answer this question, I compare the MSEs under the stratifications used in the original paper (Origin) and the MSEs when pairing according to the baseline outcome (MP base). Relative to the original stratifications used in those 10 papers, if the researchers had just paired the units according to their baseline outcomes, the MSE would be 24\% smaller on average and 56\% smaller in some cases. In fact, in many models, the MSE under original stratification is almost the same as the MSE when not stratifying. As a result, if researchers had paired the units according to their baseline outcomes, they could have reached the same statistical precision with a much smaller sample size.

Second, how much statistical precision would researchers leave on the table by pairing units according to the baseline outcome rather than using more complicated feasible methods? Note MP X usually has the smallest MSEs across all methods. On average, the MSE under MP base is about 39\% larger than that under MP X, and 10\% larger than that under MS X. As a result, researchers indeed sacrifice some statistical precision by pairing units according to the baseline outcome along instead of pairing or matching into sets of four according to all covariates. Note, however, that MP base picks up about half of the difference between the MSEs under the best feasible method (MP X) and the status-quo methods (Origin).

Third, what is the value of collecting just the baseline outcome rather than other covariates? To answer this question, I compare the MSEs under MP base and MP X2. Note the MSE under MP base is on average only 11\% larger than and sometimes almost the same as that under MP X2. As a result, the statistical precision of pairing according to the baseline outcome alone is comparable to the statistical precision of pairing according to all other covariates. Note the number of other covariates is close to or larger than 10 in most cases, so the per-covariate return for collecting all of them is limited relative to collecting the baseline outcome alone.

Fourth, are matched sets of four better than matched pairs in terms of the MSE? To answer this question, I compare the MSEs of the MS methods and the MP methods. The MSE under MS base is 4\% larger than that under MP base, and the MSE under MS X is 27\% larger than that under MP X. Therefore, the statistical precision is higher with matched pairs. The difference is pronounced when I match according to multiple covariates but tiny when I match only according to the baseline outcome.

Fifth, does the best pairing method based on a small pilot dominate pairing according to $X$? In other words, what is the value of having a small pilot? First, note the MSE under MP pen is 19\% smaller than that under MP pilot, so the penalized matched-pair design indeed has better precision than the naïve plug-in procedure. Meanwhile, the MSE under MP pen is almost the same as that under MP X, and even in the most favorable case, it is only about 8\% smaller. Therefore, the return for a small pilot in terms of the MSE is negligible.

I also study the performance of regression-adjusted estimators in stratifications with one stratum (None-reg). With one stratum, the regression-adjusted estimator usually has slightly smaller MSEs than the difference-in-means estimator. In almost all cases, however, the MSE is larger than those under all methods with matched pairs or matched sets of four, regardless of whether all or only a subset of the covariates in the regression adjustment are used in the matching. Section \ref{sec:sims-supp} of the supplement contains the results for the regression-adjusted estimator in \cite{lin2013agnostic}, which additionally includes the interactions of treatment status and covariates. The results are qualitatively similar to those for None-reg. Therefore, most of the gains in precision from matching according to the covariates cannot be retrieved by controlling for the same covariates via ex-post regression adjustment.


\subsection{Inference Methods} \label{sec:sims-inf}
Next, I discuss the properties of the inference methods. I start with the size of the tests. For matched pairs, note both the adjusted $t$-test and the ``matched pairs'' $t$-test control size well across all papers and models. When I pair units according to the baseline outcome (MP base), the size of the adjusted $t$-test is almost always close to 5\%. The size of the ``matched pairs'' $t$-test is also close to 5\%. The underrejection phenomenon for the ``matched pairs'' $t$-test in Remark \ref{remark:se} is very mild, reflecting the treatment effects heterogeneity is not very large. Several relatively noticeable cases include: for Model 2 of paper 5, the size of the ``matched pairs'' $t$-test is 3.9\% but the size of the adjusted $t$-test is 5.4\%; for Model 3 of paper 5, the size of the ``matched pairs'' $t$-test is 2.8\% but the size of the adjusted $t$-test is 3.2\%.

When I pair units according to multiple covariates (MP X, MP X2, and MP pen), the ``matched pairs'' $t$-test is still conservative except in Model 1, for the same reason as mentioned in Remark \ref{remark:se}. At the same time, the adjusted $t$-test also becomes conservative---its size is often smaller than 5\%. The reason is that the asymptotic results in this paper rely on the assumption that units are ``close,'' in the sense that a suitable normalization of the sum of the distances between the covariates within each pair is close to zero. When pairing according to multiple covariates, however, the procedures suffer from the curse of dimensionality, so the units paired together are not close enough in terms of their covariates. Therefore, the asymptotic results do not approximate the finite-sample distribution of $\hat \theta_n$ very well, and my variance estimator does not approximate the actual variance of $\hat \theta_n$ very well.

When matching according to multiple covariates, the conservativeness of the tests is somewhat alleviated by matching units into sets of four instead of pairs. The size of the test under MS X is close to 5\% even when the size of the test under MP X is much smaller than 5\%. On the other hand, such a difference is virtually nonexistent when I match only according to the baseline outcome---the size of the test under MP base and that under MS base are both close to 5\%. Our current asymptotic framework cannot explain the difference in size because the variance estimators for both MP and MS methods are consistent for the limiting variance. The exact reason for the difference is an interesting topic for future work.

I now turn to the standard errors. The findings are mostly similar to those for the MSEs, though with some important exceptions. The standard error under MP base is 10\% smaller on average and 34\% smaller in some cases than that under Origin. At the same time, although the MSE under MP X is smaller than that under MS X, the standard error of MP X is larger than that under MS X. In fact, the standard error under MP X is about the same as that under MP base, and the standard error under MP X2 is often larger than that under MP base. Therefore, although pairing according to multiple covariates is desirable for the MSE, it is often not the best choice for inference, because the standard error is too large and the size of the test could be strictly smaller than the nominal level. By matching units into sets of four instead of pairs according to the same set of covariates, researchers could lower the standard error and bring the size close to the nominal level. Note, however, that the MSE will increase, as discussed in Section \ref{sec:sims-main}.

I emphasize the validity of my tests relies on the assumptions on the sampling framework. In this paper, I assume units are drawn from a superpopulation, and the potential outcomes and the covariates are random. \cite{de_chaisemartin2021at}, on the other hand, study a finite-population setting in which the potential outcomes and the covariates are fixed. Such a setting is particularly relevant if we have a convenience sample instead of a random sample drawn from a large population. \cite{de_chaisemartin2021at} show in these settings that if the number of pairs is small, then the tests in my paper and \cite{bai2021inference} may not control size, and the ``matched pairs'' $t$-test in \cite{imbens2015causal} could become preferable.

\subsection{Multiple Outcomes} \label{sec:sims-mult}
Finally, I consider settings with multiple outcomes. I take the example of \cite{abel2019bridging}, where a primary outcome, a secondary outcome, and the baseline outcomes of both are available. The paper studies job-searching behaviors. The primary outcome is the search hours and the secondary outcome is the number of applications sent. I study the estimation of the ATE of the secondary outcome. The missing potential outcomes are imputed as in Model 1, assuming the treatment effect is zero for everyone, and Model 3, using the nearest neighbor in terms of the baseline value of the secondary outcome. I consider the following stratifications:
\begin{enumerate}
	\item[] MP 2: Matched pairs according to the baseline value of the secondary outcome.
	\item[] MS 2: Matched sets of four according to the baseline value of the secondary outcome.
	\item[] MP 1: Matched pairs according to the baseline value of the primary outcome.
	\item[] MS 1: Matched sets of four according to the baseline value of the primary outcome.
	\item[] MP 1+2: Matched pairs to minimize the sum of the squared Mahalanobis distances in \eqref{eq:maha} of the baseline values of both outcomes.
	\item[] MS 1+2: Matched sets of four to minimize the sum of the squared Mahalanobis distances of the baseline values of both outcomes.
	\item[] None: No stratification, meaning all units are in one stratum and exactly half are treated.
\end{enumerate}

\begin{table}[ht!]
\caption{MSEs, size, and standard errors for estimating the ATE of the secondary outcome in \cite{abel2019bridging}}
\begin{adjustbox}{max width=\linewidth,center}
\begin{tabular}{cccccccc}
\hline \hline
& \multicolumn{3}{c}{Model 1} & \multicolumn{3}{c}{Model 3} \\
& \multicolumn{3}{c}{$\theta = 0$} & \multicolumn{3}{c}{$\theta = 0.4449$} \\
\cmidrule(lr){2-4} \cmidrule(lr){5-7}
& MSE & size & s.e. & MSE & size & s.e. \\
& (ratio vs.\ None) & (\%) & (ratio vs.\ None) & (ratio vs.\ None) & (\%) & (ratio vs.\ None) \\
\midrule
MP 2 & 0.760 & 5.9 &0.835& 0.645 & 4.7&0.799 \\
\addlinespace
MS 2 & 0.756 & 6.0 &0.835& 0.689 & 5.8&0.799 \\
\addlinespace
MP 1 & 0.988 & 4.3 &0.980& 1.010 & 4.4&0.986 \\
\addlinespace
MS 1 & 1.117 & 6.6 &0.980& 1.070 & 5.7&0.987 \\
\addlinespace
MP 1+2 & 0.558 &4.2 &0.769 & 0.568&4.0 &0.783 \\
\addlinespace
MS 1+2 & 0.615&4.9 &0.760 & 0.598& 4.3&0.777 \\
\addlinespace
None & 1.000 & 4.5 &1.000& 1.000 & 4.7&1.000 \\
\hline\hline
\end{tabular}
\end{adjustbox}
\begin{tablenotes}
\item For each stratification, I report (1) the MSE, (2) the size of testing \eqref{eq:H0} for $\theta_0 = \theta$ at the 5\% level, in percentage, and (3) the average standard error. The parameter of interest is the ATE of the secondary outcome. The tests used in this table are as follows: for matched-pair designs, the adjusted $t$-test with the variance estimator in \eqref{eq:se} (adj); for matched sets of four, the adjusted $t$-test with the variance estimator in \eqref{eq:four-se} of the supplement (adj4). Rows are labeled according to stratifications. Columns are labeled according to the models and metrics. I display the value of $\theta$ for each model. The definitions of the stratifications can be found in the main text.
\end{tablenotes}
\label{table:multi}
\end{table}

In light of the results in Section \ref{sec:sims-main}, I only consider the adjusted $t$-tests. For each stratification, I calculate the MSE, the size of testing \eqref{eq:H0} with $\theta_0 = \theta$, and the average standard error. The results are displayed in Table \ref{table:multi}. Rows are labeled according to the stratifications. As expected, because the secondary outcome is of interest, for both models, stratifying on the baseline value of the secondary outcome produces smaller MSEs than stratifying on the baseline value of the primary outcome. For both models, the MSEs when stratifying on the baseline value of the primary outcome are close to the MSEs with one stratum, reflecting the baseline value of the primary outcome is a poor predictor of the secondary outcome. The smallest MSE is attained by pairing according to the baseline values of both outcomes, and the second smallest MSE is attained by matching units into sets of four according to both baseline outcomes. In all models, the size of the test is close to the nominal level. The test under MP 1+2 slightly underrejects for the same reason as in Section \ref{sec:sims-inf}, although the problem is mild because I only match on two variables. The standard errors are ranked in the same way as the MSEs except for that of MP 1+2. In both models, MS 1+2 produces the smallest standard errors across all methods.

\section{Discussion and Recommendations for Empirical Practice} \label{sec:conclusion}
Based on the theoretical results, in settings with large pilots, I recommend researchers to pair units according to the estimated index function from nonparametric regressions. If the conditional ATEs are homogeneous and researchers have access to a large observational dataset from the same population as that of the main experiment, then I recommend pairing according to predicted values from nonparametric regressions of the outcome on the covariates in the observational dataset. In what follows, I focus on settings with small or no pilots because these are the most common settings in practice.

A simple approach that researchers can take, assuming there is only one primary outcome of interest and its baseline value is available, is to pair units according to the baseline outcome. Indeed, if the baseline outcome is the only available covariate and the index function is monotonic in it, then my theoretical results show pairing units according to the baseline outcome is optimal at any sample size. The simulation results in Section \ref{sec:sims} also show pairing according to the baseline outcome improves upon the status-quo methods in terms of both the MSE and the standard error of the difference-in-means estimator. Further, this approach has the advantage of simplicity.

When multiple covariates are available, an attractive alternative is to match units into pairs or sets of four according to the baseline outcome and other covariates. Unless the number of covariates is very small, I recommend matched sets of four over pairs because the standard error is usually smaller with matched sets of four. In my simulation study, when I use the baseline outcome together with all the covariates that the authors control for in their regressions, matching units into sets of four leads to smaller MSEs and standard errors than pairing on the baseline outcome alone. Note that the good performance of this design could be due to the fact that the authors selected the covariates with the best predictive power ex-post, something that is not feasible at the time of randomization. Nevertheless, forming sets of four according to the baseline outcome and other covariates is an attractive alternative to pairing according to the baseline outcome alone.

In my simulation study, when multiple outcomes are of interest, pairing on one of them may not improve the MSE of the other outcomes. In those settings, researchers could consider matching units into sets of four to minimize the sum of the squared Mahalanobis distances of the baseline values of all outcomes of interest and perhaps some additional covariates.

A further question is whether pilot experiments are worth running for the sole purpose of improving the precision for estimating the ATE. My simulation results only show minor gains in statistical precision when using pilot-based stratifications instead of matching directly on the covariates. Therefore, although pilot experiments are essential for other aspects of the design of the main experiment, they are not as helpful in improving the precision of the estimator for the ATE.

Another natural question is whether one can retrieve the gains in precision from matched pairs or sets of four units by controlling for the same set of covariates through ex-post regression adjustment. Our simulation results show the answer is negative, although with one stratum, regression adjustment slightly lowers the MSE and the standard error relative to the unadjusted difference-in-means estimator. I further note the difference-in-means is unbiased for the ATE in finite sample under all stratifications considered in this paper, while the regression-adjusted estimators are only consistent for the ATE asymptotically. A very interesting direction for future work is to combine regression adjustment with stratifications defined by matched pairs or matched sets of four units.

For inference, researchers can use the test with the variance estimator in \eqref{eq:se}. They can also use the test in Theorem 10.1 of \cite{imbens2015causal}, which is valid albeit sometimes conservative. Finally, I emphasize that my framework assumes units are drawn from a superpopulation and the potential outcomes and the covariates are random. If we have a convenience sample instead of a random sample drawn from a large population, and the sample size is small, then \cite{de_chaisemartin2021at} show the tests in this paper and \cite{bai2021inference} may not control size, and the test in \cite{imbens2015causal} could become preferable.

\clearpage
\bibliography{optdesign}

\newpage
\appendix
\newgeometry{hmargin=1.25in,vmargin=1.25in}
\renewcommand{\theequation}{S.\arabic{equation}}
\renewcommand{\thefigure}{S.\arabic{figure}}
\renewcommand{\thetable}{S.\arabic{table}}

\section{Proofs of Main Results} \label{sec:proof}
In the appendix, $Q$ denotes the distribution of $((Y_i(1), Y_i(0), X_i))'$. I denote the observed quantities by $W_i = (Y_i, X_i', D_i)'$ and the pilot data by $\tilde W^{(m)} = ((\tilde Y_j, \tilde X_j', \tilde D_j)': 1 \leq j \leq m)$. $\mathrm{dim}(X_i)$ denotes the dimension of $X_i$ and $\mathrm{supp}(X_i)$ denotes the support of $X_i$. I use $a \lesssim b$ to denote there exists $c \geq 0$ such that $a \leq c b$.

\subsection{Proof of Lemma \ref{lem:mixing}}
Let $\lambda = (\lambda_1, \ldots, \lambda_S)$ be a stratification and recall $n_s = |\lambda_s|$. Let $(d_1, \ldots, d_{2n})$ be a vector of values that $(D_1, \ldots, D_{2n})$ may take under $\lambda$, and for every $s$ let $(d_1^s, \ldots, d_{n_s}^s)$ denote treatment status of the units in stratum $s$. For every $s$, there are $(n_s / 2)!$ matched-pair designs in stratum $s$ that could lead to $(d_1^s, \ldots, d_{n_s}^s)$. For each of such designs, $(d_1^s, \ldots, d_{n_s}^s)$ is realized with probability $2^{-n_s / 2}$. Accordingly, if instead of implementing $\lambda$, I implement a matched-pair design in each stratum, uniformly across all matched-pair designs within each stratum, and independently across strata, the probability that $(d_1, \ldots, d_{2n})$ is realized is
\[ \prod_{1 \leq s \leq S} \frac{1}{\binom{n_s}{n_s / 2}(n_s / 2)! / 2^{n_s/2}} (n_s / 2)! \frac{1}{2^{n_s}} = \prod_{1 \leq s \leq S} \frac{1}{\binom{n_s}{n_s / 2}}~, \]
which is the probability that $(d_1, \ldots, d_{2n})$ is realized under $\lambda$. To see the number of matched-pair designs in a stratum with $n_s$ units is
\[ \binom{n_s}{n_s / 2}(n_s / 2)! / 2^{n_s/2}~, \]
consider the following thought experiment: First, choose $n_s / 2$ units and fix their positions; next, permute the rest $n_s / 2$ units and match them to the fixed positions, and note each permutation leads to a matched-pair design; finally, note I have counted each matched-pair design repeatedly, and precisely $2^{n_s / 2}$ times, because I could flip the positions of the two units within each pair. \qed

\subsection{Proof of Theorem \ref{thm:limit}}
Follows immediately from Lemma \ref{lem:tau-limit} with $\tau = \frac{1}{2}$. Note condition (c) on $h$ in Theorem \ref{thm:tau-min} is satisfied because of Lemma \ref{lem:close}. \qed

\subsection{Proof of Theorem \ref{thm:se}}
To begin with, note $\hat \mu_n(d) \stackrel{P}{\to} E[Y_i(d)]$ and $\hat \sigma_n^2(d) \stackrel{P}{\to} \var[Y_i(d)]$ for $d \in \{0, 1\}$, by Lemma 6.5 in \cite{bai2021inference}. Next, I show
\begin{equation} \label{eq:rhohat}
E[\hat \rho_n | h^{(n)}] \stackrel{P}{\to} E[(E[Y_i(1) + Y_i(0) | h(X_i)])^2]~.
\end{equation}
For convenience, I define $\mu_d(h_i) = E[Y_i(d) | h(X_i) = h_i]$ for $d \in \{0, 1\}$ and $g_h(h_i) = \mu_1(h_i) + \mu_0(h_i)$. To see this, note
\begin{align*}
& E[(Y_{\pi^h(4j - 3)} + Y_{\pi^h(4j - 2)}) (Y_{\pi^h(4j - 1)} + Y_{\pi^h(4j)}) | h^{(n)}] \\
& = \frac{1}{4} (\mu_1(h_{\pi^h(4j - 3)}) + \mu_0(h_{\pi^h(4j - 2)}))(\mu_1(h_{\pi^h(4j - 1)}) + \mu_0(h_{\pi^h(4j)})) \\
& + \frac{1}{4} (\mu_1(h_{\pi^h(4j - 3)}) + \mu_0(h_{\pi^h(4j - 2)}))(\mu_1(h_{\pi^h(4j)}) + \mu_0(h_{\pi^h(4j - 1)})) \\
& + \frac{1}{4} (\mu_1(h_{\pi^h(4j - 2)}) + \mu_0(h_{\pi^h(4j - 3)}))(\mu_1(h_{\pi^h(4j - 1)}) + \mu_0(h_{\pi^h(4j)})) \\
& + \frac{1}{4} (\mu_1(h_{\pi^h(4j - 2)}) + \mu_0(h_{\pi^h(4j - 3)}))(\mu_1(h_{\pi^h(4j)}) + \mu_0(h_{\pi^h(4j - 1)})) \\
& = \frac{1}{4} (g_h(h_{\pi^h(4j - 3)}) + g_h(h_{\pi^h(4j - 2)}))(g_h(h_{\pi^h(4j - 1)}) + g_h(h_{\pi^h(4j)})) \\
& = \frac{1}{4} \sum_{k \in \{2, 3\}, l \in \{0, 1\}} g_h^2(h_{\pi^h(4j - k)}) + g_h^2(h_{\pi^h(4j - l)}) - (g_h(h_{\pi^h(4j - k)}) - g_h(h_{\pi^h(4j - l)}))^2~.
\end{align*}
As a result,
\begin{align*}
E[\hat \rho_n | h^{(n)}] & = \sum_{1 \leq j \leq \lfloor\frac{n}{2}\rfloor} E[(Y_{\pi^h(4j - 3)} + Y_{\pi^h(4j - 2)}) (Y_{\pi^h(4j - 1)} + Y_{\pi^h(4j)}) | h^{(n)}] \\
& = \frac{1}{2n} \sum_{1 \leq i \leq 2n} g_h^2(h(X_i)) - \frac{1}{4n} \sum_{1 \leq j \leq \lfloor\frac{n}{2}\rfloor} \sum_{k \in \{2, 3\}, l \in \{0, 1\}} (g_h(h_{\pi^h(4j - k)}) - g_h(h_{\pi^h(4j - l)}))^2~.
\end{align*}
By the assumption that $E[h^2(X_i)] < \infty$, \eqref{eq:4} holds. \eqref{eq:rhohat} then follows from the assumption that $E[Y_i^r(d) | h(X_i) = z]$ is Lipschitz in $z$ for $r = 1, 2$ and $d = 0, 1$, \eqref{eq:4}, the fact that
\begin{multline} \nonumber
E[g_h^2(h(X_i))]  \lesssim E[E[Y_i(1) | h(X_i)]^2] + E[E[Y_i(0) | h(X_i)]^2] \\
\leq E[E[Y_i^2(1) | h(X_i)]] + E[E[Y_i^2(0) | h(X_i)]] = E[Y_i^2(1) + Y_i^2(0)] < \infty
\end{multline}
because of Assumption \ref{as:moments}, and the weak law of large numbers.

It remains to show $\hat \rho_n - E[\hat \rho_n | h^{(n)}] \stackrel{P}{\to} 0$. I will prove
\[ \frac{2}{n} \sum_{1 \leq j \leq \lfloor\frac{n}{2}\rfloor} (Y_{\pi^h(4j - 2)} Y_{\pi^h(4j)} - E[Y_{\pi^h(4j - 2)} Y_{\pi^h(4j)} | h^{(n)}]) \stackrel{P}{\to} 0~, \]
and the others follow similarly. I will repeatedly use the following elementary inequalities for any $a, b \in \mathbf R$ and $\lambda > 0$:
\begin{align*}
|a + b| I \{|a + b| > \lambda\} & \leq 2|a| I \{|a| > \lambda / 2\} + 2 |b| I \{|b| > \lambda / 2\} \\
|ab| I \{|ab| > \lambda\} & \leq |a|^2 I \{|a| > \sqrt \lambda\} + |b|^2 I \{|b| > \sqrt \lambda\}~.
\end{align*}
To begin with,
\[ E[Y_{\pi^h(4j - 2)} Y_{\pi^h(4j)} | h^{(n)}] = \frac{1}{2} \mu_1(h_{\pi^h(4j - 2)}) \mu_0(h_{\pi^h(4j)}) + \frac{1}{2} \mu_1(h_{\pi^h(4j)}) \mu_0(h_{\pi^h(4j - 2)}) \]
Next, note
\begin{align}
\nonumber & \frac{2}{n} \sum_{1 \leq j \leq \lfloor\frac{n}{2}\rfloor} E[|Y_{\pi^h(4j - 2)} Y_{\pi^h(4j)} - E[Y_{\pi^h(4j - 2)} Y_{\pi^h(4j)} | h^{(n)}]| \\
\nonumber & \hspace{3em} I \{|Y_{\pi^h(4j - 2)} Y_{\pi^h(4j)} - E[Y_{\pi^h(4j - 2)} Y_{\pi^h(4j)} | h^{(n)}]| > \lambda \} | h^{(n)}] \\
\nonumber & \lesssim \frac{2}{n} \sum_{1 \leq j \leq \lfloor\frac{n}{2}\rfloor} E[|Y_{\pi^h(4j - 2)} Y_{\pi^h(4j)}| I \{|Y_{\pi^h(4j - 2)} Y_{\pi^h(4j)}| > \sqrt{\lambda / 2} \} | h^{(n)}] \\
\nonumber & \hspace{3em} + E[|E[Y_{\pi^h(4j - 2)} Y_{\pi^h(4j)} | h^{(n)}]| I \{|E[Y_{\pi^h(4j - 2)} Y_{\pi^h(4j)} | h^{(n)}]| > \sqrt{\lambda / 2}\} | h^{(n)}] \\
\nonumber & \lesssim \frac{2}{n} \sum_{1 \leq j \leq \lfloor\frac{n}{2}\rfloor} E[Y_{\pi^h(4j - 2)}^2 I \{|Y_{\pi^h(4j - 2)}| > \sqrt{\lambda / 2}\} | h^{(n)}] + E[Y_{\pi^h(4j)}^2 I \{|Y_{\pi^h(4j)}| > \sqrt{\lambda / 2}\} | h^{(n)}] \\
\nonumber & \hspace{3em} + |\mu_1(h_{\pi^h(4j - 2)}) \mu_0(h_{\pi^h(4j)})| I \{\mu_1(h_{\pi^h(4j - 2)}) \mu_0(h_{\pi^h(4j)}) > \lambda / 2\}\\
\nonumber & \hspace{3em} + |\mu_1(h_{\pi^h(4j)}) \mu_0(h_{\pi^h(4j - 2)})| I \{\mu_1(h_{\pi^h(4j)}) \mu_0(h_{\pi^h(4j - 2)}) > \lambda / 2\} \\
\nonumber & \lesssim \frac{2}{n} \sum_{1 \leq j \leq \lfloor\frac{n}{2}\rfloor} E[Y_{\pi^h(4j - 2)}^2(1) I \{|Y_{\pi^h(4j - 2)}(1)| > \sqrt{\lambda / 2}\} | h^{(n)}] \\
\nonumber & \hspace{3em} + E[Y_{\pi^h(4j - 2)}^2(0) I \{|Y_{\pi^h(4j - 2)}(0)| > \sqrt{\lambda / 2}\} | h^{(n)}] \\
\nonumber & \hspace{3em} + E[Y_{\pi^h(4j)}^2(1) I \{|Y_{\pi^h(4j)}(1)| > \sqrt{\lambda / 2}\} | h^{(n)}] + E[Y_{\pi^h(4j)}^2(0) I \{|Y_{\pi^h(4j)}(0)| > \sqrt{\lambda / 2}\} | h^{(n)}] \\
\nonumber & \hspace{3em} + \mu_1^2(h_{\pi^h(4j - 2)}) I \{|\mu_1(h_{\pi^h(4j - 2)})| > \sqrt{\lambda / 2}\} + \mu_0^2(h_{\pi^h(4j)}) I \{|\mu_0(h_{\pi^h(4j)})| > \sqrt{\lambda / 2}\} \\
\nonumber & \hspace{3em} + \mu_1^2(h_{\pi^h(4j)}) I \{|\mu_1(h_{\pi^h(4j)})| > \sqrt{\lambda / 2}\} + \mu_0^2(h_{\pi^h(4j - 2)}) I \{|\mu_0(h_{\pi^h(4j - 2)})| > \sqrt{\lambda / 2}\} \\
\nonumber & \lesssim \frac{1}{2n} \sum_{1 \leq i \leq 2n} E[Y_i^2(1) I \{|Y_i(1) > \sqrt{\lambda / 2}|\} | h(X_i)] + E[Y_i^2(0) I \{|Y_i(1) > \sqrt{\lambda / 2}|\} | h(X_i)] \\
\nonumber & \hspace{3em} + E[Y_i^2(1) | h(X_i)] I \{E[Y_i^2(1) | h(X_i)] > \sqrt{\lambda / 2}\} + E[Y_i^2(0) | h(X_i)] I \{E[Y_i^2(0) | h(X_i)] > \sqrt{\lambda / 2}\} \\
\nonumber & \stackrel{P}{\to} E[Y_i^2(1) I \{|Y_i(1) > \sqrt{\lambda / 2}|\}] + E[Y_i^2(0) I \{|Y_i(1) > \sqrt{\lambda / 2}|\}] \\
\nonumber & \hspace{3em} + E[E[Y_i^2(1) | h(X_i)] I \{E[Y_i^2(1) | h(X_i)] > \sqrt{\lambda / 2}\}] \\
\label{eq:ui} & \hspace{3em} + E[E[Y_i^2(0) | h(X_i)] I \{E[Y_i^2(0) | h(X_i)] > \sqrt{\lambda / 2}\}]~,
\end{align}
where the last line follows from WLLN and the law of iterated expectation. Since by Assumption \ref{as:moments} I have $E[Y_i^2(d)] < \infty$ and hence $E[E[Y_i(d) | h(X_i)]^2] < E[Y_i^2(d)]$ by Jensen's inequality, the limit as $\lambda \to \infty$ of the last line is 0, by the dominated convergence theorem. I finish the proof by arguing by contradiction. Suppose 
\[ \hat \rho_n - E[\hat \rho_n | h^{(n)}] \]
does not converge in probability to $0$. There must then exist $\epsilon > 0$ and $\delta > 0$ and a subsequence, which for simplicity I again denote by $\{n\}$, such that
\begin{equation} \label{eq:rho-contra}
P\{|\hat \rho_n - E[\hat \rho_n | h^{(n)}]| > \epsilon\} \to \delta
\end{equation}
along this subsequence. But because of \eqref{eq:ui}, there exists a further subsequence along which the condition in Lemma 6.3 of \cite{bai2021inference} holds with probability one for $h^{(n)}$, but then along this subsequence $\hat \rho_n - E[\hat \rho_n | h^{(n)}] \stackrel{P}{\to} 0$ conditional on $h^{(n)}$ with probability one for $h^{(n)}$, i.e., for any $\epsilon > 0$, with probability one for $h^{(n)}$,
\[ P\{|\hat \rho_n - E[\hat \rho_n | h^{(n)}]| > \epsilon | h^{(n)}\} \to 0~. \]
Since probabilities are bounded and hence uniformly integrable,
\[ P\{|\hat \rho_n - E[\hat \rho_n | h^{(n)}]| > \epsilon\} \to 0 \]
along the chosen subsequence, which implies a contradiction to \eqref{eq:rho-contra}. \qed

\subsection{Proof of Theorem \ref{thm:l2}}
By repeating the arguments in the proof of Lemma \ref{lem:tau-limit}, I write
\[\sqrt n (\hat \theta_n- \theta(Q)) = A_n - B_n + C_n - D_n~, \]
where
\begin{align*}
A_n & = \frac{1}{\sqrt n} \sum_{1 \leq i \leq 2n} (Y_i(1) D_i - E[Y_i(1) D_i | h^{(n)}, D^{(n)}]) \\
B_n & = \frac{1}{\sqrt n} \sum_{1 \leq i \leq 2n} (Y_i(0) (1 - D_i) - E[Y_i(0) (1 - D_i) | h^{(n)}, D^{(n)}]) \\
C_n & = \frac{1}{\sqrt n} \sum_{1 \leq i \leq 2n} (E[(Y_i(1) + Y_i(0)) D_i | h^{(n)}, D^{(n)}] - D_i E[Y_i(1) + Y_i(0)]) \\
D_n & = \frac{1}{\sqrt n} \sum_{1 \leq i \leq 2n} (E[Y_i(0) | h^{(n)}, D^{(n)}] - E[Y_i(0)])~.
\end{align*}
Note by the assumption that $E[h^2(X_i)] < \infty$, Assumption \ref{as:l2}, and Lemma \ref{lem:hat}, \eqref{eq:hat-2} holds. Since $0 < E[\var[Y_i(d) | h(X_i)]]$ for $d \in \{0, 1\}$, $E[Y_i^r(d) | h(X_i) = z]$ is Lipschitz in $z$ for $r = 1, 2$ and $d = 0, 1$, and \eqref{eq:hat-2} holds, by repeating the arguments in the proof of Lemma \ref{lem:tau-limit} with $\tau = \frac{1}{2}$, the desired convergence in distribution holds. In fact, instead of requiring $E[Y_i(d) | h(X_i) = z]$ and $E[Y_i^2(d) | h(X_i) = z]$ to both be Lipschitz continuous, it suffices to require $\var[Y_i(d) | h(X_i) = z]$ to be Lipschitz continuous. 

Next, I show $\hat \varsigma_{\tilde g_m, n}^2 \stackrel{P}{\to} \varsigma_g^2$ as $m, n \to \infty$. Similar arguments as those used in Theorem \ref{thm:se} go through if \eqref{eq:hat-2} and \eqref{eq:hat-4} hold. Since \eqref{eq:hat-4} follows from Assumptions \ref{as:l2} by Lemma \ref{lem:hat}, the conclusion follows. \qed

\section{Auxiliary Lemmas}
\subsection{Auxiliary Lemmas for Main Results}
\begin{lemma} \label{lem:hardy-littlewood}
Suppose $m \geq 2$, and $x_1, \dots, x_{2m}$ are real number such that $x_1 \leq \dots \leq x_{2m}$. Then, for any $\pi \in \Pi_n$,
\[ \sum_{k = 1}^m x_{\pi(2k - 1)} x_{\pi(2k)} \leq \sum_{k = 1}^m x_{2k - 1} x_{2k}~. \]
\end{lemma}

\begin{proof}[\sc Proof of Lemma \ref{lem:hardy-littlewood}]
I start by considering $\pi$ which only permutes the indices $\{k_1, k_2, k_3, k_4\}$ and leaves the other entries intact. I need only consider the case where there exists $k_1 < k_2 < k_3 < k_4$ such that at least one of $\pi(k_1), \pi(k_2)$ is greater than at least one of $\pi(k_3), \pi(k_4)$ because the lemma trivially holds otherwise. Suppose without loss of generality that $\pi(k_2) < \pi(k_3) < \pi(k_4) < \pi(k_1)$, then it is easy to verify
\[ x_{\pi(k_1)} x_{\pi(k_2)} + x_{\pi(k_3)} x_{\pi(k_4)} \leq x_{\pi(k_2)} x_{\pi(k_3)} + x_{\pi(k_1)} x_{\pi(k_4)} \]
so that by interchanging two indices I decrease the sum weakly. To conclude the proof, note a finite number of those interchanges maps $\pi$ back to the identity operator, and the conclusion follows.
\end{proof}

\begin{lemma} \label{lem:cond-convd}
Let $X_n$, $Y_n$, $Z_n$ be random variables. Suppose $Y_n = g(Z_n) \stackrel{d}{\to} Y$ as $n \to \infty$, where $g: \mathbf R \to \mathbf R$ is measurable and $X_n \stackrel{d}{\to} X$ conditional on $Z_n$, with probability one for $Z_n$. Furthermore, suppose the distributions of both $X$ and $Y$ are continuous everywhere. Then
\[ (X_n, Y_n) \stackrel{d}{\to} (X, Y)~, \]
where $X \indep Y$.
\end{lemma}

\begin{proof}[\sc Proof of Lemma \ref{lem:cond-convd}]
Since $X$ and $Y$ both have continuous distribution functions and they are independent, I need only show for any $x, y \in \mathbf R$,
\[ P\{X_n \leq x, Y_n \leq y\} \to P\{X \leq x\} P\{Y \leq y\}~. \]
To this end, note
\begin{align}
\nonumber & P\{X_n \leq x, Y_n \leq y\} - P\{X \leq x\} P\{Y \leq y\} \\
\nonumber & = E[E[I \{X_n \leq x\} I \{Y_n \leq y\}|Z_n]] - P\{X \leq x\} P\{Y \leq y\} \\
\nonumber & = E[E[I \{X_n \leq x\} | Z_n] I \{Y_n \leq y\}] - P\{X \leq x\} P\{Y \leq y\} \\
\nonumber & = E[(E[I \{X_n \leq x\} | Z_n] - P\{X \leq x\}) I \{Y_n \leq y\}] + E[P\{X \leq x\} (I \{Y_n \leq y\} - P\{Y \leq y\})] \\
\nonumber & = E[(P\{X_n \leq x | Z_n\} - P\{X \leq x\}) I \{Y_n \leq y\}] + (P\{Y_n \leq y\} - P\{Y \leq y\}) P\{X \leq x\}~.
\end{align}
Since
\[ P\{X_n \leq x | Z_n\} - P\{X \leq x\} \to 0 \]
with probability one for $Z_n$, and hence the product inside the expectation converges to 0 with probability one as well, which in turn implies the expectation converges to 0 by the dominated convergence theorem because probabilities are bounded. The second term converges to 0 because of the definition of convergence in distribution and the fact that the distribution function of $Y$ is continuous everywhere.
\end{proof}

\begin{lemma} \label{lem:tau-limit}
Suppose the sample size is $kn$ for $k \in \mathbf Z$ and the treatment assignment scheme satisfies $\tau_s \equiv \tau = \frac{l}{k}$, where $l \in \mathbf Z$, $0 < l < k$, and they are mutually prime. Suppose $Q$ satisfies Assumption \ref{as:moments} and $h$ satisfies the assumptions in Theorem \ref{thm:tau-min}. Then, under $\lambda^{\tau, h}(X^{(n)})$ defined in \eqref{eq:tau-h}, as $n \to \infty$,
\[ \sqrt{kn} (\hat \theta_n - \theta(Q)) \stackrel{d}{\to} N(0, \varsigma_{\tau, h}^2)~, \]
where
\begin{equation} \label{eq:tau-limit}
\varsigma_{\tau, h}^2 = \frac{\var[Y_i(1)]}{\tau} + \frac{\var[Y_i(0)]}{1 - \tau} - \tau (1 - \tau) E \Big [ \Big ( E \Big [ \frac{Y_i(1)}{\tau} + \frac{Y_i(0)}{1 - \tau} \Big | h(X_i) \Big] - \Big ( \frac{E[Y_i(1)]}{\tau} + \frac{E[Y_i(0)]}{1 - \tau} \Big) \Big )^2 \Big ]~.
\end{equation}
\end{lemma}

\begin{proof}[\sc Proof of Lemma \ref{lem:tau-limit}]
To begin with, note
\[ \sqrt{kn} (\hat \theta_n- \theta(Q)) = A_n - B_n + C_n - D_n~, \]
where
\[ \begin{split}
A_n & = \frac{1}{\sqrt{kn}} \sum_{1 \leq i \leq kn} \Big ( \frac{Y_i(1) D_i}{\tau} - E \Big [ \frac{Y_i(1) D_i}{\tau} \Big | h^{(n)}, D^{(n)} \Big ] \Big ) \\
B_n & = \frac{1}{\sqrt{kn}} \sum_{1 \leq i \leq kn} \Big ( \frac{Y_i(0) (1 - D_i)}{1 - \tau} - E \Big [ \frac{Y_i(0) (1 - D_i)}{1 - \tau} \Big | h^{(n)}, D^{(n)} \Big ] \Big ) \\
C_n & = \frac{1}{\sqrt{kn}} \sum_{1 \leq i \leq kn} \Big ( E \Big [ \frac{Y_i(1) D_i}{\tau} \Big | h^{(n)}, D^{(n)} \Big ] - E[Y_i(1)] \Big ) \\
D_n & = \frac{1}{\sqrt{kn}} \sum_{1 \leq i \leq kn} \Big ( E \Big [ \frac{Y_i(0) (1 - D_i)}{1 - \tau} \Big | h^{(n)}, D^{(n)} \Big ] - E[Y_i(0)] \Big )~.
\end{split} \]
Note conditional on $h^{(n)}$ and $D^{(n)}$, $A_n$ and $B_n$ are independent and $C_n$ and $D_n$ are constant.

I first study the limiting behavior of $A_n$. Conditional on $h^{(n)}$ and $D^{(n)}$, the terms in the sum are independent but not identically distributed. Therefore, I proceed to verify the Lindeberg condition holds in probability conditional on $h^{(n)}$ and $D^{(n)}$. To that end, define
\[ s_n^2 = s_n^2(h^{(n)}, D^{(n)}) = \sum_{1 \leq i \leq kn} \var \Big [ \frac{Y_i(1) D_i}{\tau} \Big | h^{(n)}, D^{(n)} \Big ] \]
and note
\begin{align}
\nonumber s_n^2 & = \sum_{1 \leq i \leq kn} \var \Big [ \frac{Y_i(1) D_i}{\tau} \Big | h^{(n)}, D^{(n)} \Big ] \\
\nonumber & = \frac{1}{\tau^2} \sum_{1 \leq i \leq kn} D_i \var[Y_i(1) | h^{(n)}] \\
\nonumber & = \frac{1}{\tau^2} \sum_{1 \leq i \leq kn: D_i = 1} \var[Y_i(1) | h(X_i)]~,
\end{align}
where the second equality follows from \eqref{eq:indep} and the third follows from the fact that units are i.i.d. It follows that
\begin{multline} \label{eq:sn-decomp}
\tau \frac{s_n^2}{kn} = \frac{1}{kn} \sum_{1 \leq i \leq kn} \var[Y_i(1) | h(X_i)] + \Big ( \frac{1 - \tau}{\tau kn}\sum_{1 \leq i \leq kn: D_i = 1} \var[Y_i(1) | h(X_i)] \\
- \frac{1}{kn}\sum_{1 \leq i \leq kn: D_i = 0} \var[Y_i(1) | h(X_i)] \Big )~.
\end{multline}
By Assumption \ref{as:moments},
\begin{equation} \label{sn:wlln}
\frac{1}{kn} \sum_{1 \leq i \leq kn} \var[Y_i(1) | h(X_i)] \stackrel{P}{\to} E[\var[Y_i(1) | h(X_i)]] < E[Y_i(1)] < \infty~.
\end{equation}
Meanwhile,
\begin{align}
\nonumber & \Big |  \frac{1 - \tau}{\tau kn}\sum_{1 \leq i \leq kn: D_i = 1} \var[Y_i(1) | h(X_i)] - \frac{1}{kn}\sum_{1 \leq i \leq kn: D_i = 0} \var[Y_i(1) | h(X_i)] \Big | \\
\nonumber & \lesssim \Big |  \frac{1 - \tau}{\tau kn}\sum_{1 \leq i \leq kn: D_i = 1} h_i - \frac{1}{kn}\sum_{1 \leq i \leq kn: D_i = 0} h_i \Big | \\
\nonumber & = \frac{1}{\tau k n} \Big | \sum_{1 \leq s \leq n} \sum_{(s - 1)k + 1 \leq j \leq sk: D_{\pi^{\tau, h}(j)} = 1} (h_{\pi^{\tau, h}(j)} - \bar h_s^{\tau}) \Big | \\
\nonumber & \leq \frac{1}{\tau k n} \sum_{1 \leq s \leq n} \sum_{(s - 1)k + 1 \leq j \leq sk: D_{\pi^{\tau, h}(j)} = 1} |h_{\pi^{\tau, h}(j)} - \bar h_s^{\tau}| \\
\nonumber & \lesssim \frac{1}{n} \sum_{1 \leq s \leq n} \sum_{(s - 1)k + 1 \leq j \leq sk}|h_{\pi^{\tau, h}(j)} - \bar h_s^{\tau}| \\
\label{eq:sn-close} & \leq k^{1/2} \Big (\frac{1}{n} \sum_{1 \leq s \leq n} \sum_{(s - 1)k + 1 \leq j \leq sk}|h_{\pi^{\tau, h}(j)} - \bar h_s^{\tau}|^2 \Big )^{1/2} \stackrel{P}{\to} 0~,
\end{align}
where the first inequality holds because $E[Y_i^r(d) | h(X_i) = z]$ is Lipschitz for $r = 1, 2$ and $d = 0, 1$, the second holds by assumption, the third holds by inspection, the second to last holds by the Cauchy-Schwarz inequality, the last holds by condition (c) in Theorem \ref{thm:tau-min}, and the equality holds by inspection. Combining \eqref{eq:sn-decomp}, \eqref{sn:wlln}, and \eqref{eq:sn-close}, I have
\begin{equation} \label{eq:sn-plim}
\frac{s_n^2}{kn} \stackrel{P}{\to} \frac{E[\var[Y_i(1) | h(X_i)]]}{\tau} > 0~,
\end{equation}
where the inequality holds by assumption.

I now argue the Lindeberg condition holds in probability conditional on $h^{(n)}$ and $D^{(n)}$, i.e., for any $\epsilon > 0$,
\begin{multline}
E_n = \frac{1}{s_n^2 \tau^2} \sum_{1 \leq i \leq kn} E[|Y_i(1) D_i - E[Y_i(1)D_i | h^{(n)}, D^{(n)}]|^2 \\
I \{|Y_i(1) D_i - E[Y_i(1) D_i | h^{(n)}, D^{(n)}]| > \epsilon \tau s_n\} | h^{(n)}, D^{(n)}] \stackrel{P}{\to} 0~.
\end{multline}
To this end, first note for any $M > 0$, 
\begin{equation} \label{eq:epsilon-sn}
P\{\epsilon \tau s_n > M\} \to 1
\end{equation}
because of \eqref{eq:sn-plim}. Next, note
\[ E[Y_i(1) D_i | h^{(n)}, D^{(n)}] = E[Y_i(1) | h(X_i)] D_i \]
because of \eqref{eq:indep}. As a result, for any $M > 0$
\begin{align}
\nonumber E_n & = \frac{1}{s_n^2 \tau^2} \sum_{1 \leq i \leq kn: D_i = 1} E[|Y_i(1) - E[Y_i(1) | h^{(n)}, D^{(n)}]|^2 I \{|Y_i(1) - E[Y_i(1 | h^{(n)}, D^{(n)}]| > \epsilon \tau s_n\} | h^{(n)}, D^{(n)}] \\
\nonumber & \leq \frac{1}{s_n^2 \tau^2} \sum_{1 \leq i \leq kn} E[|Y_i(1)- E[Y_i(1) | h^{(n)}, D^{(n)}]|^2 I \{|Y_i(1) - E[Y_i(1)| h^{(n)}, D^{(n)}]| > \epsilon \tau s_n\} | h^{(n)}, D^{(n)}] \\
\nonumber & \leq \frac{1}{s_n^2 \tau^2} \sum_{1 \leq i \leq kn} E[|Y_i(1) - E[Y_i(1) | h(X_i)] |^2 I \{|Y_i(1) - E[Y_i(1) | h(X_i)] > M\} | h^{(n)}, D^{(n)}] + o_p(1) \\
\label{eq:En-dom} & = \frac{kn}{s_n^2 \tau^2} \frac{1}{kn} \sum_{1 \leq i \leq kn} E[|Y_i(1) - E[Y_i(1) | h(X_i)] |^2 I \{|Y_i(1) - E[Y_i(1) | h(X_i)]| > M\} | h^{(n)}, D^{(n)}] + o_p(1) \\
\label{eq:En-dom-plim} & \stackrel{P}{\to} (E[\var[Y_i(1) | h(X_i)]])^{-1} E[|Y_i(1) - E[Y_i(1) | h(X_i)] |^2 I \{|Y_i(1) - E[Y_i(1) | h(X_i)]| > M\}]~,
\end{align}
where the first inequality holds by inspection, the second holds because of \eqref{eq:epsilon-sn} and the equality follows because \eqref{eq:indep} and $Q_n = Q^{kn}$, and the convergence in probability follows from \eqref{eq:sn-plim} and the fact that Assumption \ref{as:moments} implies
\begin{multline} \nonumber
E[|Y_i(1) - E[Y_i(1) | h(X_i)] |^2 I \{|Y_i(1) - E[Y_i(1) | h(X_i)]| > M\}] \\ 
\leq E[|Y_i(1) - E[Y_i(1) | h(X_i)] |^2] = E[\var[Y_i(1) | h(X_i)]] \leq E[Y_i^2(1)] < \infty~.
\end{multline}
In addition, by the dominated convergence theorem, 
\[ \lim_{M \to \infty} E[|Y_i(1) - E[Y_i(1) | h(X_i)] |^2 I \{|Y_i(1) - E[Y_i(1) | h(X_i)]| > M\}] = 0~. \]
To show $E_n \stackrel{P}{\to} 0$, fix any subsequence $\{n(j)\}$, and I argue there is a further subsequence $\{n(j(k(l)))\}$ along which $E_{n(j(k(l)))}$ converges to 0 almost surely. Indeed, for the subsequence $\{n(j)\}$, for any fixed $M$, $E_{n(j)}$ is bounded by \eqref{eq:En-dom}, which I define as $U_{n(j)}(M)$. I know from above that $U_{n(j)}(M) \stackrel{P}{\to} U(M)$, where $U(M)$ is defined as \eqref{eq:En-dom-plim}. Hence, there exists a further subsequence $\{n(j(k))\}$ along which $U_{n(j(k))}(M) \to U(M)$ almost surely. I then choose a sequence $\{M_{n(j(k))}\}_{n \geq 1}$ such that $M_{n(j(k))} \to \infty$. By the dominated convergence theorem, $\lim_{n \to \infty} U(M_{n(j(k))}) = 0$. By a diagonalization argument, I could construct a further subsequence $\{n(j(k(l)))\}$ along which $U_{n(j(k(l)))}(M_{n(j(k(l)))}) \to 0$. Along this subsequence, because $E_n \leq U_n(M_n)$ for each $n$, the almost sure limit of $E_n$ must be zero because it is non-negative.

I now argue
\[ \sup_{t \in \mathbf R} | P\{A_n \leq t | h^{(n)}, D^{(n)}\} - \Phi ( t / \sqrt{E[\var[Y_i(1) | h(X_i)]] / \tau} ) | \stackrel{P}{\to} 0~. \]
Fix any subsequence. Since $E_n \stackrel{P}{\to} 0$, there exists a further subsequence along which $E_n \to 0$ with probability one for $h^{(n)}, D^{(n)}$. Because of the Lindeberg condition and \eqref{eq:sn-plim}, it follows that with probability one for $h^{(n)}, D^{(n)}$, $A_n \stackrel{d}{\to} N(0, E[\var[Y_i(1) | h(X_i)]] / \tau)$ conditional on $h^{(n)}, D^{(n)}$. But then the left-hand side of the preceding display must converge almost surely to 0 by Pólya's theorem. Since for any subsequence there exists a further subsequence along which it converges to 0 almost surely, it must converge to 0 in probability.

A similar argument establishes
\[ \sup_{t \in \mathbf R}| P\{B_n \leq t | h^{(n)}, D^{(n)}\} - \Phi(t / \sqrt{E[\var[Y_i(0) | h(X_i)]] / (1 - \tau)})| \stackrel{P}{\to} 0~. \]
Since $A_n$ and $B_n$ are independent conditional on $h^{(n)}$ and $D^{(n)}$, it follows by a similar subsequencing argument that
\begin{equation} \label{eq:An-Bn}
\sup_{t \in \mathbf R}| P\{A_n - B_n \leq t | h^{(n)}, D^{(n)}\} - \Phi( t / \sqrt{E[\var[Y_i(1) | h(X_i)] / \tau + E[\var[Y_i(0) | h(X_i)]] / (1 - \tau)})| \stackrel{P}{\to} 0~.
\end{equation}

To study $C_n$, note by \eqref{eq:indep},
\[ C_n = \frac{1}{\sqrt{kn}} \sum_{1 \leq i \leq kn} \Big (E \Big [ \frac{Y_i(1)}{\tau} \Big | h(X_i) \Big ] D_i - E[Y_i(1)] \Big )~. \]
So I have
\[ E[C_n | h^{(n)}] = \frac{1}{\sqrt{kn}} \sum_{1 \leq i \leq kn} (E[Y_i(1) | h(X_i)] - E[Y_i(1)])~. \]
Furthermore, by the assumptions that $E[Y_i^r(d) | h(X_i) = z]$ is Lipschitz for $r = 1, 2$ and $d = 0, 1$ and $\displaystyle \frac{1}{n} \displaystyle \sum_{1 \leq s \leq n} \sum_{(s - 1)k + 1 \leq j \leq sk}|h_{\pi^{\tau, h}(j)} - \bar h_s^{\tau}|^2 \stackrel{P}{\to} 0$, I have
\[ \var[C_n | h^{(n)}] \propto \frac{1}{kn} \sum_{1 \leq s \leq n} (h_{\pi^{\tau, h}(i)} - \bar h_\tau^s)^2 \stackrel{P}{\to} 0~, \]
where the first relation can be established by repeating the arguments used in the last step of establishing Theorem \ref{thm:tau-oracle}. It therefore follows by Chebyshev's inequality that for any $\epsilon > 0$,
\[ P\{|C_n - E[C_n | h^{(n)}]| > \epsilon | h^{(n)}\} \stackrel{P}{\to} 0~, \]
and because probabilities are bounded and hence uniformly integrable,
\[ P\{|C_n - E[C_n | h^{(n)}]| > \epsilon \} \stackrel{P}{\to} 0~, \]
and hence
\[ C_n = \frac{1}{\sqrt{kn}} \sum_{1 \leq i \leq kn} (E[Y_i(1) | h(X_i)] - E[Y_i(1)]) + o_p(1)~. \]
A similar proof shows
\[ D_n = \frac{1}{\sqrt{kn}} \sum_{1 \leq i \leq kn} (E[Y_i(0) | h(X_i)] - E[Y_i(0)]) + o_p(1)~. \]
and therefore
\begin{align} \label{eq:Cn-Dn}
\nonumber C_n - D_n & = \frac{1}{\sqrt {kn}} \sum_{1 \leq i \leq kn} ( E[Y_i(1) | h(X_i)] - E[Y_i(1)]- (E[Y_i(0) | h(X_i)] - E[Y_i(0)])) + o_p(1) \\
\nonumber & \stackrel{d}{\to} N \Big ( 0, E \Big [ ( E[Y_i(1) | h(X_i)] - E[Y_i(1)]- (E[Y_i(0) | h(X_i)] - E[Y_i(0)]) )^2 \Big ] \Big )~.
\end{align}

I now show by contradiction that
\[ \sup_{t \in \mathbf R} | P\{\sqrt n(\hat \theta_n- \theta(Q)) \leq t\} - \Phi(t / \varsigma_h) | \to 0~. \]
Suppose not, then there must exist a subsequence along which the left-hand side of the above display converges to some $\delta > 0$. Along this subsequence, I could find a further subsequence along which the left-hand side of \eqref{eq:An-Bn} converges to 0 with probability one for $h^{(n)}$ and $D^{(n)}$, i.e.,
\[ A_n - B_n \stackrel{d}{\to} N \Big ( 0, \frac{E[\var[Y_i(1) | h(X_i)]}{\tau} + \frac{E[\var[Y_i(0) | h(X_i)]]}{1 - \tau} \Big ) \]
with probability one for $h^{(n)}$ and $D^{(n)}$. Since $C_n - D_n$ is constant for each $h^{(n)}$ and $D^{(n)}$, Lemma \ref{lem:cond-convd} establishes
\begin{multline} \nonumber
A_n - B_n + C_n - D_n \stackrel{d}{\to} N \Big ( 0, \frac{E[\var[Y_i(1) | h(X_i)]}{\tau} + \frac{E[\var[Y_i(0) | h(X_i)]]}{1 - \tau} + \\
E \Big [ (E[Y_i(1) | h(X_i)] - E[Y_i(1)]- (E[Y_i(0) | h(X_i)] - E[Y_i(0)]))^2 \Big ] \Big )~,
\end{multline}
which, by P\'{o}lya's Theorem, implies a contradiction.

Finally, note
\begin{align*}
& \frac{E[\var[Y_i(1) | h(X_i)]}{\tau} + \frac{E[\var[Y_i(0) | h(X_i)]]}{1 - \tau} \\
& \hspace{3em} + E \Big [( E[Y_i(1) | h(X_i)] - E[Y_i(1)]- (E[Y_i(0) | h(X_i)] - E[Y_i(0)]))^2 \Big ] \\
& = \frac{\var[Y_i(1)]}{\tau} + \frac{\var[Y_i(0)]}{1 - \tau} - \frac{\var[E[Y_i(1) | h(X_i)]]}{\tau} - \frac{\var[E[Y_i(0) | h(X_i)]]}{1 - \tau} + \\
& \hspace{3em} E \Big [ (E[Y_i(1) | h(X_i)] - E[Y_i(1)] -( E[Y_i(0) | h(X_i)] - E[Y_i(0)]))^2 \Big ] \\
& = \frac{\var[Y_i(1)]}{\tau} + \frac{\var[Y_i(0)]}{1 - \tau} - \frac{1 - \tau}{\tau} E[(E[Y_i(1) | h(X_i)] - E[Y_i(1)])^2] \\
& \hspace{3em} - \frac{\tau}{1 - \tau} E[(E[Y_i(0) | h(X_i)] - E[Y_i(0)])^2] \\
& \hspace{3em} - 2 E \Big [( E[Y_i(1) | h(X_i)] - E[Y_i(1)])(E[Y_i(0) | h(X_i)] - E[Y_i(0)]) \Big ] \\
& = \frac{\var[Y_i(1)]}{\tau} + \frac{\var[Y_i(0)]}{1 - \tau} - \tau (1 - \tau) E \Big [ \Big ( E \Big [ \frac{Y_i(1)}{\tau} + \frac{Y_i(0)}{1 - \tau} \Big | h(X_i) \Big] - \Big ( \frac{E[Y_i(1)]}{\tau} + \frac{E[Y_i(0)]}{1 - \tau} \Big) \Big )^2 \Big ]~,
\end{align*}
and the result follows.
\end{proof}

\begin{lemma} \label{lem:max}
Suppose $U_i$, $1 \leq i \leq n$ are i.i.d.\ random variables where $E|U_i|^r < \infty$. Then
\[ n^{-1/r} \max_{1 \leq i \leq n} |U_i| \stackrel{P}{\to} 0~. \]
\end{lemma}

\begin{proof}[\sc Proof of Lemma \ref{lem:max}]
Note for all $\epsilon > 0$,
\begin{multline} \nonumber
P \Big \{ n^{-1/r} \max_{1 \leq i \leq n} |U_i| > \epsilon \Big \} = P \Big \{ \max_{1 \leq i \leq n} |U_i|^r > n \epsilon^r \Big \} \\
\leq n P \{|U_i|^r > n \epsilon^r\} \leq \frac{n}{n \epsilon^r} E[|U_i|^r I \{|U_i|^r > n \epsilon^r\}] = \frac{1}{\epsilon^r} E[|U_i|^r I \{|U_i|^r > n \epsilon^r\}] \to 0~,
\end{multline}
where the convergence follows because of the dominated convergence theorem and that $E|U_i|^r < \infty$.
\end{proof}

\begin{lemma} \label{lem:close}
Suppose $E[h^2(X_i)] < \infty$. Then, as $n \to \infty$,
\begin{equation} \label{eq:2}
\frac{1}{n} \sum_{1 \leq s \leq n} |h_{\pi^h(2s - 1)} - h_{\pi^h(2s)}|^2 \stackrel{P}{\to} 0~,
\end{equation}
and
\begin{equation} \label{eq:4}
\frac{1}{n} \sum_{1 \leq j \leq \lfloor\frac{n}{2}\rfloor} |h_{\pi^h(4j - k)} - h_{\pi^h(4j - l)}|^2 \stackrel{P}{\to} 0 \text{ for } k \in \{2, 3\} \text{ and } l \in \{0, 1\}~.
\end{equation}

\end{lemma}

\begin{proof}[\sc Proof of Lemma \ref{lem:close}]
Note
\[ \sum_{1 \leq s \leq n} |h_{\pi^h(2s - 1)} - h_{\pi^h(2s)}|^2 \leq |h_{\pi^h(2n)} - h_{\pi^h(1)}|^2 \leq 4 \max_{1 \leq i \leq 2n} h^2(X_i)~, \]
where the first inequality follows from the definition of $\pi^h$ and the second inequality follows by inspection, and therefore it follows from Lemma \ref{lem:max} that
\[ \frac{1}{n} \sum_{1 \leq s \leq n} |h_{\pi^h(2s - 1)} - h_{\pi^h(2s)}|^2 \leq \frac{4}{n} \max_{1 \leq i \leq 2n} h^2(X_i) \stackrel{P}{\to} 0~. \]
\eqref{eq:2} thus holds. To see Assumption \ref{eq:4} holds, note
\[ \frac{1}{n} \sum_{1 \leq j \leq \lfloor\frac{n}{2}\rfloor} |h_{\pi^h(4j - k)} - h_{\pi^h(4j - l)}|^2 \lesssim \frac{1}{n} |h_{\pi^h(2n)} - h_{\pi^h(1)}|^2~, \]
and the result follows similarly.
\end{proof}

\begin{lemma} \label{lem:hat}
Suppose $E[h^2(X_i)] < \infty$ and $\tilde h_m$ satisfies Assumption \ref{as:l2}. Then, as $m, n \to \infty$,
\begin{equation} \label{eq:hat-2}
\frac{1}{n} \sum_{1 \leq s \leq n} |h_{\pi^{\tilde h_m}(2s - 1)} - h_{\pi^{\tilde h_m}(2s)}|^2  \stackrel{P}{\to} 0~,
\end{equation}
and
\begin{equation} \label{eq:hat-4}
\frac{1}{n} \sum_{1 \leq j \leq \lfloor\frac{n}{2}\rfloor} |h_{\pi^{\tilde h_m}(4j - k)} - h_{\pi^{\tilde h_m}(4j - l)}|^2 \stackrel{P}{\to} 0 \text{ for } k \in \{2, 3\} \text{ and } l \in \{0, 1\}~.
\end{equation}
\end{lemma}

\begin{proof}[\sc Proof of Lemma \ref{lem:hat}]
I only prove the first conclusion as the second can be shown by similar arguments. I first show Assumption \ref{as:l2} implies
\begin{equation} \label{eq:emp-l2}
\frac{1}{n} \sum_{1 \leq i \leq 2n} |\tilde h_i - h_i|^2 \stackrel{P}{\to} 0~.
\end{equation}
Suppose Assumption \ref{as:l2} holds. For any $\epsilon > 0$, $\delta > 0$, there exists $M > 0$ such that for $m > M$,
\begin{equation} \label{eq:inner}
P \Big \{ \int_{\mathrm{supp}(X_i)} |\tilde h_m(x) - h(x)|^2\,Q_X(d x) > \frac{\epsilon \delta}{2} \Big \} \leq \frac{\delta}{2}~.
\end{equation}
By Chebyshev's inequality again, if
\[ \int_{\mathrm{supp}(X_i)} |\tilde h_m(x) - h(x)|^2\,Q_X(d x) \leq \frac{\epsilon \delta}{2}~, \]
then by the independence of $\tilde W^{(m)}$ and $W^{(n)}$,
\begin{multline} \label{eq:outer}
P \Big \{ \frac{1}{2n} \sum_{1 \leq i \leq 2n} |\tilde h_i - h_i|^2 > \epsilon \Big | \tilde W^{(m)} \Big \} \leq \frac{1}{\epsilon} E \Big [ \frac{1}{2n} \sum_{1 \leq i \leq 2n} |\tilde g_i - g_i|^2 \Big | \tilde W^{(m)} \Big ] \\
= \frac{1}{\epsilon} \int_{\mathrm{supp}(X_i)} |\tilde h_m(x) - h(x)|^2\,Q_X(d x) \leq \frac{\delta}{2}~.
\end{multline}
Then,
\begin{align*}
P \Big \{ \frac{1}{2n} \sum_{1 \leq i \leq 2n} |\tilde h_i - h_i|^2 > \epsilon \Big \} & \leq P \Big \{ \frac{1}{2n} \sum_{1 \leq i \leq 2n} |\tilde h_i - h_i|^2 > \epsilon \Big | \tilde W^{(m)} \Big \} \\
& \hspace{3em} \times P \Big \{ \int_{\mathrm{supp}(X_i)} |\tilde h_m(x) - h(x)|^2\,Q_X(d x) \leq \frac{\epsilon \delta}{2} \Big \} \\
& \hspace{3em} + P \Big \{ \int_{\mathrm{supp}(X_i)} |\tilde h_m(x) - h(x)|^2\,Q_X(d x) > \frac{\epsilon \delta}{2} \Big \} \\
& \leq \frac{\delta}{2} \Big (1 - \frac{\delta}{2} \Big ) + \frac{\delta}{2} \leq \delta~,
\end{align*}
where the first inequality follows by definition, and the second inequality follows from \eqref{eq:inner} and \eqref{eq:outer}.

Next, note because $|a + b|^2 \leq 2(a^2 + b^2)$ for any $a, b \in \mathbf R$,
\begin{equation} \label{eq:h-diff}
\frac{1}{n} \sum_{1 \leq s \leq n} |h_{\pi^{\tilde h_m}(2s - 1)} - h_{\pi^{\tilde h_m}(2s)}|^2 \lesssim \frac{1}{n} \sum_{1 \leq s \leq n} |\tilde h_{\pi^{\tilde h_m}(2s - 1)} - \tilde h_{\pi^{\tilde h_m}(2s)}|^2 + \frac{1}{n} \sum_{1 \leq i \leq 2n} |\tilde h_i - h_i|^2~.
\end{equation}
Next, note
\begin{multline} \label{eq:hhat-diff} 
\frac{1}{n} \sum_{1 \leq s \leq n} |\tilde h_{\pi^{\tilde h_m}(2s - 1)} - \tilde h_{\pi^{\tilde h_m}(2s)}|^2 \lesssim \frac{1}{n} \max_{1 \leq i \leq 2n} |\tilde h_i|^2 \\
\lesssim \frac{1}{n} \max_{1 \leq i \leq 2n} |h_i|^2 + \frac{1}{n} \max_{1 \leq i \leq 2n} |\tilde h_i - h_i|^2 \lesssim \frac{1}{n} \max_{1 \leq i \leq 2n} |h_i|^2 + \frac{1}{n} \sum_{1 \leq i \leq 2n} |\tilde h_i - h_i|^2~.
\end{multline}
The conclusion then follows from  \eqref{eq:emp-l2}, \eqref{eq:h-diff}, \eqref{eq:hhat-diff}, the assumption that $E[h^2(X_i)] < \infty$ and an application of Lemma \ref{lem:max}.
\end{proof}

\subsection{Sufficient Conditions for Lipschitz Continuity} \label{sec:geometry}
Let $f$ denote the density function of $X$. Recall $C^{(r)}$ is the class of functions which are $r$th continuously differentiable. I impose the following assumption on $h$ and $f$. 

\begin{assumption} \label{as:geometry}
The function $h$ and density function $f$ satisfy the following conditions.
\begin{enumerate}[\rm (a)]
\item $h \in C^{(2)}$.
\item $\frac{\partial h(x)}{\partial x_p} \neq 0$ Lebesgue a.e.
\item $f \in C^{(2)}$.
\end{enumerate}
\end{assumption}

\begin{lemma}[Theorem 24.4 of \cite{munkres1991analysis}] \label{lem:manifold}
Let $O$ be open in $\mathbf R^p$ and $f: O \to \mathbf R$ be of class $C^{(r)}$ for $r \geq 1$. Let $M$ be the set of points $x$ for which $f(x) = 0$ and $N$ be the set of points $x$ for which $f(x) \geq 0$. Suppose $M$ is non-empty and $Df(x)$ has rank 1 at each point of $M$. Then $N$ is a $p$-manifold in $\mathbf R^p$ and $\partial N = M$.
\end{lemma}

\begin{lemma} \label{lem:p-1-manifold}
Suppose Assumption \ref{as:geometry}(a)--(b) hold. Then $M = \{x: h(x) = z\}$ is a $(p - 1)$-manifold in $\mathbf R^p$.
\end{lemma}

\begin{proof}[\sc Proof of Lemma \ref{lem:p-1-manifold}]
For each $x \in M$, I aim at providing a coordinate patch on $M$ about $x$. Indeed, by Assumption \ref{as:geometry}(a)--(b) and Theorem 9.2 (implicit function theorem) of \cite{munkres1991analysis}, there exists an open set $U$ containing $u = (x_1, \dots, x_{p - 1})$, an open ball $B(z)$ containing $z$ and an open set $O$ in $\mathbf R$ containing $x_p$, and a function $k: U \times B(z) \to \mathbf R^p$ of class $C^{(2)}$ such that $h(u, k(u, z')) = z'$ for all $u \in U$, $z' \in B(z)$ and $x \in O$. Moreover, $k(U \times B(z)) = O$. Define the coordinate patch $\alpha(u; z) = (u, k(u, z))$. The conclusion follows by Theorem 5-2 of \cite{spivak1965calculus}.
\end{proof}

Note $M = \{x: h(x) = z\}$ is a $(p - 1)$-manifold by Lemmas \ref{lem:manifold} and \ref{lem:p-1-manifold}. In what follows, I will need the definition of the integral of a function $g$ over the manifold $M$. In order to do so, note there exists a coordinate patch as $\{\alpha_j: U_j \subseteq \mathbf R^{p - 1} \to V_j \subseteq M, j \in \mathcal J\}$, where $\alpha_j(u) = \alpha_j(u, z)$, and each $\alpha_j(u) = (u, k_j(u))$ for some function $k_j: U \to \mathbf R$ which is of class $C^2$, as shown in the proof of Lemma \ref{lem:p-1-manifold}, and $\alpha_j(U_j) = V_j$. Next, there exists a partition of unity $\{\phi_i: i \in \mathcal I\}$ dominated by the $\{V_j: j \in \mathcal J\}$. Moreover, both $\mathcal I$ and $\mathcal J$ can be chosen to be countable, according to Section 25 of \cite{munkres1991analysis}. The integral of a scalar function $g$ over the manifold is written as
\[ \int_M g\,\dd V = \sum_{j \in \mathcal J} \sum_{i \in \mathcal I} \int_{U_j} [(g \phi_i) \circ \alpha_j] V(D \alpha_j)~, \]
where $V(A) = \sqrt{\mathrm{det}(A'A)}$ is the volume. I have
\[ D \alpha_j = \Big [ I_{p - 1} \quad \frac{\partial k_j(u, z)}{\partial u} \Big ]~, \]
so that
\[ V(D \alpha_j) = \sqrt{1 + \frac{\partial k_j(u, z)}{\partial u'}\frac{\partial k_j(u, z)}{\partial u}} = \frac{\|\nabla h(u, k_j(u, z))\|}{|D_p h(u, k_j(u, z))|}~, \]
where $D_p = \frac{\partial}{\partial x_p}$, by the implicit function theorem and matrix determinant lemma. Note on one hand, for each $j \in \mathcal J$, only a finite number of $\phi_i$ is positive, and on the other hand, $\{\phi_i: i \in \mathcal I\}$ is dominated by the coordinate patch, which means each $\phi_i$ is supported on a compact set inside a single $V_j$. As a result, the order of the above double sum can be interchanged.

By p.345 of \cite{bogachev2007measure}, the conditional expectation of a function $g$ on the manifold $M$ is defined as
\[ E[g(X) | M] = \lim_{t \to 0} \frac{E[g(X) I \{z \leq h(X) \leq z + t\}]}{P \{z \leq h(X) \leq z + t\}}~. \]

\begin{lemma} \label{lem:ce-manifold}
Suppose Assumption \ref{as:geometry}(a)--(c) hold. Then
\begin{equation} \label{eq:ce-M}
E[g(X) | M] = \frac{\displaystyle \int_M \frac{f g}{\|\nabla h\|}\,\dd V}{\displaystyle \int_M \frac{f}{\|\nabla h\|}\,\dd V}~.
\end{equation}
\end{lemma}

For a continuously differentiable function $h: \mathbf R^p \to \mathbf R$, $x \in \mathbf R^p$ is a critical point of $h$ if $\nabla h(x) = 0$, where $\nabla h(x)$ is the gradient of $h$ at $x$; otherwise $x$ is a regular point of $h$. A value $z$ is a critical value of $h$ if the set $\{x: h(x) = z\}$ contains at least one critical point; otherwise $z$ is a regular value of $h$.

\begin{proof}[\sc Proof of Lemma \ref{lem:ce-manifold}]
By L'Hospital's rule,
\[ E[g(X) | M] = \frac{\displaystyle \lim_{t \to 0} \frac{E[g(X) I \{z \leq h(X) \leq z + t\}]}{t}}{\displaystyle \lim_{t \to 0} \frac{P \{z \leq h(X) \leq z + t\}}{t}}~, \]
and the lemma follows from Lemma A.1 of \cite{chernozhukov2018sorted}. In particular, the denominator equals the one in \eqref{eq:ce-M} directly by the same lemma, while for the numerator I merely need to redefine the `density' function as $f g$ and the same proof goes through.
\end{proof}

\begin{lemma} \label{lem:partial-z}
Suppose Assumption \ref{as:geometry}(a)--(b) hold. Let $M = \{x: h(x) = z\}$, where $z$ is a regular value of $h$ on $\mathbf R^p$. Then for any $g \in C^{(2)}$, 
\begin{equation} \label{eq:partial-z}
\frac{\partial}{\partial z} \int_M g\,\dd V = \int_M \frac{D_p g}{D_p h}\,\dd V + \int_M g \frac{1}{\|\nabla h\|^2} \sum_{1 \leq i \leq p} \frac{D_i h D_{ip} h}{D_p h}\,\dd V - \int_M g \frac{D_{pp} h}{D_p^2 h}\,\dd V~.
\end{equation}
\end{lemma}

\begin{proof}[\sc Proof of Lemma \ref{lem:partial-z}]
To begin with, note
\begin{align}
\nonumber \frac{\partial}{\partial z} & \int_{U_j} [(g \phi_i) \circ \alpha_j] V(D \alpha_j) \\
\nonumber & = \int_{U_j} D_p (g \phi_i) \frac{\partial k_j(u, z)}{\partial z} \frac{\|\nabla h\|}{|D_p h|} \\
\label{eq:partial-z-j} & \hspace{1cm} + \int_{U_j} g \phi_i \frac{|D_p h|}{\|\nabla h\|} \frac{\partial k_j(u, z)}{\partial z} \frac{1}{D_p^4 h} \Big ( D_p^2 h \sum_{1 \leq i \leq p} D_i h D_{ip} h - D_p h D_{pp} h \sum_{1 \leq i \leq p} D_i^2 h  \Big )~,
\end{align}
where $D_{ij} h = \partial_i \partial_j h$ for any function $h \in C^{(2)}$. I have suppressed the arguments of $h$, being $(u, k_j(u, z))$. Note it is legitimate to pass differentiation inside the integral by the dominated convergence theorem. By the Implicit Function Theorem again,
\begin{equation} \label{eq:partial-k}
\frac{\partial k_j(u, z)}{\partial z} = \frac{1}{D_p h(u, k_j(u, z))}~.
\end{equation}
By Theorem 7.17 of \cite{rudin1976principles}, I know $\frac{\partial}{\partial z} \int_M g(x)\,\dd V$ is the sum over $i \in \mathcal I, j \in \mathcal J$ of the two terms in \eqref{eq:partial-z-j}. Using \eqref{eq:partial-k}, the sum of the first term is
\begin{align}
\nonumber & \sum_{j \in \mathcal J} \sum_{i \in \mathcal I} \int_{U_j} (\phi_i D_p g + g D_p \phi_i) \frac{1}{D_p h} \frac{\|\nabla h\|}{|D_p h|} \\
\nonumber & = \sum_j \int_{U_j} \frac{D_p g}{D_p h} V(D \alpha_j) \\
\label{eq:partial-z-1} & = \int_M \frac{D_p g}{D_p h}\,\dd V~,
\end{align}
because $\sum_{i \in \mathcal I} \phi_i = 1$ and hence $\sum_{i \in \mathcal I} D_p \phi_i = D_p \sum_{i \in \mathcal I} \phi_i = 0$. Again, the interchange of differentiation and sum is allowed because the sum is actually over a finite number of terms, by definition of a partition of unity. The sum of the second term is
\begin{align}
\nonumber & \sum_{j \in \mathcal J} \sum_{i \in \mathcal I} \int_{U_j} g \phi_i \frac{|D_p h|}{\|\nabla h\|} \frac{1}{D_p^4 h} \sum_{1 \leq i \leq p} (D_i h D_p h D_{ip} h - D_i^2 h D_{pp} h) \\
\nonumber & = \sum_{j \in \mathcal J} \int_{U_j} g \frac{D_p^2 h}{\|\nabla h\|^2} \frac{1}{D_p^4 h} \sum_{1 \leq i < p} (D_i h D_p h D_{ip} h - D_i^2 h D_{pp} h) V(D \alpha) \\
\nonumber & = \int_M  g \frac{1}{\|\nabla h\|^2 D_p^2h} \sum_{1 \leq i \leq p} (D_i h D_p h D_{ip} h - D_i^2 h D_{pp} h)\,\dd V \\
\label{eq:partial-z-2} & =  \int_M g \frac{1}{\|\nabla h\|^2 } \sum_{1 \leq i \leq p} \frac{D_i h D_{ip} h}{D_p h}\,\dd V - \int_M g \frac{D_{pp} h}{D_p^2 h}\,\dd V~.
\end{align} 
\eqref{eq:partial-z} now follows from \eqref{eq:partial-z-1} and \eqref{eq:partial-z-2}.
\end{proof}

\begin{theorem} \label{thm:ce-manifold-partial}
Suppose Assumption \ref{as:geometry} holds. If $z$ is a regular value of $h$, then
\begin{equation} \label{eq:ce-manifold-partial}
\frac{\partial}{\partial z} E[g(X) | M] = \frac{\displaystyle \int_M \frac{D_p(f g / D_{p} h)}{\|\nabla h\|}\,\dd V \int_M \frac{f}{\|\nabla h\|}\,\dd V - \int_M \frac{D_p(f / D_{p} h)}{\|\nabla h\|}\,\dd V \int_M \frac{f g}{\|\nabla h\|}\,\dd V}{\displaystyle \Big [ \int_M \frac{f}{\|\nabla h\|}\,\dd V \Big ]^2}~.
\end{equation}
\end{theorem}

\begin{proof}[\sc Proof of Theorem \ref{thm:ce-manifold-partial}]
To begin with, replace $g$ in Lemma \ref{lem:partial-z} with $\frac{f}{\|\nabla h\|}$. I then have
\begin{align}
\nonumber & \frac{\partial}{\partial z} \int_M \frac{f}{\|\nabla h\|}\,\dd V \\
\nonumber & = \int_M \frac{\|\nabla h\| \displaystyle D_p f - \frac{f \sum_{1 \leq i \leq p} D_i h D_{ip} h}{\|\nabla h\|}}{\|\nabla h\|^2 \displaystyle D_p h}\,\dd V \\
\nonumber & \hspace{3em} + \int_M \frac{f}{\|\nabla h\|^3} \sum_{1 \leq i \leq p} \frac{D_i h D_{ip} h}{D_p h}\,\dd V - \int_M \frac{f D_{pp} h}{\|\nabla h\| D_p^2 h}\,\dd V \\
\nonumber & = \int_M \frac{D_pf D_p h - f D_{pp} h}{\|\nabla h\| D_p^2 h}\,\dd V \\
\label{eq:ce-numer-partial} & = \int_M \frac{D_p(f / D_{p} h)}{\|\nabla h\|}\,\dd V~. 
\end{align}
By the same arguments, 
\begin{equation} \label{eq:ce-denom-partial}
\frac{\partial}{\partial z} \int_M \frac{f g}{\|\nabla h\|}\,\dd V = \int_M \frac{D_p(f g / D_{p} h)}{\|\nabla h\|}\,\dd V~.
\end{equation}
\eqref{eq:ce-manifold-partial} now follows from \eqref{eq:ce-numer-partial} and \eqref{eq:ce-denom-partial} together with the quotient rule.
\end{proof}

In general, by the Law of Iterated Expectation
\[ E[Y_i^r(d) | h(X) = z] = E[E[Y_i^r(d) | X] | h(X) = z]~. \]
Suppose $h$ and the density function of $X$, $f(X)$ satisfy the smoothness conditions in Assumption \ref{as:geometry}, the derivative
\[ \frac{\partial}{\partial z} E[g(X) | h(X) = z] \]
is given in Theorem \ref{thm:ce-manifold-partial}, where $g(x) = E[Y_i^r(d) | X = x]$ for $r = 1, 2$ and $d = 0, 1$. In particular, it is equal to
\begin{align}
\nonumber & E \Big [\frac{D_p g}{D_p h} + \frac{g D_p f}{f D_p h} - \frac{g D_{pp} h}{D_p^2 h} \Big | h(X) = z \Big] - E \Big [ \frac{D_p f}{f D_p h} - \frac{D_{pp} h}{D_p^2 h} \Big | h(X) = z \Big ] E \Big [ g \Big | h(X) = z \Big ] \\
\label{eq:partial-z-cov} & = E \Big [ \frac{D_p g}{D_p h} \Big | h(X) = z \Big ] + \mathrm{Cov} \Big [ \frac{D_p f}{f D_p h} - \frac{D_{pp} h}{D_p^2 h}, g \Big | h(X) = z \Big ]~.
\end{align}

\begin{lemma} \label{lem:lip-sufficient}
Each of the following conditions imply the boundedness of \eqref{eq:partial-z-cov}.
\begin{enumerate}
	\item $h$ is linear, $\|D_p g\|_\infty < \infty$, $\|g\|_\infty < \infty$ and $\|D_p (\ln f)\|_\infty < \infty$.
	\item $h$ is linear, $\sup_{z \in \mathbf R} | E[D_p g | h(X) = z] | < \infty$, $\sup_{z \in \mathbf R} | E[g^2 | h(X) = z] | < \infty$ and $\sup_{z \in \mathbf R} | E[D_p^2(\ln f) | h(X) = z] | < \infty$.
	\item $h$ includes linear and interaction terms, $\Big \| \frac{D_p g}{D_p h} \Big \|_\infty < \infty$, $\|g\|_\infty < \infty$ and $\Big \|\frac{D_p (\ln f)}{D_p h} \Big \|_\infty < \infty$.
\end{enumerate}
\end{lemma}

\begin{proof}[\sc Proof of Lemma \ref{lem:lip-sufficient}]
Follows from inspection.
\end{proof}

\section{Supplementary Theoretical Results} \label{sec:supplemental}
\subsection{Optimal Stratification for General Treated Fractions} \label{sec:tau}
The next theorem shows the infeasible optimal stratification has a similar structure to \eqref{eq:oracle} when $\tau \neq \frac{1}{2}$.
\begin{theorem} \label{thm:tau-oracle}
Assume $\tau = \frac{l}{k}$ where $l, k \in \mathbf N$, $0 < l < k$, and that the sample size is $kn$. Let $\pi^{\tau, g^\tau}$ be a permutation of $\{1, \dots, kn\}$ such that $g_{\pi^{\tau, g^\tau}(1)}^\tau \leq \dots \leq g_{\pi^{\tau, g^\tau}(kn)}^\tau$ for $g^\tau$ defined in \eqref{eq:tau-g}. Then, \eqref{eq:min-mse} is solved by
\begin{equation} \label{eq:tau-oracle}
\lambda^{\tau, g}(X^{(n)}) = \{\{\pi^{\tau, g^\tau}((s - 1)k + 1), \dots \pi^{\tau, g^\tau}(s k)\}: 1 \leq s \leq n\}~.
\end{equation}
\end{theorem}

\begin{proof}[\sc Proof of Theorem \ref{thm:tau-oracle}]
First, note
\[ \hat \theta_n = \frac{1}{kn} \sum_{1 \leq i \leq kn} \Big ( \frac{1}{\tau} Y_i(1) D_i - \frac{1}{1 - \tau} Y_i(0) (1 - D_i) \Big )~. \]
Next,
\[ \mse(\lambda | X^{(n)}) = \var_\lambda[\hat \theta_n | X^{(n)}]~, \]
so that I need only consider conditional variances of $\hat \theta$ given $X^{(n)}$ which can be decomposed as in \eqref{eq:var}. By repeating the arguments in the proof of Lemma \ref{lem:post}, for any $\lambda \in \Lambda_n$, the first term of the right-hand side of \eqref{eq:var} equals
\[ \frac{1}{k^2 n^2} \sum_{1 \leq i \leq kn} \Big ( \frac{\var[Y_i(1) | X_i]}{\tau} + \frac{\var[Y_i(0) | X_i]}{1 - \tau} \Big )~, \]
again identical across all $\lambda \in \Lambda_n$. Therefore, I need only consider
\[ \var_\lambda[E[\hat \theta_n | X^{(n)}, D^{(n)}] | X^{(n)}]~. \]
By repeating the arguments in the proof of Lemma \ref{lem:mixing}, a stratum of size $k l$ where $l > 1$ is a convex combination of stratifications with strata only of size $k$. In particular, let $\Lambda_n^k$ denote the set of all stratifications for which each stratum is of size $k$. Then, I have $\Lambda_n \in \mathrm{co}(\Lambda_n^k)$. I could therefore focus on the case where each stratum is of size $k$. For any stratification of the form $\lambda = \{\{\pi((s - 1)k + 1, \dots \pi(s k)\}: 1 \leq s \leq n\}$,
\[ \var_\lambda[E[\hat \theta_n | X^{(n)}, D^{(n)}] | X^{(n)}] \propto \sum_{1 \leq s \leq n} \sum_{(s - 1)k + 1 \leq j \leq sk} (g_{\pi(j)}^\tau - \bar g_s^\tau)^2~, \]
where $g_i^\tau$ is defined in \eqref{eq:tau-g} and
\[ \bar g_s^\tau = \frac{1}{k} \sum_{(s - 1)k + 1 \leq j \leq sk} g_{\pi(j)}^\tau~. \]
To see this, first note units are independent across strata, so that by repeating the arguments in the proof of Lemma \ref{lem:post},
\[ \var_\lambda[E[\hat \theta_n | X^{(n)}, D^{(n)}] | X^{(n)}] \propto \sum_{1 \leq s \leq n} \var_\lambda \Big [ \sum_{(s - 1)k + 1 \leq j \leq sk} g_{\pi(j)}^\tau D_{\pi(j)} \Big ]~. \]
Next,
\begin{align*}
& \var_\lambda \Big [ \sum_{(s - 1)k + 1 \leq j \leq sk} g_{\pi(j)}^\tau D_{\pi(j)} \Big ] \\
& = \frac{1}{\binom{k}{l}} \sum_{(s - 1)k + 1 \leq j_1 < \dots < j_l \leq sk} \Big ( \sum_{1 \leq \iota \leq l} g_{\pi(j_\iota)}^\tau - l \bar g_s^\tau \Big )^2 \\
& = \frac{l}{k} \sum_{(s - 1)k + 1 \leq j \leq sk} (g_{\pi(j)}^\tau - \bar g_s^\tau)^2 + \frac{1}{\binom{k}{l}} \sum_{(s - 1)k + 1 \leq j_1 < \dots < j_l \leq sk} \sum_{1 \leq \iota_1 \neq \iota_2 \leq l} (g_{\pi(j_{\iota_1})} - \bar g_s^\tau) (g_{\pi(j_{\iota_2})} - \bar g_s^\tau) \\
& = \frac{l}{k} \sum_{(s - 1)k + 1 \leq j \leq sk} (g_{\pi(j)}^\tau - \bar g_s^\tau)^2 + \frac{\binom{k - 2}{l - 2}}{\binom{k}{l}} \Big [ \Big ( \sum_{(s - 1)k + 1 \leq j \leq sk} g_{\pi(j)}^\tau - k \bar g_s^\tau \Big )^2 - \sum_{(s - 1)k + 1 \leq j \leq sk} (g_{\pi(j)}^\tau - \bar g_s^\tau)^2 \Big ] \\
& \propto \sum_{(s - 1)k + 1 \leq j \leq sk} (g_{\pi(j)}^\tau - \bar g_s^\tau)^2~,
\end{align*}
where the first equality holds by definition, the second holds by expanding the square, the third holds by accounting for cross product terms, and the fourth holds because the first term inside the square bracket on the fourth line is 0. The conclusion follows from similar arguments to those used in the proof of Lemma \ref{lem:hardy-littlewood}.
\end{proof}

\begin{remark} \rm
Researchers are sometimes faced with the the situation where both $l$ and $k$ in Theorem \ref{thm:tau-oracle} are large and they are mutually prime. For example, suppose there are 52 participants and 31 seats for treatment. In that case, because the treated fraction is close to $3 / 5$, our recommendation is to split the sample into 6 strata of size 10, 10, 10, 10, 10, 2, treat 6 of the 10 units in each of the first five strata, and 1 of the 2 units in the last stratum.
\end{remark}

The next theorem is the limiting counterpart to Theorem \ref{thm:tau-oracle}. It shows the asymptotic variance of $\hat \theta_n$ is minimized by choosing $h = g^{\tau}$ defined in \eqref{eq:tau-g}.

\begin{theorem} \label{thm:tau-min}
Suppose $\tau \in (0, 1)$. Let $h: \mathrm{supp}(X_i) \to \mathbf R$ be a measurable function, and $\pi^{\tau, h}$ be a permutation of $\{1, \dots, kn\}$ such that $h_{\pi^{\tau, h}(1)} \leq \dots \leq h_{\pi^{\tau, h}(kn)}$. Define
\begin{equation} \label{eq:tau-h}
\lambda^{\tau, h}(X^{(n)}) = \{\{\pi^{\tau, h}((s - 1)k + 1), \dots \pi^{\tau, h}(s k)\}: 1 \leq s \leq n\}~.
\end{equation}
Further define $\bar h_s^{\tau} = \frac{1}{k} \sum_{(s - 1)k + 1 \leq j \leq sk} h_{\pi^{\tau, h}(j)}$.
Suppose $h$ satisfies
\begin{enumerate}[\rm (a)] 
\item $0 < E[\var[Y_i(d) | h(X_i)]]$ for $d \in \{0, 1\}$.
\item $E[Y_i^r(d) | h(X_i) = z]$ is Lipschitz for $r = 1, 2$ and $d = 0, 1$.
\item $\displaystyle \frac{1}{n} \displaystyle \sum_{1 \leq s \leq n} \sum_{(s - 1)k + 1 \leq j \leq sk}|h_{\pi^{\tau, h}(j)} - \bar h_s^{\tau}|^2 \stackrel{P}{\to} 0$.
\end{enumerate}
Then,
\[ \varsigma_{\tau, g^\tau}^2 \leq \varsigma_{\tau, h}^2~, \]
for $\varsigma_{\tau, g^\tau}^2$ and $\varsigma_{\tau, h}^2$ defined in \eqref{eq:tau-limit} and $g^\tau$ defined in \eqref{eq:tau-g}. Moreover, the inequality is strict unless $E \Big [ \frac{Y_i(1)}{\tau} + \frac{Y_i(0)}{1 - \tau} \Big | h(X_i) \Big] = g^\tau(X_i)$ with probability one under $Q$.
\end{theorem}

\begin{proof}[\sc Proof of Theorem \ref{thm:tau-min}]
By the definition of $\varsigma_{\tau, h}^2$ in \eqref{eq:tau-limit}, minimizing $\varsigma_{\tau, h}^2$ with respect to $h$ is equivalent to maximizing
\[ E \Big [ \Big ( E \Big [ \frac{Y_i(1)}{\tau} + \frac{Y_i(0)}{1 - \tau} \Big | h(X_i) \Big] - \Big ( \frac{E[Y_i(1)]}{\tau} + \frac{E[Y_i(0)]}{1 - \tau} \Big) \Big )^2 \Big ]~. \]
Next, note
\begin{align}
\nonumber & E \Big [ \Big ( g^\tau(X_i) - E \Big [ \frac{Y_i(1)}{\tau} + \frac{Y_i(0)}{1 - \tau} \Big | h(X_i) \Big] \Big ) \Big ( E \Big [ \frac{Y_i(1)}{\tau} + \frac{Y_i(0)}{1 - \tau} \Big | h(X_i) \Big] - \Big ( \frac{E[Y_i(1)]}{\tau} + \frac{E[Y_i(0)]}{1 - \tau} \Big) \Big ) \Big ] \\
\nonumber & = E \Big [ E \Big [ g^\tau(X_i) - E \Big [ \frac{Y_i(1)}{\tau} + \frac{Y_i(0)}{1 - \tau} \Big | h(X_i) \Big]  \Big | h(X_i) \Big ] \\
\nonumber & \hspace{3em} \Big ( E \Big [ \frac{Y_i(1)}{\tau} + \frac{Y_i(0)}{1 - \tau} \Big | h(X_i) \Big] - \Big ( \frac{E[Y_i(1)]}{\tau} + \frac{E[Y_i(0)]}{1 - \tau} \Big) \Big ) \Big ] \\
\label{eq:cross} & = 0~,
\end{align}
where the second equality holds because
\[ E[g^\tau(X_i) | h(X_i)] = E \Big [ \frac{Y_i(1)}{\tau} + \frac{Y_i(0)}{1 - \tau} \Big | h(X_i) \Big] \]
by the law of iterated expectation. Therefore,
\begin{align*}
& E \Big [ \Big ( g^\tau(X_i) - \Big ( \frac{E[Y_i(1)]}{\tau} + \frac{E[Y_i(0)]}{1 - \tau} \Big) \Big )^2 \Big ] \\
& = E \Big [ \Big ( g^\tau(X_i) - E \Big [ \frac{Y_i(1)}{\tau} + \frac{Y_i(0)}{1 - \tau} \Big | h(X_i) \Big] + E \Big [ \frac{Y_i(1)}{\tau} + \frac{Y_i(0)}{1 - \tau} \Big | h(X_i) \Big] - \Big ( \frac{E[Y_i(1)]}{\tau} + \frac{E[Y_i(0)]}{1 - \tau} \Big) \Big )^2 \Big ] \\
& = E \Big [ \Big ( g^\tau(X_i) - E \Big [ \frac{Y_i(1)}{\tau} + \frac{Y_i(0)}{1 - \tau} \Big | h(X_i) \Big] \Big )^2 \Big ] \\
& \hspace{3em} + E \Big [ \Big ( E \Big [ \frac{Y_i(1)}{\tau} + \frac{Y_i(0)}{1 - \tau} \Big | h(X_i) \Big] - \Big ( \frac{E[Y_i(1)]}{\tau} + \frac{E[Y_i(0)]}{1 - \tau} \Big) \Big )^2 \Big ]~, \\
& \geq E \Big [ \Big ( E \Big [ \frac{Y_i(1)}{\tau} + \frac{Y_i(0)}{1 - \tau} \Big | h(X_i) \Big] - \Big ( \frac{E[Y_i(1)]}{\tau} + \frac{E[Y_i(0)]}{1 - \tau} \Big) \Big )^2 \Big ]~.
\end{align*}
where the second equality follows from \eqref{eq:cross} and the last inequality is strict except unless $E \Big [ \frac{Y_i(1)}{\tau} + \frac{Y_i(0)}{1 - \tau} \Big | h(X_i) \Big] = g^\tau(X_i)$ with probability one under $Q$.
\end{proof}

\subsection{Unequal Treated Fractions Across Subpopulations} \label{sec:unequal}
In this section, I consider settings in which treated fractions are allowed to vary across subpopulations. Let $1 \leq r \leq R$ index the subpopulations, where $R \geq 1$ is an integer. I assume the subpopulations are determined by the covariates according to a function $f: \mathrm{supp}(X_i) \to \{1, \dots, R\}$. I replace Assumption \ref{as:half} by
\begin{assumption} \label{as:vary}
For $1 \leq r \leq R$, exactly $\tau_r$ fraction of the units of each stratum in the $r$th subpopulation are treated.
\end{assumption}
Moreover, I assume treatment status is assigned independently across subpopulations. Under Assumption \eqref{as:vary}, $\hat \theta_n$ is generally inconsistent for $\theta$. In such settings researchers often use the estimator from the fully saturated regression in \cite{bugni2019inference}. For $1 \leq r \leq R$, let $n_r$ denote the total number of observations in the $r$th subpopulation. For $1 \leq r \leq R$ and $d \in \{0, 1\}$, define
\[ \hat \mu_{n, r}(1) = \frac{1}{n_r \tau_r} \sum_{i: f(X_i) = r} Y_i D_i \]
and
\[ \hat \mu_{n, r}(0) = \frac{1}{n_r (1 - \tau_r)} \sum_{i: f(X_i) = r} Y_i (1 - D_i)~. \]
The estimator for the ATE from the fully saturated regression is
\begin{equation} \label{eq:sat}
\hat \theta_n^{\rm sat} = \sum_{1 \leq r \leq R} \frac{n_r}{n} (\hat \mu_{n, r}(1) - \hat \mu_{n, r}(0))~.
\end{equation}
Note $\hat \theta_n^{\rm sat}$ and $\hat \theta_n$ coincide whenever $\tau_r \equiv \tau \in (0, 1)$. See \cite{bugni2018inference}, \cite{tabord-meehan2020stratification}, and \cite{bugni2019inference} for more details. By repeating the arguments used in the proof of Theorem \ref{thm:oracle} and Theorem \ref{thm:tau-oracle}, I could find the stratification that minimizes the conditional MSE of $\hat \theta_n^{\rm sat}$. The solution is as follows: I first calculate the stratification defined in \eqref{eq:tau-oracle} with $\tau$, $g$, and $X^{(n)}$ defined separately for each subpopulation, and then take the union of those stratifications. Moreover, the next theorem enables us to derive feasible procedures when treated fractions are allowed to vary across subpopulations. In particular, it reveals any plug-in estimator that satisfies the regularity conditions in Theorem \ref{thm:tau-min} leads to a stratification under which the asymptotic variance of $\hat \theta_n^{\rm sat}$ is no greater than and typically strictly less than that under procedures with each subpopulation as a stratum.

\begin{theorem} \label{thm:tabord}
Suppose the sample size is $n$. Define $N_r = \{i: f(X_i) = r\}$, $X^{N_r} = (X_i: i \in N_r)$, $n_r = |N_r|$, and $p(r) = Q \{f(X_i) = r\}$. Define $\lambda^{\rm large} = \displaystyle \bigcup_{1 \leq r \leq R} {N_r}$. For $1 \leq r \leq R$, let $\tau_r$ be the treated fraction in $N_r$. Define functions $h^r: \mathrm{supp}(X_i) \to \mathbf R$ for $1 \leq r \leq R$. Define $\lambda^{\rm small} = \displaystyle \bigcup_{1 \leq r \leq R} \lambda^{\tau_r, h^r}(X^{N_r})$, where $\lambda^{\tau_r, h^r}(X^{N_r})$ is defined in \eqref{eq:tau-h}. Suppose $Q$ satisfies Assumption \ref{as:moments} and the treatment assignment scheme satisfies Assumption \ref{as:vary}. Then, under $\lambda^{\rm large}$, for $\hat \theta_n^{\rm sat}$ defined in \eqref{eq:sat}, as $n \to \infty$,
\[ \sqrt n(\hat \theta_n^{\rm sat} - \theta(Q)) \stackrel{d}{\to} N(0, \varsigma_{\rm large}^2)~, \]
where
\begin{multline*}
\varsigma_{\rm large}^2 = \sum_{1 \leq r \leq R} p(r) \Big ( \frac{\mathrm{Var}[Y_i(1) | f(X_i) = r]}{\tau_r} + \frac{\mathrm{Var}[Y_i(0) | f(X_i) = r]}{1 - \tau_r} \Big ) \\
+ \sum_{1 \leq r \leq R} p(r) (E[Y_i(1) - Y_i(0) | f(X_i) = r] - E[Y_i(1) - Y_i(0)])^2~.
\end{multline*}
Suppose in addition that $h^r, 1 \leq r \leq R$ satisfy the assumption in Theorem \ref{thm:tau-min}, under $Q$ restricted to $\{x \in \mathrm{supp}(X_i): f(x) = r\}$. Then, under $\lambda^{\rm small}$, for $\hat \theta_n^{\rm sat}$ defined in \eqref{eq:sat}, as $n \to \infty$,
\[ \sqrt n(\hat \theta_n^{\rm sat} - \theta(Q)) \stackrel{d}{\to} N(0, \varsigma_{\rm small}^2)~, \]
where
\begin{multline*}
\varsigma_{\rm small}^2 = \sum_{1 \leq r \leq R} p(r) \Big ( \frac{\mathrm{Var}[Y_i(1) | f(X_i) = r]}{\tau_r} + \frac{\mathrm{Var}[Y_i(0) | f(X_i) = r]}{1 - \tau_r} \\
- \tau_r (1 - \tau_r) E \Big [ \Big ( E \Big [ \frac{Y_i(1)}{\tau_r} + \frac{Y_i(0)}{1 - \tau_r} \Big | h^r(X_i) \Big ] - E \Big [ \frac{Y_i(1)}{\tau_r} + \frac{Y_i(0)}{1 - \tau_r} \Big | f(X_i) = r \Big ] \Big )^2 \Big | f(X_i) = r \Big ] \Big ) \\
+ \sum_{1 \leq r \leq R} p(r) (E[Y_i(1) - Y_i(0) | f(X_i) = r] - E[Y_i(1) - Y_i(0)])^2~.
\end{multline*}
In addition, $\varsigma_{\rm small}^2 \leq \varsigma_{\rm large}^2$, where the inequality is strict unless for all $1 \leq r \leq R$,
\[ E \Big [ \frac{Y_i(1)}{\tau_r} + \frac{Y_i(0)}{1 - \tau_r} \Big | h^r(X_i) \Big] = E \Big [ \frac{Y_i(1)}{\tau_r} + \frac{Y_i(0)}{1 - \tau_r} \Big | f(X_i) = r \Big ] \]
with probability one under
\[ Q^r(A) = \frac{Q(A \cap \{f(X_i) = r\})}{Q\{f(X_i) = r\}}~. \]
Moreover, among all choices of $(h^r: 1 \leq r \leq R)$, $\varsigma_{\rm small}^2$ is minimized by setting $h^r = g^{\tau_r}$, where $g^{\tau_r}$ is defined in \eqref{eq:tau-g}.
\end{theorem}

\begin{proof}[\sc Proof of Theorem \ref{thm:tabord}]
The first convergence holds by Theorem 3.1 of \cite{bugni2019inference}. Define $\theta_r = E[Y_i(1) - Y_i(0) | f(X_i) = r]$. For the second convergence, note
\begin{equation} \label{eq:sat-conv}
\sqrt n (\hat \theta_n^{\rm sat} - \theta(Q)) = \sum_{1 \leq r \leq R} \Big ( \Big ( \frac{n_r}{n} \Big )^{1/2} \sqrt n_r (\hat \mu_{n, r}(1) - \hat \mu_{n, r}(0) - \theta_r) + \sqrt n \Big (\frac{n_r}{n} - p(r) \Big ) \theta_r \Big )~.
\end{equation}
Define
\begin{align*}
L_n^1 & = (\sqrt{n_r} (\hat \mu_{n, r}(1) - \hat \mu_{n, r}(0) - \theta_r): 1 \leq r \leq R) \\
L_n^2 & = \Big (\sqrt n \Big (\frac{n_r}{n} - p(r) \Big ): 1 \leq r \leq R \Big )~.
\end{align*}
It follows from the coupling argument in Lemma C.1 of \cite{bugni2019inference} that
\[ (L_n^1, L_n^2) = (L_n^{\ast 1}, L_n^2) + o_P(1) \]
where $L_n^{\ast 1} \indep L_n^2$ and $L_n^{\ast 1} \stackrel{d}{\to} N(0, \mathrm{diag}(\varsigma_{r, \rm small}^2: 1 \leq r \leq R))$ with
\begin{multline*}
\varsigma_{r, \rm small}^2 = \frac{E[(Y_i(1) - E[Y_i(1) | h^r(X_i)])^2 | f(X_i) = r]}{\tau_r} + \frac{E[(Y_i(0) - E[Y_i(0) | h^r(X_i)])^2 | f(X_i) = r]}{1 - \tau_r} \\
+ E [ (E[Y_i(1) - Y_i(0) | h^r(X_i)] - \theta_r)^2 | f(X_i) = r ]~.
\end{multline*}
Meanwhile, the central limit theorem implies
\[ L_n^2 \stackrel{d}{\to} N(0, \mathrm{diag}(p(r): 1 \leq r \leq R) - (p(r): 1 \leq r \leq R) (p(r): 1 \leq r \leq R)' )~. \]
In addition, it follows from the weak law of large numbers that $\frac{n_r}{n} \stackrel{P}{\to} p(r)$ for $1 \leq r \leq R$. By \eqref{eq:sat-conv} and Slutsky's lemma, I have that $\sqrt n (\hat \theta_n^{\rm sat} - \theta(Q))$ converges to a normal distribution with zero mean and variance
\[ \sum_{1 \leq r \leq R} p(r) \varsigma_{r, \rm small}^2 + \sum_{1 \leq r \leq R} p(r) (\theta_r - \theta(Q))^2~, \]
where I use the fact that
\[ \sum_{1 \leq r \leq R} p(r) \theta_r = \theta(Q)~. \]
The first result then follows from a similar calculation to that at the end of the proof of Lemma \ref{lem:tau-limit}. The last two results can be shown by similar arguments to those used in the proof of Theorem \ref{thm:tau-min}.
\end{proof}

\begin{remark} \rm \label{remark:tabord-detail}
\cite{tabord-meehan2020stratification} considers stratification trees, which leads to a small number of large strata, with different treated fractions in each stratum. Using results from Theorem \ref{thm:tabord}, it is straightforward to combine his procedure with procedures in this paper. The asymptotic variance of $\hat \theta_n^{\rm sat}$ under the combined procedure is no greater than and typically strictly less than that under his procedure alone. The combined procedure is as follows: First, implement the procedure in \cite{tabord-meehan2020stratification}, which produces a finite number of strata with a target treated fraction for each stratum. Second, I view each stratum as a subpopulation and calculate the stratification in \eqref{eq:tau-h} either with a fixed function $h$ or some plug-in estimate, with $\tau$ equal the target treated fraction. Finally, I take the union of these stratifications. The desired properties mentioned above now follow from Theorem \ref{thm:tabord}.
\end{remark}

\subsection{Formal Justification of Results with Attrition} \label{sec:attrition}
Let $A_i$ be a binary variable such that $A_i = 1$ if and only if the unit does not attrite. The difference-in-means estimator for the non-attritors is
\[ \hat \theta_n^{\rm A} = \frac{\frac{1}{n} \sum_{i: D_i = 1} A_i Y_i(1)}{\frac{1}{n} \sum_{i: D_i = 1} A_i} - \frac{\frac{1}{n} \sum_{i: D_i = 0} A_i Y_i(0)}{\frac{1}{n} \sum_{i: D_i = 0} A_i}~. \]
By repeating the arguments in the proof of Theorem S.1.5 of \cite{bai2021inference}, under the assumption that $((Y^{(n)}(0), Y^{(n)}(1), A^{(n)}) \indep D^{(n)} | X^{(n)}$, I have
\begin{align*}
\frac{1}{n} \sum_{i: D_i = 1} A_i Y_i(1) & \stackrel{P}{\to} E[A_i Y_i(1)] \\
\frac{1}{n} \sum_{i: D_i = 1} A_i & \stackrel{P}{\to} E[A_i] \\
\frac{1}{n} \sum_{i: D_i = 0} A_i Y_i(0) & \stackrel{P}{\to} E[A_i Y_i(0)] \\
\frac{1}{n} \sum_{i: D_i = 0} A_i & \stackrel{P}{\to} E[A_i]~.
\end{align*}
As a result,
\[ \hat \theta_n^A \stackrel{P}{\to} \frac{E[A_i(Y_i(1) - Y_i(0))]}{E[A_i]}~. \]
If $A_i \indep (Y_i(1) - Y_i(0))$, then the right hand side is $\theta(Q)$. \qed

\subsection{Nonnegativity of the Variance Estimator} \label{sec:nonneg}
In this subsection, I show $\hat \varsigma_{h, n}^2$ in \eqref{eq:se} is nonnegative. For convenience of notation, I suppress $h$ in the subscripts. I have
\begin{align*}
\hat \varsigma_n^2 & = \hat \sigma_n^2(1) + \hat \sigma_n^2(0) - \frac{1}{2} \hat \rho_n + \frac{1}{2} (\hat \mu_n(1) + \hat \mu_n(0))^2 \\
& = \frac{1}{n} \sum_{1 \leq i \leq 2n: D_i = 1} Y_i^2 - \hat \mu_n^2(1) + \frac{1}{n} \sum_{1 \leq i \leq 2n: D_i = 0} Y_i^2 - \hat \mu_n^2(0) - \frac{1}{2} \hat \rho_n + \frac{1}{2} (\hat \mu_n(1) + \hat \mu_n(0))^2 \\
& = \frac{1}{n} \sum_{1 \leq i \leq 2n} Y_i^2 - \frac{1}{2} \hat \rho_n - \frac{1}{2} (\hat \mu_n(1) - \hat \mu_n(0))^2 \\
& = \frac{1}{2n} \sum_{1 \leq i \leq 2n} Y_i^2 - \frac{1}{n} \sum_{1 \leq s \leq n} Y_{\pi(2s - 1)} Y_{\pi(2s)} + \frac{1}{2n} \sum_{1 \leq s \leq n} (Y_{\pi(2s - 1)} + Y_{\pi(2s)})^2 \\
& \hspace{3em} - \frac{1}{n} \sum_{1 \leq j \leq n / 2} (Y_{\pi(4j - 3)} + Y_{\pi(4j - 2)})(Y_{\pi(4j - 1)} + Y_{\pi(4j)}) \\
& \hspace{3em} - \frac{1}{2} \Big ( \frac{1}{n} \sum_{1 \leq s \leq n} (D_{\pi(2s - 1)} - D_{\pi(2s)}) (Y_{\pi(2s - 1)} - Y_{\pi(2s)}) \Big )^2 \\
& = \frac{1}{2} \Big ( \frac{1}{n} \sum_{1 \leq s \leq n} ((D_{\pi(2s - 1)} - D_{\pi(2s)})(Y_{\pi(2s - 1)} - Y_{\pi(2s)}))^2 \\
& \hspace{3em} - \Big ( \frac{1}{n} \sum_{1 \leq s \leq n} (D_{\pi(2s - 1)} - D_{\pi(2s)})(Y_{\pi(2s - 1)} - Y_{\pi(2s)}) \Big )^2 \Big ) \\
& \hspace{3em} + \frac{1}{2n} \sum_{1 \leq j \leq n / 2}(Y_{\pi(4j - 3)} + Y_{\pi(4j - 2)} - (Y_{\pi(4j - 1)} + Y_{\pi(4j)}) )^2~.
\end{align*}
In the last expression, the second term is obviously nonnegative, while the first term is nonnegative because it is one half the sample variance of $\{(D_{\pi(2s - 1)} - D_{\pi(2s)})(Y_{\pi(2s - 1)} - Y_{\pi(2s)}): 1 \leq s \leq n\}$. \qed

\subsection{Details of the Penalized Procedure} \label{sec:pen-detail}
In this section, I discuss the details of the penalized procedure. For $d \in \{0, 1\}$, let $\tilde \beta_m(d)$ denote the least-square estimators of the linear regression coefficients among the treated or untreated units in the pilot experiment:
\[ \tilde \beta_m(d) = \Big ( \sum_{1 \leq j \leq m: \tilde D_j = d} \tilde X_j \tilde X_j' \Big )^{-1} \sum_{1 \leq j \leq m: \tilde D_j = d} \tilde X_j \tilde Y_j~, \]
and let $\tilde \Omega_m(d)$ denote the variance estimators assuming homoskedasticity:
\[ \tilde \Omega_m(d) = \tilde \nu_m^2(d) \Big ( \sum_{1 \leq j \leq m: \tilde D_j = d} \tilde X_j \tilde X_j' \Big )^{-1}~, \]
where
\[ \tilde \nu_m^2(d) = \frac{\sum_{1 \leq j \leq m} (\tilde Y_j - \tilde X_j' \tilde \beta_m(d))^2 I \{\tilde D_j = d\}}{\sum_{1 \leq j \leq m} I \{\tilde D_j = d\}}~. \]
For $d^{\rm pen}$ defined in \eqref{eq:pen-metric}, let $\pi^{\rm pen}$ denote the solution to
\begin{equation} \label{eq:pen-min}
\min_{\pi \in \Pi} \sum_{1 \leq s \leq n} d^{\rm pen}(X_{\pi(2s - 1)}, X_{\pi(2s)})~.
\end{equation}
Units are then paired to solve \eqref{eq:pen-min}, so the stratification is given by
\begin{equation} \label{eq:pen-strat}
\lambda^{\rm pen}(X^{(n)}) = \{\{\pi^{\rm pen}(2s - 1), \pi^{\rm pen}(2s)\}: 1 \leq s \leq n\}~.
\end{equation}
I start with a further justification for \eqref{eq:pen-strat} by discussing its optimality in a Bayesian framework, in the sense that it minimizes the integrated risk in a Bayesian framework with a diffuse normal prior, where the conditional expectations of potential outcomes are linear. With some abuse of notation, denote the conditional MSE in \eqref{eq:min-mse} by $\mse(\lambda | g, X^{(n)})$, where I make explicit the dependence on $g$. Suppose I have a prior distribution of $g$, denoted by $F(d g)$. Let $Q_X^n(d x^{(n)})$ denote the distribution of $X^{(n)}$ and $Q_{\tilde W}^m(d \tilde w^{(m)})$ denote the distribution of $\tilde W^{(m)}$. Consider the solution to following problem of minimizing the integrated risk across all measurable functions of the form $u: (\tilde w^{(m)}, x^{(n)}) \mapsto \lambda \in \Lambda_n$:
\begin{equation} \label{eq:bayes}
\min_u \int\!\!\!\int\!\!\!\int \mse(u(\tilde w^{(m)}, x^{(n)}) | g, x^{(n)}) Q_X^n(d x^{(n)}) Q_{\tilde W}^m(d \tilde w^{(m)}) F(d g)~.
\end{equation}

I focus on the special case under which and $Y_i(d) \sim N(X_i' \beta(d), \sigma^2)$ for $d \in \{0, 1\}$. Note the potential outcomes are homoskedastic conditional on the covariates. Define $\beta = \beta(1) + \beta(0)$, and I have $g(x) = x' \beta$. As before, I suppose $\tilde W^{(m)} = ((\tilde Y_j, \tilde X_j', \tilde D_j)': 1 \leq j \leq m)$ is available from a pilot experiment. Suppose the prior on $\beta(d)$ is $G_d \stackrel{d}{=} N(\eta(d), \Omega(d))$ for $d \in \{0, 1\}$, being independent across $d \in \{0, 1\}$. The prior distribution of $\beta$ is then $G(d \beta) \stackrel{d}{=} N(\eta(1) + \eta(0), \Omega(1) + \Omega(0))$. I could show the posterior distribution of $\beta(d)$ conditional on $\tilde W^{(m)}$ is
\[ \bar G_d(d \beta | \tilde W^{(m)}) \stackrel{d}{=} N(\bar \eta, \bar \Omega)~, \]
where for $d \in \{0, 1\}$,
\begin{align*}
\bar \eta(d) & = \Big ( (\sigma^2)^{-1} \sum_{j: \tilde D_j = d} \tilde X_j \tilde X_j' + \Omega^{-1}(d) \Big )^{-1} \Big ( (\sigma^2)^{-1} \sum_{j: \tilde D_j = d} \tilde X_j \tilde Y_j + \Omega^{-1}(d) \eta(d) \Big )\\
\bar \Omega(d) & = \Big ( (\sigma^2)^{-1} \sum_{j: \tilde D_j = d} \tilde X_j \tilde X_j' + \Omega^{-1}(d) \Big )^{-1}~.
\end{align*}
Define $\bar \eta = \bar \eta(1) + \bar \eta(0)$ and $\bar \Omega = \bar \Omega(1) + \bar \Omega(0)$. The posterior distribution for $\beta$ is
\[ \bar G(d \beta | \tilde W^{(m)}) \stackrel{d}{=} (\bar \eta, \bar \Omega)~, \]
because $G_d(d \beta)$'s are independent across $d \in \{0, 1\}$.

The next lemma provides the solution to the Bayesian problem in \eqref{eq:bayes}, where the choice set is over all measurable functions $u: (\tilde w^{(m)}, x^{(n)}) \mapsto \lambda \in \Lambda_n$.

\begin{lemma} \label{lem:pen}
The solution to \eqref{eq:bayes} maps each $(\tilde w^{(m)}, x^{(n)})$ to $\lambda = \{\{\pi(2s - 1), \pi(2s)\}: 1 \leq s \leq n/2\}$, where $\pi$ solves
\[ \min_{\pi \in \Pi_n} \sum_{1 \leq s \leq n} \bar d \Big ( x_{\pi(2s-1)}, x_{\pi(2s)} \Big )~, \]
where
\begin{equation} \label{eq:pen-gen}
\bar d(x_1, x_2) = (x_1' \bar \eta - x_2' \bar \eta)^2 + (x_1 - x_2)' \bar \Omega (x_1 - x_2)~.
\end{equation}
\end{lemma}

\begin{proof}[\sc Proof]
First note by similar calculations to those leading to Lemma \ref{lem:post}, \eqref{eq:bayes} is equivalent to
\begin{equation} \label{eq:bayes-linear}
\min_u \int\!\!\!\int\!\!\!\int L(u(\tilde w^{(m)}, x^{(n)}) | \beta, x^{(n)}) Q_X^n(d x^{(n)}) Q_{\tilde W}^m(d \tilde w^{(m)}) G(d \beta)~,
\end{equation}
where
\[ L(u | \beta, x^{(n)}) = \beta' (x^{(n)})' \var_\lambda[D^{(n)} | X^{(n)}] x^{(n)} \beta~. \]
Next, note I could solve the problem pointwise for $\tilde w^{(m)}$ and $x^{(n)}$ because \eqref{eq:bayes-linear} is equivalent to
\begin{equation} \label{eq:bayes-ex-post}
\min_u \bar R(u | \tilde W^{(m)})~,
\end{equation}
where
\[ \bar R(u | \tilde W^{(m)}) = \int L(u(\tilde W^{(m)}, x^{(n)}) | \beta, x^{(n)}) \bar G(d \beta | \tilde W^{(m)})~. \]

To solve \eqref{eq:bayes-ex-post}, first note because $\bar R(u | \tilde W^{(m)})$ is linear in $u$, by Lemma \ref{lem:mixing}, it is solved by a matched-pair design. Next, 
\[ \bar R(u | \tilde W^{(m)}) = \sum_{1 \leq s \leq n} ((x_{\pi(2s - 1)}' \bar \eta - x_{\pi(2s)}' \bar \eta)^2 + (x_{\pi(2s - 1)} - x_{\pi(2s)})' \bar \Omega (x_{\pi(2s - 1)} - x_{\pi(2s)}))~. \]
As a result, minimizing it is equivalent to minimizing the sum of the distances defined in \eqref{eq:pen-gen}.
\end{proof}

Lemma \ref{lem:pen} shows the solution to \eqref{eq:bayes} is not to naïvely pair units according to the values of $X_i' \bar \eta$, where $\bar \eta$ is posterior mean of $\beta$. Instead, the solution to \eqref{eq:bayes} depends not only on the posterior mean of $\beta$, but also on the posterior variance of it. The posterior variance serves as a penalty to matching naïvely on the posterior mean of $\beta$: the larger the variance, the more it penalizes matching on the posterior mean. 

Finally, I make the prior irrelevant. For this purpose, suppose $\Omega = c I$ where $I$ is an identity matrix. I let the constant $c \to \infty$, so that the prior diverges to a diffuse (uninformative) one. Then, $\bar \eta(d)$ converges to $\tilde \beta_m(d)$ and $\bar \Omega(d)$ converges to $\tilde \Omega_m(d)$. The metric then \eqref{eq:pen-gen} converges to the metric defined in \eqref{eq:pen-metric}.

In practice, \eqref{eq:pen-min} can be solved as follows. Define $R_m$ as the result of the following Cholesky decomposition:
\[ R_m' R_m = \tilde \beta_m \tilde \beta_m' + \tilde \Omega_m~, \]
and define
\[ Z_i = R_m X_i~. \]
To see \eqref{eq:pen-min} is equivalent to
\begin{equation} \label{eq:pen-z}
\min_{\pi \in \Pi_n} \frac{1}{n} \sum_{1 \leq s \leq n} \|Z_{\pi(2s - 1)} - Z_{\pi(2s)}\|^2~,
\end{equation}
note
\[ d^{\rm pen}(x_1, x_2) = (x_1 - x_2)' \tilde \beta_m \tilde \beta_m' (x_1 - x_2) + (x_1 - x_2)' \tilde \Omega_m (x_1 - x_2) = (x_1 - x_2)' R_m' R_m (x_1 - x_2)~. \]
The penalized stratification pairs units to minimize the sum of distances in terms of $Z_i$ within pairs. When $\mathrm{dim}(X_i)$ is not too large, the problem can be solved quickly by the package \texttt{nbpMatching} in \texttt{R}.

Because the penalized matched-pair design can be viewed as pairing to minimize the Euclidean distances of $Z$ as in \eqref{eq:pen-z}, inference can be implemented by ``pairing the pairs'' as in Section 4 of \cite{bai2021inference}. I refer the readers to that paper for details.

The next theorem establishes the behavior of the difference-in-means estimator under the penalized procedure as the sample sizes of the pilot and the main experiment both increase. Not surprisingly, as the sample size of the pilot experiment goes to infinity, the penalized procedure behaves similarly to pairing units according to $h$, where $h$ is a linear function. In particular, if the selection on observable assumption holds in the pilot data, then the conclusion of the next theorem holds with $\beta = \beta(1) + \beta(0)$, where $\beta(d) = E[X X']^{-1} E[X Y(d)]$ for $d \in \{0, 1\}$.

\begin{theorem} \label{thm:pen-large}
Suppose $Q$ satisfies Assumption \ref{as:moments}, $h(x) = x' \beta$ satisfies Assumption \ref{as:H} for some $\beta \in \mathbf R^{\mathrm{dim}(X_i)}$. Suppose $E[X X'] < \infty$. Further suppose $\tilde \beta_m \stackrel{P}{\to} \beta$ and $\tilde \Omega_m \stackrel{P}{\to} 0$ as $m \to \infty$. Then, under $\lambda^{\rm pen}$ defined in \eqref{eq:pen-strat}, as $m, n \to \infty$, $\sqrt n(\hat \theta_n - \theta(Q)) \stackrel{d}{\to} N(0, \varsigma_h^2)$ for $\varsigma_h^2$ in \eqref{eq:limit}. Furthermore, $(\hat \varsigma_n^{\rm pen})^2 \stackrel{P}{\to} \varsigma_h^2$.
\end{theorem}

\begin{proof}
I only prove the convergence in distribution because the convergence of the standard error follows from similar arguments to those used in the proof of Theorem 4.3 in \cite{bai2021inference}. Define $\tilde h_m(x) = x' \tilde \beta_m$. Let $\|\tilde \Omega_m\|_{\rm op}$ denote the operator norm of $\tilde \Omega_m$. Note
\begin{align*}
& \frac{1}{n} \sum_{1 \leq s \leq n} ((X_{\pi^{\rm pen}(2s - 1)}' \tilde \beta_m - X_{\pi^{\rm pen}(2s)}' \tilde \beta_m)^2 + (X_{\pi^{\rm pen}(2s - 1)} - X_{\pi^{\rm pen}(2s)})' \tilde \Omega_m (X_{\pi^{\rm pen}(2s - 1)} - X_{\pi^{\rm pen}(2s)})) \\
& \leq \frac{1}{n} \sum_{1 \leq s \leq n} ((X_{\pi^{\tilde h_m}(2s - 1)}' \tilde \beta_m - X_{\pi^{\tilde h_m}(2s)}' \tilde \beta_m)^2 + (X_{\pi^{\tilde h_m}(2s - 1)} - X_{\pi^{\tilde h_m}(2s)})' \tilde \Omega_m (X_{\pi^{\tilde h_m}(2s - 1)} - X_{\pi^{\tilde h_m}(2s)})) \\
& = \frac{1}{n} \sum_{1 \leq s \leq n} ((X_{\pi^{\tilde h_m}(2s - 1)}' \tilde \beta_m - X_{\pi^{\tilde h_m}(2s)}' \tilde \beta_m)^2 + \|\tilde \Omega_m\|_{\rm op} |X_{\pi^{\tilde h_m}(2s - 1)} - X_{\pi^{\tilde h_m}(2s)}|^2) \\
& \leq \frac{1}{n} \sum_{1 \leq s \leq n} (X_{\pi^{\tilde h_m}(2s - 1)}' \tilde \beta_m - X_{\pi^{\tilde h_m}(2s)}' \tilde \beta_m)^2 + \|\tilde \Omega_m\|_{\rm op} \frac{2}{n} \sum_{1 \leq i \leq 2n} |X_i|^2 \\
& = \frac{1}{n} \sum_{1 \leq s \leq n} (X_{\pi^{\tilde h_m}(2s - 1)}' \tilde \beta_m - X_{\pi^{\tilde h_m}(2s)}' \tilde \beta_m)^2 + o_P(1)~,
\end{align*}
where the first inequality follows because $\pi^{\rm pen}$ solves \eqref{eq:pen-min}, the first equality follows from the definition of the operator norm, the second inequality follows from the fact that $|a + b|^2 \leq 2(|a|^2 + |a|^2)$ for $a, b \in \mathbf R^{\mathrm{dim}(X_i)}$, and the last equality follows from the assumptions that $E[X X'] < \infty$ and $\tilde \Omega_m \stackrel{P}{\to} 0$ and the weak law of large numbers. Because I also know
\[ \frac{1}{n} \sum_{1 \leq s \leq n} (X_{\pi^{\rm pen}(2s - 1)}' \tilde \beta_m - X_{\pi^{\rm pen}(2s - 1)}' \tilde \beta_m)^2 \geq \frac{1}{n} \sum_{1 \leq s \leq n} (X_{\pi^{\tilde h_m}(2s - 1)}' \tilde \beta_m - X_{\pi^{\tilde h_m}(2s - 1)}' \tilde \beta_m)^2~, \]
I have
\[ \frac{1}{n} \sum_{1 \leq s \leq n} (X_{\pi^{\rm pen}(2s - 1)}' \tilde \beta_m - X_{\pi^{\rm pen}(2s - 1)}' \tilde \beta_m)^2 = \frac{1}{n} \sum_{1 \leq s \leq n} (X_{\pi^{\tilde h_m}(2s - 1)}' \tilde \beta_m - X_{\pi^{\tilde h_m}(2s - 1)}' \tilde \beta_m)^2 + o_P(1)~. \]
The rest of the proof follows from similar arguments to those used in the proof of Lemma \ref{lem:hat}.
\end{proof}

\subsection{Inference with Pooled Data} \label{sec:pooled}
So far I have disregarded data from the pilot experiment for inference except when computing $\tilde g_m$. I end this section by describing a test that combines data from the pilot and the main experiments. Define
\[ \tilde \theta_m = \tilde \mu_m(1) - \tilde \mu_m(0)~, \]
where
\[ \tilde \mu_m(d) = \frac{\sum_{1 \leq j \leq m} \tilde Y_j I \{\tilde D_j = d\}}{\sum_{1 \leq j \leq m} I \{\tilde D_j = d\}} \]
for $d \in \{0, 1\}$. I define the new estimator for $\theta(Q)$ as
\[ \hat \theta_n^{\rm combined} = \frac{m}{m + 2n} \tilde \theta_m + \frac{2n}{m + 2n} \hat \theta_n~. \]
Let $\tilde \varsigma_{\mathrm{pilot}, m}^2$ denote the variance estimator of $\tilde \theta_m$ in the pilot experiment. I define the test as
\begin{equation} \label{eq:combined}
\phi_n^{\rm combined}(W^{(n)}, \tilde W^{(m)}) = I \{|T_n^{\rm combined}(W^{(n)}, \tilde W^{(m)})| > \Phi^{-1}(1 - \frac{\alpha}{2}) \}~,
\end{equation}
where
\[ T_n^{\rm combined}(W^{(n)}, \tilde W^{(m)}) = \frac{\sqrt{m + 2n}(\hat \theta_n^{\rm combined} - \theta_0)}{\sqrt{\frac{m}{m + 2n} \tilde \varsigma_{\mathrm{pilot}, m}^2 + \frac{2n}{m + 2n} 2 \hat \varsigma_{\tilde h_m, n}^2}}~, \]
and $\Phi^{-1}(1 - \frac{\alpha}{2})$ denotes the $(1 - \frac{\alpha}{2})$-th quantile of the standard normal distribution.

The following theorem shows the test defined in \eqref{eq:combined} is asymptotically exact as the sample sizes of both the pilot and the main experiments increase. The main additional requirement is as $m \to \infty$, $\sqrt m(\tilde \theta_m - \theta(Q))$ converges in distribution to a normal distribution whose variance is consistently estimable. The assumption is satisfied by many treatment assignment schemes, including i.i.d.\ coin flips and covariate-adaptive randomization. See \cite{bugni2018inference} and \cite{bugni2019inference} for details.

\begin{theorem} \label{thm:combined}
Suppose the treatment assignment scheme satisfies Assumption \ref{as:half}, $Q$ satisfies Assumptions \ref{as:moments}, $h$ satisfies Assumption \ref{as:H}, and $\tilde h_m$ satisfies Assumption \ref{as:l2}. Suppose in addition that as $m \to \infty$, $\sqrt m(\tilde \theta_m - \theta(Q)) \stackrel{d}{\to} N(0, \varsigma_{\rm pilot}^2)$, $\tilde \varsigma_{\mathrm{pilot}, m}^2 \stackrel{P}{\to} \varsigma_{\rm pilot}^2~$, and that as $m, n \to \infty$,
\[ \frac{m}{m + 2n} \to \nu \in [0, 1]~. \]
Then, under $\lambda^{\tilde g_m}(X^{(n)})$ for $h = \tilde g_m$, as $m, n \to \infty$,
\[ \frac{\sqrt{m + 2n}(\hat \theta_n^{\rm combined} - \theta(Q))}{\sqrt{\frac{m}{m + 2n} \tilde \varsigma_{\mathrm{pilot}, m}^2 + \frac{2n}{m + 2n} 2 \hat \varsigma_{\tilde h_m, n}^2}} \stackrel{d}{\to} N(0, 1)~. \]
Thus, for the problem of testing \eqref{eq:H0} at level $\alpha \in (0, 1)$, $\phi_n^{\rm combined}(W^{(n)}, \tilde W^{(m)})$ in \eqref{eq:combined} satisfies
\[ \lim_{m, n \to \infty} E[\phi_n^{\rm combined}(W^{(n)}, \tilde W^{(m)})] = \alpha~, \]
whenever $Q$ additionally satisfies the null hypothesis, i.e.\, $\theta(Q) = \theta_0$.
\end{theorem}

\begin{proof}
To begin with, note I need only establish as $m, n \to \infty$,
\begin{equation} \label{eq:combined-limit}
\sqrt{m + 2n}(\hat \theta_n^{\rm combined} - \theta(Q)) \stackrel{d}{\to} N(0, \nu \varsigma_{\rm pilot}^2 + (1 - \nu) 2 \varsigma_h^2)~,
\end{equation}
and the rest follows from Slutsky's lemma. I prove \eqref{eq:combined-limit} by contradiction. Suppose \eqref{eq:combined-limit} does not hold. Then, there exists a subsequence still denoted by $\{m, n\}$ for notational simplicity, along which as $m, n \to \infty$,
\begin{equation} \label{eq:combined-contra}
\sup_{t \in \mathbf R} \Big | \sqrt{m + 2n}(\hat \theta_n^{\rm combined} - \theta(Q)) - \Phi(z / \sqrt{\nu \varsigma_{\rm pilot}^2 + (1 - \nu) 2 \varsigma_h^2}) \Big | \to c~,
\end{equation}
where $c > 0$, and
\[ \frac{m}{m + 2n} \to \nu \in [0, 1]~. \]
Now consider this subsequence. Since the two convergences in the Lemma \ref{lem:hat} hold in probability, there exists a further subsequence along which they hold with probability one. By Theorem \ref{thm:limit}, I could see along this subsequence, as $m, n \to \infty$, with probability one for $\tilde W^{(m)}$,
\begin{equation} \label{eq:combined-conv}
\sup_{t \in \mathbf R} \Big | Q \{ \sqrt n(\hat \theta_n - \theta(Q)) \leq t | \tilde W^{(m)} \} - \Phi(z / \varsigma_h) \Big | \to 0~.
\end{equation}
Along the subsequence I construct, because $\frac{m}{m + 2n} \to \nu$, by \eqref{eq:combined-conv}, Slutsky's lemma, and Lemma \ref{lem:cond-convd},
\[ \sqrt{m + 2n}(\hat \theta_n^{\rm combined} - \theta(Q)) \stackrel{d}{\to} N(0, \nu \varsigma_{\rm pilot}^2 + (1 - \nu) 2 \varsigma_h^2)~, \]
which is a contradiction to \eqref{eq:combined-contra}. The theorem therefore holds.
\end{proof}

\subsection{Adjusted Standard Error With Four Units} \label{sec:four}
Still suppose the sample size is $2n$, in order to be consistent with the notation in the main text. Suppose units are matched into sets of four units according to a function $h \in \mathbf H$. Let $\pi^h$ be such that $h_{\pi^h(1)} \leq \dots \leq h_{\pi^h(2n)}$. The stratification is given by
\begin{equation} \label{eq:four}
\{\{\pi^h(4s - 3), \pi^h(4s - 2), \pi^h(4s - 1), \pi^h(4s)\}: 1 \leq s \leq n / 2\}~.
\end{equation}
By Lemma \ref{lem:tau-limit}, I still have $\sqrt n(\hat \theta_n - \theta(Q)) \stackrel{d}{\to} N(0, \varsigma_h^2)$, for $\varsigma_h^2$ in \eqref{eq:limit}. The variance estimator is
\begin{equation} \label{eq:four-se}
(\hat \varsigma_{h, n}^{\rm four})^2 = \hat \sigma_n^2(1) + \hat \sigma_n^2(0) - \frac{1}{2} \hat \rho_n^{\rm four} + \frac{1}{2} (\hat \mu_n(1) + \hat \mu_n(0))^2~,
\end{equation}
where
\begin{equation} \label{eq:four-correction}
\hat \rho_n^{\rm four} = \frac{2}{n} \sum_{1 \leq s \leq n / 2} \frac{1}{2} \sum_{i, j, k, l \in \lambda_s, i < j, k < l: D_i = D_j = 1, D_k = D_l = 0} (Y_i + Y_k)(Y_j + Y_l)~.
\end{equation}
To establish 
\[ \hat \rho_n^{\rm four} \stackrel{P}{\to} E[E[g(X_i) | h(X_i)]^2]~, \]
note
\begin{align*}
& E \Big [ \frac{1}{2} \sum_{i, j, k, l \in \lambda_s, i < j, k < l: D_i = D_j = 1, D_k = D_l = 0} (Y_i + Y_k)(Y_j + Y_l) | h^{(n)} \Big ] \\
& = \frac{1}{12} \sum_{i, j, k, l \in \{0, 1, 2, 3\}, i < j, k < l} (\mu_1(h_{\pi^h(4s - i)}) + \mu_0(h_{\pi^h(4s - k)}))(\mu_1(h_{\pi^h(4s - j)}) + \mu_0(h_{\pi^h(4s - l)})) \\
& \hspace{3em} + (\mu_1(h_{\pi^h(4s - i)}) + \mu_0(h_{\pi^h(4s - l)}))(\mu_1(h_{\pi^h(4s - j)}) + \mu_0(h_{\pi^h(4s - k)})) \\
& = \frac{1}{12} (g_h(h_{\pi^h(4j - 3)}) + g_h(h_{\pi^h(4j - 2)}))(g_h(h_{\pi^h(4j - 1)}) + g_h(h_{\pi^h(4j)})) \\
& \hspace{3em} + \frac{1}{12} (g_h(h_{\pi^h(4j - 3)}) + g_h(h_{\pi^h(4j - 1)}))(g_h(h_{\pi^h(4j - 2)}) + g_h(h_{\pi^h(4j)})) \\
& \hspace{3em} + \frac{1}{12} (g_h(h_{\pi^h(4j - 3)}) + g_h(h_{\pi^h(4j)}))(g_h(h_{\pi^h(4j - 2)}) + g_h(h_{\pi^h(4j - 1)}))~,
\end{align*}
where the coefficient $\frac{1}{12} = \frac{1}{2} \times \frac{1}{6}$ appears in the first equality because there are $\binom{4}{2} = 6$ ways to choose 2 units among 4 units to be treated. The consistency of $(\hat \varsigma_{h, n}^{\rm four})^2$ in \eqref{eq:four-se} then follows from similar arguments to those used in the proof of Theorem \ref{thm:se}. By repeating the arguments in Section \ref{sec:nonneg}, I can also show the variance estimator is nonnegative.

Finally, I discuss the variance estimator on p.100 of in \cite{athey2017econometrics}. Suppose $\lambda_s = \{i(s), j(s), k(s), l(s)\}$, $D_{i(s)} = D_{j(s)} = 1$, and $D_{k(s)} = D_{l(s)} = 0$. The variance estimator is constructed as
\[ \sum_{1 \leq s \leq n / 2} n \Big ( \frac{4}{2n} \Big )^2 \Big ( \frac{1}{4} (Y_{i(s)} - Y_{j(s)})^2 + \frac{1}{4} (Y_{k(s)} - Y_{l(s)})^2 \Big )~. \]
It follows from similar arguments to those used above that the variance estimator converges in probability to
\[ E[\var[Y_i(1) | h(X_i)]] + E[\var[Y_i(1) | h(X_i)]]~, \]
which is less than the asymptotic variance of $\hat \theta_n$, i.e., $\varsigma_h^2$ in \eqref{eq:limit}, and strictly so unless \eqref{eq:hom} holds. Therefore, unless \eqref{eq:hom} holds, the test in \cite{athey2017econometrics} fails to control size. The failure of size control for \cite{athey2017econometrics} also arises in settings with a fixed number of strata, as discussed in Section 5 of \cite{bugni2019inference}. \qed

\section{Additional Simulation Results} \label{sec:sims-supp}
This section contains the tables with the raw numbers for the main text and some additional simulation results.

Table \ref{table:sims-main-all-1-5}--\ref{table:sims-main-all-6-10} contain the raw numbers for Table \ref{table:sims-main} in the main text.

Table \ref{table:sims-others} contains the summary statistics for the following stratifications in addition to the ones in the main text:
\begin{enumerate}
	\item[(i'')] None-reg-int: No stratification with the estimator in \cite{lin2013agnostic}, i.e., the OLS estimator of the coefficient on $D$ in the linear regression of $Y$ on a constant, $D$, $X - \bar X_n$, and $D (X - \bar X_n)$, where $\bar X_n$ is the sample average of $X_i$'s, together with White's heteroskedasticity-robust standard error.
	\item[(j)] MS X2: Matched sets of four to minimize the sum of the Mahalanobis distances of $X_2$, namely, all covariates in the main regression specification except the baseline outcome.
	\item[(k)] MS pilot: Matched sets of four according to $\tilde g_m$ from the pilot.
	\item[(l)] MS pen: Matched sets of four to minimize the sum of the distances in \eqref{eq:pen-metric} of all covariates.
\end{enumerate}

Tables \ref{table:sims-others-1-5}--\ref{table:sims-others-6-10} contain the raw numbers for stratifications (j)--(l).

Table \ref{table:sims-ai} includes the size of the test in \cite{athey2017econometrics} for matched sets of four. \cite{athey2017econometrics} assume a finite-population setting, in which the potential outcomes and the covariates are fixed, and the parameter of interest is the sample ATE, defined as the average difference between the treated and untreated potential outcomes of all units in the sample. My paper instead focuses on the ATE and assumes that units are drawn from a superpopulation, so potential outcomes and covariates are random. I have shown in Section \ref{sec:four} that because of the differences in sampling frameworks, the test in \cite{athey2017econometrics} does not control size in my setting unless \eqref{eq:hom} holds, which only happens in Model 1. Therefore, in Model 1, the size of the test is around the nominal level of 5\%. In Models 2 and 3, the size of the test is larger than 5\%.

\begin{table}[ht]
\caption{MSEs, size, and standard errors for stratifications (a)--(i) across papers 1--5}
\begin{adjustbox}{max width=0.8\linewidth,center}
\begin{tabular}{cccccccccccccc} 
\hline\hline
& \multicolumn{3}{c}{}                                  & (a)    & (b)    & (c)     & (d)     & (e)      & (f)      & (g)      & (h)      & (i)  & (i')   \\
Paper              & Model                    & $\theta$       &            & MP X   & MS X   & MP base & MS base & MP X2 & MP pilot & MP pen & Origin &None & None-reg   \\ 
\cmidrule{1-14}
\multirow{10}{*}[-0.2em]{1} & \multirow{5}{*}{1} & \multirow{5}{*}{0}         &MSE& 0.00048 & 0.00070 & - & - & 0.00045 & 0.00066 & 0.00046 & - & 0.001 & 0.00099 \\ 
&&                          &size (adj/adj4)                            & 1.4 & 5.9 & - & - & 0.5 & 3.0 & 0.7 & - & 5.1 & 4.7 \\ 
&&                          &size (MPt)                            & 5.1 & - & - & - & 4.7 & 5.5 & 5.4 & - & - & - \\ 
&&                          &s.e. (adj/adj4)                            & 0.028 & 0.026 & - & - & 0.028 & 0.029 & 0.028 & - & 0.032 & 0.031 \\ 
&&                          &s.e. (MPt)                            & 0.021 & - & - & - & 0.021 & 0.025 & 0.021 & - & - & - \\  
\cmidrule{2-14}
& \multirow{5}{*}{2} & \multirow{5}{*}{0.033} &MSE& 0.00055 & 0.00073 & - & - & 0.00052 & 0.00071 & 0.00054 & - & 0.00092 & 0.00092 \\  
&&                          &size (adj/adj4)                            & 2.5 & 5.1 & - & - & 2.3 & 3.4 & 1.6 & - & 3.6 & 3.2\\ 
&&                          &size (MPt)                            & 2.7 & - & - & - & 2.2 & 3.7 & 1.5 & - & - & - \\  
&&                          &s.e. (adj/adj4)                      & 0.028 & 0.027 & - & - & 0.028 & 0.029 & 0.028 & - & 0.032 & 0.031 \\ 
&&                          &s.e. (MPt)                            & 0.028 & - & - & - & 0.028 & 0.029 & 0.028 & - & - & -  \\ 

\cmidrule{1-14}
\multirow{15}{*}[-0.5em]{2} & \multirow{5}{*}{1} & \multirow{5}{*}{0}  &MSE& 0.038 & 0.053 & 0.070 & 0.07 & 0.063 & 0.053 & 0.033 & - & 0.087  & 0.079\\ 
&&                          &size (adj/adj4)                            & 1.2 & 4.8 & 4.4 & 4.6 & 3.9 & 2.6 & 0.60 & - & 5.7 & 5.3 \\ 
&&                          &size (MPt)                            & 4.8 & - & 4.6 & - & 5.4 & 5.3 & 5.0 & - & - & - \\ 
&&                          &s.e. (adj/adj4)                            & 0.25 & 0.23 & 0.27 & 0.27 & 0.26 & 0.26 & 0.25 & - & 0.29 & 0.28 \\ 
&&                          &s.e. (MPt)                            & 0.19 & - & 0.27 & - & 0.24 & 0.22 & 0.19 & - & - & - \\ 

\cmidrule{2-14}
& \multirow{5}{*}{2} & \multirow{5}{*}{0.18} &MSE& 0.047 & 0.053 & 0.073 & 0.070 & 0.062 & 0.057 & 0.049 & - & 0.081 & 0.073 \\ 
&&                          &size (adj/adj4)                            & 1.7 & 3.6 & 4.9 & 4.4 & 3.9 & 3.3 & 2.0 & - & 4.9 & 5.2 \\ 
&&                          &size (MPt)                            & 1.7 & - & 4.7 & - & 3.7 & 2.7 & 1.9 & - & - & - \\ 
&&                          &s.e. (adj/adj4)                            & 0.25 & 0.24 & 0.27& 0.27 & 0.26 & 0.26 & 0.25 & - & 0.29 & 0.28 \\   
&&                          &s.e. (MPt)                            & 0.26 & - & 0.28 & - & 0.27 & 0.27 & 0.26 & - & - & - \\ 
\cmidrule{2-14}
& \multirow{5}{*}{3} & \multirow{5}{*}{0.012}    &MSE& 0.060 & 0.069 & 0.058 & 0.062 & 0.076 & 0.070 & 0.058 & - & 0.091 & 0.090 \\ 
&&                          &size (adj/adj4)                            & 2.5 & 5.2 & 4.6 & 4.8 & 4.7 & 3.4 & 2.5 & - & 4.9 & 4.8 \\ 
&&                          &size (MPt)                            & 2.5 & - & 4.2 & - & 3.5 & 3.4 & 2.4 & - & - & - \\ 
&&                          &s.e. (adj/adj4)                            & 0.28 & 0.27 & 0.25 & 0.25 & 0.29 & 0.28 & 0.28 & - & 0.3 & 0.3 \\  
&&                          &s.e. (MPt)                            & 0.28 & - & 0.25 & - & 0.29 & 0.29 & 0.28 & - & -  & - \\

\cmidrule{1-14}
\multirow{15}{*}[-0.5em]{3} & \multirow{5}{*}{1} & \multirow{5}{*}{0}         &MSE& 0.079 & 0.13 & 0.15 & 0.17 & 0.15 & 0.11 & 0.075 & 0.22 & 0.22 & 0.17 \\ 
&&                          &size (adj/adj4)                            & 0.3 & 5.2 & 4.8 & 4.9 & 2.9 & 2.7 & 0.4 & 5.2 & 4.4 & 4.5 \\ 
&&                          &size (MPt)                            & 5.2 & - & 4.9 & - & 4.1 & 5.3 & 4.0 & - & - & - \\ 
&&                          &s.e. (adj/adj4)                            & 0.38 & 0.35 & 0.41 & 0.41 & 0.42 & 0.38 & 0.38 & 0.47 & 0.47 & 0.41 \\
&&                          &s.e. (MPt)                            & 0.28 & - & 0.40 & - & 0.39 & 0.33 & 0.28 & - & - & - \\

\cmidrule{2-14}
& \multirow{5}{*}{2} & \multirow{5}{*}{0.41} &MSE& 0.11 & 0.15 & 0.16 & 0.16 & 0.16 & 0.13 & 0.11 & 0.20 & 0.21 & 0.16 \\
&&                          &size (adj/adj4)                            & 2.6 & 6.6 & 4.9 & 5.4 & 4.1 & 4.2 & 2.50 & 4 & 4.2  & 4.8 \\ 
&&                          &size (MPt)                            & 2.2 & - & 4.9 & - & 3.6 & 3.1 & 1.9 & - & -  & - \\ 
&&                          &s.e. (adj/adj4)                            & 0.38 & 0.36 & 0.40 & 0.40 & 0.43 & 0.39 & 0.38 & 0.46 & 0.46 & 0.41 \\  
&&                          &s.e. (MPt)                           & 0.39 & - & 0.40 & - & 0.44 & 0.40 & 0.39 & - & - & - \\ 

\cmidrule{2-14}
& \multirow{5}{*}{3} & \multirow{5}{*}{0.60}    &MSE& 0.084 & 0.10 & 0.098 & 0.10 & 0.11 & 0.092 & 0.078 & 0.13 & 0.13 & 0.11 \\ 
&&                          &size (adj/adj4)                            & 3.0 & 6.2 & 4.7 & 4.6 & 4.8 & 5.1 & 2.7 & 5.3 & 5 & 4.7 \\ 
&&                          &size (MPt)                            & 2.3 & - & 4.1 & - & 3.8 & 3.7 & 1.9 & - & - & - \\ 
&&                          &s.e. (adj/adj4)                            & 0.32 & 0.30 & 0.31 & 0.31 & 0.34 & 0.31 & 0.31 & 0.36 & 0.36 & 0.33 \\ 
&&                          &s.e. (MPt)                            & 0.33 & - & 0.32 & - & 0.35 & 0.33 & 0.33 & - & - & - \\

\cmidrule{1-14}
\multirow{15}{*}[-0.5em]{4} & \multirow{5}{*}{1} & \multirow{5}{*}{0}         &MSE& 0.045 & 0.069 & 0.089 & 0.084 & 0.086 & 0.058 & 0.042 & 0.14 & 0.15 & 0.14 \\ 
&&                          &size (adj/adj4)                            & 0.9 & 4.7 & 5.9 & 4.7 & 2.0 & 2.2 & 1.0 & 3.7 & 4.5 & 4.7 \\ 
&&                          &size (MPt)                            & 5.5 & - & 6 & - & 5.4 & 5.6 & 5.6 & - & - & - \\ 
&&                          &s.e. (adj/adj4)                            & 0.28 & 0.26 & 0.29 & 0.29 & 0.35 & 0.27 & 0.26 & 0.37 & 0.39 & 0.38 \\ 
&&                          &s.e. (MPt)                            & 0.21 & - & 0.29 & - & 0.29 & 0.23 & 0.20 & - & - & - \\

\cmidrule{2-14}
& \multirow{5}{*}{2} & \multirow{5}{*}{-1.31} &MSE& 0.058 & 0.072 & 0.075 & 0.078 & 0.089 & 0.064 & 0.052 & 0.14 & 0.15 & 0.14 \\
&&                          &size (adj/adj4)                            & 2.2 & 4.8 & 3.9 & 4.6 & 2.3 & 3.1 & 2.1 & 5.4 & 4.4 & 3.9 \\  %
&&                          &size (MPt)                            & 2.2 & - & 3.6 & - & 3.0 & 2.7 & 1.5 & - & - & - \\ 
&&                          &s.e. (adj/adj4)                            & 0.28 & 0.27 & 0.29 & 0.29 & 0.35 & 0.27 & 0.27 & 0.37 & 0.39 & 0.38 \\  
&&                          &s.e. (MPt)                            & 0.29 & - & 0.29 & - & 0.34 & 0.28 & 0.28 & - & - & - \\

\cmidrule{2-14}
& \multirow{5}{*}{3} & \multirow{5}{*}{-1.78}    &MSE& 0.075 & 0.086 & 0.070 & 0.070 & 0.13 & 0.072 & 0.062 & 0.18 & 0.17 & 0.17 \\ 
&&                          &size (adj/adj4)                            & 2.2 & 4.3 & 5 & 5.9 & 2.7 & 4.8 & 2.7 & 5.7 & 4.3 & 3.6 \\  %
&&                          &size (MPt)                            & 1.9 & - & 4.7 & - & 3.1 & 3.8 & 1.3 & - & - &-  \\ 
&&                          &s.e. (adj/adj4)                            & 0.32 & 0.30 & 0.27 & 0.27 & 0.40 & 0.28 & 0.27 & 0.41 & 0.43 & 0.42 \\  
&&                          &s.e. (MPt)                            & 0.32 & - & 0.27 & - & 0.39 & 0.30 & 0.30 & - & - & - \\ 

\cmidrule{1-14}
\multirow{15}{*}[-0.5em]{5} & \multirow{5}{*}{1} & \multirow{5}{*}{0}         &MSE& 0.55 & 0.82 & 0.84 & 1.02 & 1.03 & 0.79 & 0.59 & - & 1.24 & 1.25\\ 
&&                          &size (adj/adj4)                            & 1.4 & 5.2 & 3.4 & 6.1 & 6.0 & 2.9 & 1.8 & - & 6.9 & 5.4 \\ 
&&                          &size (MPt)                            & 4.6 & - & 5.6 & - & 5.4 & 5.1 & 5.5 & - & -  & - \\ 
&&                          &s.e. (adj/adj4)                            & 0.99 & 0.92 & 1.04 & 1.01 & 1.01 & 1.03 & 1.01 & - & 1.09 & 1.13 \\ 
&&                          &s.e. (MPt)                            & 0.75 & - & 0.90 & - & 1 & 0.88 & 0.74 & - & - & - \\
\cmidrule{2-14}
& \multirow{5}{*}{2} & \multirow{5}{*}{-0.35} &MSE& 0.82 & 0.90 & 1.07 & 1.12 & 1.17 & 1.01 & 0.90 & - & 1.26 & 1.28 \\
&&                          &size (adj/adj4)                            & 3.8 & 5.1 & 5.4 & 6.1 & 6.9 & 4.6 & 4.6 & - & 5.6 & 5.1 \\ 
&&                          &size (MPt)                            & 2.8 & - & 3.9 & - & 6.3 & 4.5 & 3.3 & - & - & - \\ 
&&                          &s.e. (adj/adj4)                            & 1.00 & 0.97 & 1.05 & 1.04 & 1.07 & 1.05 & 1.02 & - & 1.13 & 1.16 \\
&&                          &s.e. (MPt)                            & 1.07 & - & 1.11 & - & 1.07 & 1.08 & 1.07 & - & -  & - \\  

\cmidrule{2-14}
& \multirow{5}{*}{3} & \multirow{5}{*}{-0.54}    &MSE& 0.92 & 1.01 & 0.94 & 1.08 & 1.11 & 1.03 & 0.94 & - & 1.20 & 1.21 \\  
&&                          &size (adj/adj4)                            & 3.5 & 5.3 & 3.2 & 5 & 4.8 & 4.2 & 3.80 & - & 5.5 & 4.6 \\
&&                          &size (MPt)                            & 2.6 & - & 2.8 & - & 5.3 & 4.2 & 2.7 & - & - &- \\ 
&&                          &s.e. (adj/adj4)                           & 1.05 & 1.01& 1.04 & 1.02 & 1.06 & 1.05 & 1.04 & - & 1.09 & 1.13 \\ 
&&                          &s.e. (MPt)                           & 1.11 & - & 1.10 & - & 1.06 & 1.09 & 1.11 & - & - & - \\ 
\hline\hline
\end{tabular}
\end{adjustbox}
\label{table:sims-main-all-1-5}
\begin{tablenotes} \footnotesize
\item For each stratification, I report (1) the MSE, (2) the size of testing \eqref{eq:H0} for $\theta_0 = \theta$ at the 5\% level, in percentage, and (3) the average standard error. The tests used in this table are as follows: for matched-pair designs, the adjusted $t$-test with the variance estimator in \eqref{eq:se} (adj) and the test in \cite{imbens2015causal} (MPt); for matched sets of four, the adjusted $t$-test with the variance estimator in \eqref{eq:four-se} in the supplement (adj4); for the original stratifications, the test in (23) of \cite{bugni2018inference}; for no stratification, the two-sample $t$-test; for the regression-adjusted estimator, the $t$-test with White's heteroskedasticity-robust standard error. Rows are labeled according to the papers, models, and metrics. Columns are labeled according to the stratifications. The definitions of the stratifications can be found in the main text.
\end{tablenotes}
\end{table}

\begin{table}[ht]
\caption{MSEs, size, and standard errors for stratifications (a)--(i) across papers 6--10}
\begin{adjustbox}{max width=0.85\linewidth,center}
\begin{tabular}{cccccccccccccc} 
\hline\hline
& \multicolumn{3}{c}{}                                  & (a)    & (b)    & (c)     & (d)     & (e)      & (f)      & (g)      & (h)      & (i)  & (i')  \\
Paper              & Model                    & $\theta$     &              & MP X   & MS X   & MP base & MS base & MP X2 & MP pilot & MP pen & Origin &None  & None-reg \\ 

\cmidrule{1-14}
\multirow{10}{*}[-0.2em]{6} & \multirow{5}{*}{1} & \multirow{5}{*}{0}        &MSE & 0.000046 & 0.000072 & - & - & 0.000045 & 0.000072 & 0.000044 & 0.00012 & 0.00012 & 0.00011 \\ 
&&                          &size (adj/adj4)                            & 0.7 & 5.1 & - & - & 0.6 & 3.0 & 0.8 & 5.5 & 4.4 & 4.2 \\ 
&&                          &size (MPt)                            & 5.0 & - & - & - & 5.1 & 5.3 & 3.7 & - & - & - \\ 
&&                          &s.e. (adj/adj4)                            & 0.0094 & 0.0086 & - & - & 0.0094 & 0.0097 & 0.0094 & 0.011 & 0.011 & 0.011 \\ 
&&                          &s.e. (MPt)                            & 0.0068 & - & - & - & 0.0068 & 0.0084 & 0.0068 & - & - & - \\ 

\cmidrule{2-14}
& \multirow{5}{*}{2} & \multirow{5}{*}{0.01} &MSE& 0.000082 & 0.000098 & - & - & 0.000081 & 0.000094 & 0.000081 & 0.00013 & 0.00013 & 0.00013 \\
&&                          &size (adj/adj4)                            & 3.1 & 4.8 & - & - & 2.2 & 3.5 & 2.8 & 5.3 & 6.1  & 6.7 \\ 
&&                          &size (MPt)                            & 2.7 & - & - & - & 1.7 & 3.0 & 2.0 & - & - & - \\ 
&&                          &s.e. (adj/adj4)                            & 0.010 & 0.0098 & - & - & 0.010 & 0.010 & 0.010 & 0.011 & 0.011 & 0.011\\ 
&&                          &s.e. (MPt)                            & 0.011 & - & - & - & 0.011 & 0.011 & 0.011 & - & - & - \\ 

\cmidrule{1-14}
\multirow{10}{*}[-0.2em]{7} & \multirow{5}{*}{1} & \multirow{5}{*}{0}         &MSE& 0.014 & 0.024 & - & - & 0.013 & 0.020 & 0.015 & 0.033 & 0.031 & 0.029 \\  
&&                          &size (adj/adj4)                            & 0.3 & 5.3 & - & - & 0.5 & 2.6 & 0.7 & 4.6 & 4.2 & 3.8 \\  
&&                          &size (MPt)                            & 4.6 & - & - & - & 3.4 & 5.2 & 5.3 & - & - & - \\ 
&&                          &s.e. (adj/adj4)                            & 0.17 & 0.15 & - & - & 0.16 & 0.17 & 0.16 & 0.18 & 0.19 & 0.18 \\ 
&&                          &s.e. (MPt)                            & 0.12 & - & - & - & 0.12& 0.14 & 0.12 & - & - & - \\ 

\cmidrule{2-14}
& \multirow{5}{*}{2} & \multirow{5}{*}{0.21} &MSE& 0.024 & 0.028 & - & - & 0.022 & 0.028 & 0.025 & 0.033 & 0.036 & 0.033 \\
&&                          &size (adj/adj4)                            & 2.6 & 4.9 & - & - & 2.3 & 3.8 & 3.6 & 4.7 & 5 & 4.5 \\ 
&&                          &size (MPt)                            & 1.7 & - & - & - & 1.3 & 3.7 & 3.0 & - & - & - \\ 
&&                          &s.e. (adj/adj4)                            & 0.17 & 0.17 & - & - & 0.17 & 0.18 & 0.17 & 0.19 & 0.19 & 0.19 \\  
&&                          &s.e. (MPt)                            & 0.18 & - & - & - & 0.18 & 0.18 & 0.18 & - & - & - \\

\cmidrule{1-14}
\multirow{15}{*}[-0.5em]{8} & \multirow{5}{*}{1} & \multirow{5}{*}{0}         &MSE& 0.19 & 0.28 & 0.41 & 0.41 & 0.20 & 0.27 & 0.18 & - & 0.46 & 0.47 \\ 
&&                          &size (adj/adj4)                            & 0.5 & 4.5 & 5.7 & 4.8 & 0.9 & 2.1 & 1.5 & - & 5.4 & 5.9 \\ 
&&                          &size (MPt)                            & 6.4 & - & 5.6 & - & 4.2 & 4.6 & 5.4 & - & - & - \\ 
&&                          &s.e. (adj/adj4)                            & 0.59 & 0.54 & 0.63& 0.63 & 0.60 & 0.59 & 0.58 & - & 0.67 & 0.67 \\ 
&&                          &s.e. (MPt)                            & 0.42 & - & 0.63 & - & 0.45 & 0.51 & 0.42 & - & - & - \\ 

\cmidrule{2-14}
& \multirow{5}{*}{2} & \multirow{5}{*}{0.041} &MSE& 0.24 & 0.32 & 0.43 & 0.38 & 0.26 & 0.32 & 0.25 & - & 0.45 & 0.45 \\
&&                          &size (adj/adj4)                            & 1.9 & 5.5 & 5.4 & 4.7 & 2.3 & 3.1 & 2.0 & - & 4.6 & 4.9 \\ %
&&                          &size (MPt)                            & 2.7 & - & 5.7 & - & 3.2 & 3.1 & 3.0 & - & - & - \\ 
&&                          &s.e. (adj/adj4)                           & 0.59 & 0.56 & 0.63 & 0.63 & 0.61& 0.61 & 0.59 & - & 0.67 & 0.66 \\ 
&&                          &s.e. (MPt)                            & 0.56 & - & 0.63 & - & 0.58 & 0.60 & 0.57 & - & -  & - \\

\cmidrule{2-14}
& \multirow{5}{*}{3} & \multirow{5}{*}{1.07}    &MSE& 0.42 & 0.47 & 0.40 & 0.41 & 0.42 & 0.52 & 0.41 & - & 0.59 & 0.59\\ 
&&                          &size (adj/adj4)                            & 3.5 & 4.9 & 5.7 & 5.3 & 3.0 & 5.5 & 3.6 & - & 5.6 & 5.0 \\
&&                          &size (MPt)                            & 2.3 & - & 4.8 & - & 2.4 & 4.0 & 3.0 & - & - & - \\ 
&&                          &s.e. (adj/adj4)                            & 0.71 & 0.69 & 0.63 & 0.63 & 0.72 & 0.72 & 0.71 & - & 0.75 & 0.76 \\ 
&&                          &s.e. (MPt)                            & 0.76 & - & 0.64 & - & 0.76 & 0.76 & 0.75 & - & - & - \\  

\cmidrule{1-14}
\multirow{15}{*}[-0.5em]{9} & \multirow{5}{*}{1} & \multirow{5}{*}{0}         &MSE& 0.0042 & 0.0065 & 0.0082 & 0.0090 & 0.0044 & 0.0057 & 0.0038 & 0.011 & 0.0095 & 0.0094 \\ 
&&                          &size (adj/adj4)                            & 1.0 & 5.2 & 3.6 & 5.1 & 0.8 & 2.2 & 0.5 & 5.7 & 5.5 & 6.0 \\ 
&&                          &size (MPt)                            & 5.0 & - & 3.9 & - & 5.5 & 4.8 & 4.8 & - & - & - \\ 
&&                          &s.e. (adj/adj4)                        &    0.086 & 0.079 & 0.094 & 0.094 & 0.089 & 0.087 & 0.085 & 0.099 & 0.10 & 0.098 \\ 
&&                          &s.e. (MPt)                            & 0.063 & - & 0.094 & - & 0.065 & 0.075 & 0.062 & - & - & - \\ 

\cmidrule{2-14}
& \multirow{5}{*}{2} & \multirow{5}{*}{-0.10} &MSE& 0.0065 & 0.0077 & 0.0097 & 0.0089 & 0.0068 & 0.0075 & 0.0065 & 0.010 & 0.0094 & 0.0093 \\ 
&&                          &size (adj/adj4)                            & 2.9 & 6.1 & 5.8 & 5.2 & 3.3 & 4.3 & 3.0 & 7.3 & 5.9 & 5.7 \\ 
&&                          &size (MPt)                            & 2.2 & - & 6.2 & - & 2.6 & 3.6 & 2.9 & - & - & - \\ 
&&                          &s.e. (adj/adj4)                            & 0.09 & 0.087 & 0.096 & 0.096 & 0.091 & 0.092 & 0.09 & 0.097 & 0.098  & 0.097 \\  
&&                          &s.e. (MPt)                            & 0.094 & - & 0.096 & - & 0.094 & 0.095 & 0.093 & - & - & - \\  

\cmidrule{2-14}
& \multirow{5}{*}{3} & \multirow{5}{*}{-0.012}    & MSE& 0.0065 & 0.0073 & 0.0076 & 0.0077 & 0.0064 & 0.0068 & 0.0065 & 0.0080 & 0.0078 & 0.0079 \\  
&&                          &size (adj/adj4)                            & 4.8 & 6.4 & 5.6 & 6.2 & 4.3 & 5.4 & 4.8 & 6.4 & 5.2 & 5.7 \\ 
&&                          &size (MPt)                            & 3.4 & - & 6 & - & 2.2 & 3.0 & 2.3 & - & - & - \\ 
&&                          &s.e. (adj/adj4)                            & 0.083 & 0.081 & 0.085 & 0.085 & 0.083 & 0.083 & 0.083 & 0.085 & 0.086 & 0.087 \\
&&                          &s.e. (MPt)                            & 0.094 & - & 0.086 & - & 0.094 & 0.091 & 0.094 & - & -  & - \\

\cmidrule{1-14}
\multirow{15}{*}[-0.5em]{10} & \multirow{5}{*}{1} & \multirow{5}{*}{0}         &MSE& 38.26 & 44.42 & 44.68 & 51.34 & 48.29 & 37.41 & 35.58 & - & 56.29 & 52.54 \\ 
&&                          &size (adj/adj4)                            & 2.9 & 5.2 & 4.2 & 5.5 & 3.8 & 2.6 & 2.3 & - & 5.7 & 5.0 \\ 
&&                          &size (MPt)                            & 5.1 & - & 4.6 & - & 3.8 & 3.3 & 3.3 & - & - & - \\ 
&&                          &s.e. (adj/adj4)                            & 6.73 & 6.56 & 6.91 & 6.86 & 7.1 & 6.65 & 6.63 & - & 7.44 & 7.25 \\  
&&                          &s.e. (MPt)                            & 6.16 & - & 6.72 & - & 6.99 & 6.32 & 6.1 & - & - & -  \\ 

\cmidrule{2-14}
& \multirow{5}{*}{2} & \multirow{5}{*}{5.61} &MSE& 53.10 & 57.04 & 78.45& 86.84 & 76.40 & 57.85 & 53.09 & - & 91.23 & 84.21 \\ 
&&                          &size (adj/adj4)                            & 4.4 & 4.9 & 4.9 & 6.2 & 4.4 & 5.1 & 4.0 & - & 5.3 & 5.0 \\ 
&&                          &size (MPt)                            & 4.8 & - & 4.8 & - & 4.9 & 5.3 & 3.5 & - & - & - \\ 
&&                          &s.e. (adj/adj4)                            & 7.67 & 7.58 & 8.9 & 8.87 & 8.79 & 7.65 & 7.59 & - & 9.5 & 9.23 \\  
&&                          &s.e. (MPt)                            & 7.61 & - & 8.92 & - & 8.79 & 7.65 & 7.57 & - & - & - \\ 

\cmidrule{2-14}
& \multirow{5}{*}{3} & \multirow{5}{*}{1.91}    &MSE& 49.23 & 55.22 & 46.99 & 52.37 & 71.83 & 51.18 & 50.39 & - & 83.09 & 79.53 \\ 
&&                          &size (adj/adj4)                            & 3.4 & 4.2 & 4.2 & 6.3 & 4.2 & 3.9 & 3.4 & - & 5.5 & 5.4 \\ 
&&                          &size (MPt)                            & 3.4 & - & 4.8 & - & 4.2 & 3.4 & 3.8 & - & - & - \\ 
&&                          &s.e. (adj/adj4)                            & 7.76 & 7.59 & 7.21 & 7.16 & 8.5 & 7.69 & 7.68 & - & 8.69 & 8.63\\  
&&                          &s.e. (MPt)                            & 7.68 & - & 7.14 & - & 8.55 & 7.7 & 7.66 & - & - & - \\ 
\hline\hline
\end{tabular}
\end{adjustbox}
\label{table:sims-main-all-6-10}
\begin{tablenotes} \footnotesize
\item For each stratification, I report (1) the MSE, (2) the size of testing \eqref{eq:H0} for $\theta_0 = \theta$ at the 5\% level, in percentage, and (3) the average standard error. The tests used in this table are as follows: for matched-pair designs, the adjusted $t$-test with the variance estimator in \eqref{eq:se} (adj) and the test in \cite{imbens2015causal} (MPt); for matched sets of four, the adjusted $t$-test with the variance estimator in \eqref{eq:four-se} in the supplement (adj4); for the original stratifications, the test in (23) of \cite{bugni2018inference}; for no stratification, the two-sample $t$-test; for the regression-adjusted estimator, the $t$-test with White's heteroskedasticity-robust standard error. Rows are labeled according to the papers, models, and metrics. Columns are labeled according to the stratifications. The definitions of the stratifications can be found in the main text.
\end{tablenotes}
\end{table}

\begin{table}[ht]
\caption{Summary statistics for MSEs, size, and standard errors for additional methods not in the main text across all papers and models}
\begin{adjustbox}{max width=0.7\linewidth,center}
\begin{tabular}{ccccc}
\hline\hline
& Stratification & MSE (ratio vs.\ None) & size (\%) & s.e. (ratio vs.\ None) \\ 
\midrule
&  MS X2     & 0.798 & 5.207 & 0.870 \\   
& & [0.551, 0.938] & [3.200, 6.300] & [0.720, 0.968] \\ 
\addlinespace
& MS pilot & 0.749 & 5.256 & 0.828\\
& & [0.444, 0.939] & [3.900, 6.600] & [0.663, 0.946] \\ 
\addlinespace
& MS pen & 0.693 & 4.604 & 0.886 \\ 
& & [0.402, 0.966] & [3.100, 7.200] & [0.629, 0.989] \\ 
\addlinespace
& None-reg-int & 0.949 & 4.870 & 0.984 \\ 
& & [0.782, 1.014] &[3.300, 6.800] &  [0.890, 1.041] \\ 
\hline\hline    
\end{tabular}
\end{adjustbox}
\label{table:sims-others}
\begin{tablenotes} \footnotesize
\item For each stratification, I report summary statistics across all papers and models of (1) the ratio between the MSE under the particular stratification and the MSE under no stratification, (2) the size of testing \eqref{eq:H0} for $\theta_0 = \theta$ at the 5\% level, in percentage, and (3) the ratio between the average standard error under the particular stratification and the average standard error under no stratification. The tests used in this table are: for MS, the adjusted $t$-test with the variance estimator in \eqref{eq:four-se}; for the regression-adjusted estimator, the $t$-test with White's heteroskedasticity-robust standard error. For each metric, I show the average and $[\min, \max]$ across all papers and models. Rows are labeled according to the stratifications. Columns are labeled according to the metrics. The definitions of the stratifications can be found in the text.
\end{tablenotes}
\end{table}

\begin{table}[ht]
\caption{MSEs, size, and standard errors for stratifications (j)--(l) across papers 1--5}
\begin{adjustbox}{max width=0.6\linewidth,center}
\begin{tabular}{cccccccc} 
\hline\hline
& & & & (j) & (k) & (l) & (i'') \\
Paper              & Model                    & $\theta$        &           & MS X2   & MS pilot   & MS pen  & None-reg-int \\ 
\cmidrule{1-8}
\multirow{6}{*}[-0.2em]{1} & \multirow{3}{*}{1} & \multirow{3}{*}{0}      & MSE  & 0.00070 & 0.00075 & 0.00067 & 0.00099 \\ 
&& & size (adj4) & 6.3 & 6.1 & 3.1 & 4.7 \\ 
&& & s.e. (adj4)& 0.026 & 0.028 & 0.027 & 0.032 \\                         
\cmidrule{2-8}
& \multirow{3}{*}{2} & \multirow{3}{*}{0.033} &MSE& 0.00076 & 0.00081 & 0.00072 & 0.00092 \\  
&& &size (adj4)& 6.3 & 4.7 & 4.9 & 3.3 \\ 
&&&s.e. (adj4)& 0.027 & 0.029 & 0.027 & 0.031 \\ 

\cmidrule{1-8}
\multirow{9}{*}[-0.2em]{2} & \multirow{3}{*}{1} & \multirow{3}{*}{0}   &MSE& 0.063 & 0.057 & 0.048  & 0.079 \\  
&& & size (adj4)& 4.3 & 4.6 & 3.2 & 5.3 \\ 
&& & s.e. (adj4)& 0.26 & 0.25 & 0.25 & 0.28 \\              
\cmidrule{2-8}
& \multirow{3}{*}{2} & \multirow{3}{*}{0.18} &MSE& 0.067 & 0.067 & 0.052 & 0.073 \\ 
&& & size (adj4)& 4.8 & 6.0 & 3.3 & 5.3 \\ 
&& & s.e. (adj4)& 0.26 & 0.26 & 0.25 & 0.28 \\ 
\cmidrule{2-8}
& \multirow{3}{*}{3} & \multirow{3}{*}{0.012} &MSE& 0.076 & 0.074 & 0.069 & 0.090 \\ 
&& & size (adj4) & 4.3 & 3.9 & 4.4 & 4.7 \\
&& & s.e. (adj4) & 0.28 & 0.28 & 0.28 & 0.30 \\ 

\cmidrule{1-8}
\multirow{9}{*}[-0.2em]{3} & \multirow{3}{*}{1} & \multirow{3}{*}{0}   &MSE& 0.17 & 0.14 & 0.13 & 0.17 \\  
&&& size (adj4)& 4.6 & 5.1 & 3.9 & 4.3 \\
&&& s.e. (adj4)& 0.41 & 0.37 & 0.37 & 0.42 \\         
\cmidrule{2-8}
& \multirow{3}{*}{2} & \multirow{3}{*}{0.41} &MSE& 0.17 & 0.16 & 0.13 & 0.16 \\ 
&&& size (adj4)& 4.7 & 6.2 & 4.3 & 4.6 \\ 
&&& s.e. (adj4)& 0.42 & 0.38 & 0.37 & 0.41 \\ 
\cmidrule{2-8}
& \multirow{3}{*}{3} & \multirow{3}{*}{0.60} &MSE& 0.12 & 0.099 & 0.090 & 0.11 \\ 
&&& size (adj4)& 5.4 & 5.9 & 4.5 & 5.1 \\ 
&&&s.e. (adj4) & 0.33 & 0.31 & 0.31 & 0.33 \\ 

\cmidrule{1-8}
\multirow{9}{*}[-0.2em]{4} & \multirow{3}{*}{1} & \multirow{3}{*}{0}   &MSE& 0.11 & 0.067 & 0.063 & 0.14 \\ 
&&& size (adj4)& 5.2 & 4.6 & 3.3 & 4.7 \\  %
&&& s.e. (adj4)& 0.33 & 0.26 & 0.27 & 0.38 \\        
\cmidrule{2-8}
& \multirow{3}{*}{2} & \multirow{3}{*}{-1.31} &MSE& 0.11 & 0.070 & 0.065 & 0.14 \\ 
&&& size (adj4) & 5.5 & 3.9 & 3.4 & 3.9 \\ 
&&& s.e. (adj4)& 0.34 & 0.27 & 0.27  & 0.38 \\ 
\cmidrule{2-8}
& \multirow{3}{*}{2} & \multirow{3}{*}{-1.78} &MSE& 0.15 & 0.077 & 0.070 & 0.17 \\
&&& size (adj4) & 4.8 & 5.0 & 3.2 & 3.6 \\  %
&&& s.e. (adj4) & 0.38 & 0.28 & 0.30 & 0.42 \\ 

\cmidrule{1-8}
\multirow{9}{*}[-0.2em]{5} & \multirow{3}{*}{1} & \multirow{3}{*}{0}   &MSE& 1.01 & 0.93 & 0.99 & 1.25 \\ 
&&& size (adj4) & 5.3 & 4.9 & 6.9 & 5.2 \\ 
&&& s.e. (adj4)& 1.01 & 0.99 & 0.96 & 1.13 \\       
\cmidrule{2-8}
& \multirow{3}{*}{2} & \multirow{3}{*}{-0.35} &MSE& 1.11 & 1.08 & 1.00 & 1.28 \\ 
&&& size (adj4)& 6.1 & 6.3 & 4.4 & 5.0 \\ 
&&& s.e. (adj4)& 1.07 & 1.03 & 0.99 & 1.17 \\ 
\cmidrule{2-8}
& \multirow{3}{*}{3} & \multirow{3}{*}{-0.54} &MSE& 1.07 & 1.03 & 1.12 & 1.21 \\ 
&&& size (adj4)& 5.6 & 4.7 & 6.3 & 4.5 \\ 
&&& s.e. (adj4)& 1.06 & 1.04 & 1.02  & 1.13 \\

\hline\hline
\end{tabular}
\end{adjustbox}
\label{table:sims-others-1-5}
\begin{tablenotes} \footnotesize
\item For each stratification, I report (1) the MSE, (2) the size of testing \eqref{eq:H0} for $\theta_0 = \theta$ at the 5\% level, in percentage, and (3) the average standard error. The tests used in this table are: for MS, the adjusted $t$-test with the variance estimator in \eqref{eq:four-se}; for the regression-adjusted estimator, the $t$-test with White's heteroskedasticity-robust standard error. Rows are labeled according to the papers and models. Columns are labeled according to the stratifications. The definitions of the stratifications can be found in the text.
\end{tablenotes}
\end{table}

\begin{table}[ht]
\caption{MSEs, size, and standard errors for stratifications (j)--(l) across papers 6--10}
\begin{adjustbox}{max width=0.6\linewidth,center}
\begin{tabular}{cccccccc} 
\hline\hline
& & & & (j) & (k) & (l) & (i'') \\
Paper              & Model                    & $\theta$       &            & MS X2   & MS pilot   & MS pen & None-reg-int \\ 

\cmidrule{1-8}
\multirow{6}{*}[-0.2em]{6} & \multirow{3}{*}{1} & \multirow{3}{*}{0}  &  MSE     & 0.000068 & 0.000093 & 0.000079 & 0.00011 \\  
&&&size (adj4)& 3.2 & 5.7 & 4.7 & 4.2 \\ 
&&& s.e. (adj4)& 0.0086 & 0.0093 & 0.00091 & 0.011 \\           

\cmidrule{2-8}
& \multirow{3}{*}{2} & \multirow{3}{*}{0.010} &MSE& 0.000096 & 0.00012 & 0.00010 & 0.00013 \\ 
&&&size (adj4)& 5.4 & 6.6 & 5.4 & 6.8 \\ 
&&&s.e. (adj4)& 0.0098 & 0.010 & 0.0099 & 0.011 \\ 

\cmidrule{1-8}
\multirow{6}{*}[-0.2em]{7} & \multirow{3}{*}{1} & \multirow{3}{*}{0}       & MSE & 0.025 & 0.026 & 0.023 & 0.029 \\  
&&& size (adj4)& 6.3 & 4.9 & 5.1 & 3.7 \\  
&&&  s.e. (adj4)& 0.15 & 0.16 & 0.15 & 0.19 \\ %
\cmidrule{2-8}
& \multirow{3}{*}{2} & \multirow{3}{*}{0.21} &MSE& 0.028 & 0.031 & 0.029 & 0.033 \\ 
&&&size (adj4)& 5.0 & 5.2 & 5.1 & 4.6 \\  
&&&s.e. (adj4)& 0.17 & 0.17 & 0.17 & 0.19 \\ 

\cmidrule{1-8}
\multirow{9}{*}[-0.2em]{8} & \multirow{3}{*}{1} & \multirow{3}{*}{0}   &MSE& 0.31 & 0.34 & 0.28 & 0.47 \\ 
&&& size (adj4)& 5.2 & 5.4 & 3.4 & 5.8 \\ 
&&& s.e. (adj4)& 0.56 & 0.57 & 0.56 & 0.67 \\ 

\cmidrule{2-8}
& \multirow{3}{*}{2} & \multirow{3}{*}{0.041} &MSE& 0.34 & 0.35 & 0.31 & 0.45 \\ 
&&& size (adj4)& 5.6 & 4.8 & 3.6 & 4.8 \\ 
&&& s.e. (adj4)& 0.59 & 0.59 & 0.58 & 0.67 \\  

\cmidrule{2-8}
& \multirow{3}{*}{3} & \multirow{3}{*}{1.07} &MSE& 0.46 & 0.52 & 0.48 & 0.59 \\ 
&&& size (adj4)& 4.5 & 6.0 & 4.7 & 5.1 \\
&&& s.e. (adj4)& 0.70 & 0.71 & 0.69 & 0.75 \\  

\cmidrule{1-8}
\multirow{9}{*}[-0.2em]{9} & \multirow{3}{*}{1} & \multirow{3}{*}{0}   &MSE& 0.0067 & 0.0068 & 0.0065 & 0.0094 \\ 
&&& size (adj4)& 4.8 & 4.8 & 5.2 &  6.0 \\ 
&&& s.e. (adj4)& 0.082 & 0.083 & 0.082 & 0.010 \\      

\cmidrule{2-8}
& \multirow{3}{*}{2} & \multirow{3}{*}{-0.10} &MSE& 0.0080 & 0.0083 & 0.0085 & 0.0094 \\
&&& size (adj4)& 5.4 & 6.0 & 6.2 & 5.6 \\
&&& s.e. (adj4)& 0.088 & 0.090 & 0.088 & 0.097 \\

\cmidrule{2-8}
& \multirow{3}{*}{3} & \multirow{3}{*}{-0.012} &MSE& 0.0068 & 0.0074 & 0.0076 & 0.0079 \\   
&&& size (adj4)& 5.8 & 6.5 & 7.2 & 5.6 \\ 
&&& s.e. (adj4)& 0.082 & 0.082 & 0.081 & 0.088 \\

\cmidrule{1-8}
\multirow{9}{*}[-0.2em]{10} & \multirow{3}{*}{1} & \multirow{3}{*}{0}   &MSE& 52.82 & 39.16 & 40.30 & 52.49 \\
&&& size (adj4)& 5.1 & 4.4 & 3.7 & 4.8 \\ 
&&& s.e. (adj4)& 7.07 & 6.55 & 6.56 & 7.29 \\     

\cmidrule{2-8}
& \multirow{3}{*}{2} & \multirow{3}{*}{5.61} &MSE& 85.45 & 59.39 & 56.69 & 84.05 \\
&&&size (adj4) & 6.3 & 4.4 & 6.1 & 5.0 \\ 
&&& s.e. (adj4)& 8.78 & 7.62 & 7.55 & 9.26 \\ 

\cmidrule{2-8}
& \multirow{3}{*}{3} & \multirow{3}{*}{1.91} &MSE& 71.99 & 56.51 & 54.92 & 79.53 \\ 
&&&size (adj4) & 4.8 & 5.3 & 4.8 & 5.3 \\ 
&&&s.e. (adj4) & 8.47 & 7.60 & 7.57 & 8.63 \\
\hline\hline
\end{tabular}
\end{adjustbox}
\label{table:sims-others-6-10}
\begin{tablenotes} \footnotesize
\item For each stratification, I report (1) the MSE, (2) the size of testing \eqref{eq:H0} for $\theta_0 = \theta$ at the 5\% level, in percentage, and (3) the average standard error. The tests used in this table are: for MS, the adjusted $t$-test with the variance estimator in \eqref{eq:four-se}; for the regression-adjusted estimator, the $t$-test with White's heteroskedasticity-robust standard error. Rows are labeled according to the papers and models. Columns are labeled according to the stratifications. The definitions of the stratifications can be found in the text.
\end{tablenotes}
\end{table}

\begin{table}[ht]
\caption{The size of the test in \cite{athey2017econometrics} for matched sets of four}
\begin{adjustbox}{max width=0.3\linewidth,center}
\begin{tabular}{cccc}
\hline\hline
& & (b) & (d)\\
Paper & Model & MS X & MS base \\ 
\cmidrule{1-4}
\multirow{2}{*}{1} &1& 6.0 &  -\\
&2& 8.2 &  -\\
\cmidrule{1-4}
\multirow{3}{*}{2} &1& 5.1 & 4.8 \\
&2& 6.3 & 5.1\\
&3& 6.3 & 5.1\\
\cmidrule{1-4}
\multirow{3}{*}{3} &1& 5.7 & 5.2 \\
&2& 8.7 & 5.8\\
&3& 8.0 & 5.7\\
\cmidrule{1-4}
\multirow{3}{*}{4} &1& 5.0 & 5.0 \\
&2& 7.6 & 4.2\\
&3& 6.4 & 5.9\\
\cmidrule{1-4}
\multirow{3}{*}{5} &1& 6.0 & 6.6\\
&2& 7.7 & 9.6\\
&3& 7.7 & 9.0\\
\cmidrule{1-4}
\multirow{2}{*}{6} &1& 5.2 &   - \\
&2& 7.7 &   -\\
\cmidrule{1-4}
\multirow{2}{*}{7} &1& 5.5 &   -\\
&2& 8.5 &   -\\
\cmidrule{1-4}
\multirow{3}{*}{8} &1& 4.4 & 4.8 \\
&2&7.1 & 5.0\\
&3& 8.5 & 5.6\\
\cmidrule{1-4}
\multirow{3}{*}{9} &1& 6.2 & 5.3 \\
&2& 8.7 & 5.2\\
&3& 10.4 & 6.2\\
\cmidrule{1-4}
\multirow{3}{*}{10} &1& 5.3 & 5.8\\
&2& 6.1 & 6.3\\
&3& 5.2 & 6.6\\
\hline\hline
\end{tabular}
\end{adjustbox}
\label{table:sims-ai}
\begin{tablenotes} \footnotesize
\item This table shows the size of the test in \cite{athey2017econometrics} for testing \eqref{eq:H0} for $\theta_0 = \theta$ at the 5\% level, in percentage. Rows are labeled according to papers and models. Columns are labeled according to the stratifications. The definitions of the stratifications can be found in the main text.
\end{tablenotes}
\end{table}

\clearpage

Here are the details of all the data used in simulation:
\begin{enumerate}
	\item \cite{herskowitz2021gambling}:

	I re-implement the analysis on p.93 of the original paper and estimate the effect of lumpy prime on demand. I use Wave 2 data from ``Panel-Clean.dta''. There are 997 observations for simulation. I fill in the missing values and choose the covariates following the "OVERALL SPECIFICATION COVARIATES" part in Analysis-3.do. 20\% of the original data are sampled with replacement and fixed throughout the replications to be used as the pilot data, which contains 199 observations. I consider Model 1 and Model 2 for data imputation. When drawing from the empirical distribution, 996 observations are sampled with replacement so the sample size is a multiple of four.

	dependent: lmp\_matrix (an indicator for whether the maximum number of tickets was demanded.)

	covariates: meaninc (mean income for duration of study), betmean\_prop (mean amount spent on betting / mean income during study), lmp\_purchased, lumpy\_incprop

	treatment: lumpyprime (lumpy expenditure prime treatment group)

	\item \cite{lee2021poverty}:

	I re-implement the analysis in (1) on p.49 of the original paper and estimate the effect of treatment on total remittances sent from migrants. I rerun the 1--3 code files in folder ``Migrant-Survey'' and get the migrants data ``Endline-Data-Combined-Status-Merged-Ready-18-Active-Acc.dta'' for regression. There are 809 observations. 20\% of the original data are sampled with replacement and fixed throughout the replications to be used as the pilot data, which contains 161 observations. I consider Models 1--3 for data imputation. When drawing from the empirical distribution, 808 observations are sampled with replacement so the sample size is a multiple of four.

	dependent: log\_total\_remittances (missing values generated because of log transformation and are filled using 0)

	covariates: log\_total\_remittances\_b (baseline outcome), household\_size (missing values are filled using baseline outcomes), hohh\_age, hohh\_female, and hohh\_completed\_primary 

	treatment: treatment

	\item \cite{abel2020value}:

	I re-implement the analysis in (4) on p.56 of the original paper and estimated the effect of reference letters on employment. I used experiment 2 data from ``experiment2\_employment.dta''. There are 1000 observations. 20\% of the original data are sampled with replacement and fixed throughout the replications to be used as the pilot data, which contains 200 observations. I consider Models 1--3 for data imputation. When drawing from the empirical distribution, 1000 observations are sampled with replacement so the sample size is a multiple of four.

	dependent: f2\_b3 (number of jobs applied in the last four months, 246 observations with missing values are dropped. the treatment percentage did not change much, about 0.55)

	covariates: bs\_c3\_jobs\_applied (baseline outcome), age\_yr, female\_d, educ\_yr (missing values are filled using the mean), married\_d, lang\_zulu\_d, lang\_xhosa\_d, lang\_venda\_d (there are 4 languages and here we used 3 dummy variables)

	treatment: reference\_d

	original stratification: gender, total 2

	\item \cite{gerber2020one}:

	I re-perform the OLS on p.303 of the original paper and estimat the effect of the close poll treatment on vote margin predictions using data from 2010 RCT experiment and following lines 677--710 in code file ``Main\_e2010.do''. Because the observations in the control group do not have any experiment records, I only consider the two treatment groups ``close'' and ``not close.'' There are 6650 observations. 20\% of the original data are sampled with replacement and fixed throughout the replications to be used as the pilot data, which contains 1330 observations. I consider Models 1--3 for data imputation. When drawing from the empirical distribution, 6648 observations are sampled with replacement so the sample size is a multiple of four.

	dependent: votemarg\_post (post-treat vote margin prediction)  

	covariates: int\_gov\_scale (interest in politics, 1-5 scale), pelosi (identify Nancy Pelosi as speaker), vote\_admin\_past (share voted previous 5 elections), male, race (4 dummy variables), schooling (schooling years), inc(1-5 scale), age (1-7 scale)

	treatment: t\_close

	original stratification: ppstaten, in each replication, only consider states with more than 5 observations.

	\item \cite{deserranno2019leader}:

	I re-implement the group-level analysis on p.263 of the original paper, which is also shown in Table A.15 of the online appendix, and estimate the effect of treatment on wealth score inequality. I follow lines 1250--1289 in code file ``AEJApp-2018-0248\_Tables-and-Figures.do'' and transfer member-level panel data into group-level using IQR function. After discarding observations with missing group indices, I finally get 92 groups. 20\% of the original data are sampled with replacement and fixed throughout the replications to be used as the pilot data, which contains 18 observations. I consider Models 1--3 for data imputation. When drawing from the empirical distribution, 92 observations are sampled with replacement so the sample size is a multiple of four.

	dependent: iqrwealth (endline group-level wealth scores, generated from wealth\_endline)  

	covariates: branch (scale from 1 to 9, numeric)

	treatment: vote

	\item \cite{barrera-osorio2019medium-}:

	I re-implement the analysis on p.268, table 3, column (1) of the original paper and estimate the impacts of the basic treatment on on-time secondary enrollment outcomes. I follow lines 144--151 in code file ``Final\_Tables\_Journal.do.'' Observations with missing values on the dependent variable are filtered out. Moreover, variables ending in ``\_missing'' recorded the missing status of correspondent variables ending in ``\_fill'' and I impute ``\_fill'' variables using the median. Running one replication using full data takes 6 hours, so I randomly sample 7880 out of the 15759 observations to reduce the running time. 20\% of the original data are sampled with replacement and fixed throughout the replications to be used as the pilot data, which contains 1576 observations. I consider Models 1--2 for data imputation. When drawing from the empirical distribution, 7880 observations are sampled with replacement so the sample size is a multiple of four.

	dependent: on\_time (binary, whether enrolled or not)

	covariates: s\_teneviv\_fill (indicator of house ownership), s\_utilities\_fill (number of utilities in the house), s\_durables\_fill (number of durable goods), s\_infraest\_hh\_fill (infrastructure in the household, scale 0-22) , s\_age\_sorteo\_fill (age at the moment of lottery), s\_sexo\_fill (gender of student), s\_yrs\_fill (years of education), grade\_fill (grade at baseline), s\_single\_fill (if the household is single headed), s\_edadhead\_fill (age of the head of the household), s\_yrshead\_fill (years of education of the head), s\_tpersona\_fill (number of individuals in the house, scale 0-22), s\_num18\_fill (number of people under 18 in the house), s\_estrato\_fill (strata of the household, scale 0-2), s\_puntaje\_fill (SISBEN score), s\_ingtotal\_fill (income, from 0-4000). 

	treatment: treatment

	original stratification: grader(6-11), gender, total 12 

	\item \cite{himmler2019soft}:

	I re-implement the analysis of table 2, column (9) on p.130 of the original paper and estimate the effect of commitment treatment on the number of exams passed. I use data in ``soft\_commitments\_AEJ.dta'' and followed lines 74--77 in code file ``soft\_commitments\_AEJ.do.'' There are 392 observations, and 32.91\% are assigned to the treatment group in the original paper. Then, 20\% of the original data are sampled with replacement and fixed throughout the replications to be used as the pilot data, which contains 78 observations. There are no baseline outcomes. Models 1--2 are considered for data imputation. When drawing from the empirical distribution, 392 observations are sampled with replacement so the sample size is a multiple of four. 

	dependent: pass\_all (number of exams passed) 

	covariates: male, c\_HSGPA (centered high school GPA), age (scale 1-21, generated from original dummy variables dage1 - dage21),  dschooltype1 , dschooltype2 (two binary school type variables), nongerman (foreign citizership), c\_app\_day (centered application day), freshdeg (fresh HS degree).

	treatment: commitment

	original stratification: gpa, 4 strata

	\item \cite{abel2019bridging}:

	I re-implement the analysis of table 3, column (2) on p.292 of the original paper and estimate the effects of WorkshopPlus on search hours. I use data in file ``final\_data2.dta'' and follow lines 159--163 in code file ``AP\_final\_AEJ.do.'' I only consider the control group vs the ``workshopPlus'' group and filter out observations originally assigned to the ``workshop'' group. There are two outcomes in the experiment. In order to be consistent with the second simulation regarding to the second outcome ``application number,'' I discard the observations missing any ``search hours'' or ``application numbers'' values. I get 1479 observations, 45.57\% of which are assigned to the treatment group originally. 20\% of the original data are sampled with replacement and fixed throughout the replications to be used as pilot data, which contains 295 observations. Models 1-3 are considered for data imputation. Model 2 used baseline1 + covariates. When drawing from the empirical distribution, 1476 observations are sampled with replacement so the sample size is a multiple of four.

	dependent: b2\_t (search hours)  

	covariates: educ\_yr (education years), age\_yr (age in years), female\_d, lang\_ (three dummy variables indicating spoken languages), location\_f\_ (two dummy variables indicating location), round (follow-up 2)

	baseline: nomiss\_bs\_c2 (variable miss\_bs\_c2 is an indicator for missing values, and we filled  in missing baseline outcomes using the median)

	treatment: ws\_plus\_d

	secondary outcome: b3\_t (application numbers)

	\item \cite{de_mel2019labor}:

	I re-implement the analysis in table 5 panel A ``Number of paid workers,'' column ``After subsidy Year 3+'' on p.220 of the original paper and estimated the effect of treatment on employment. I follow lines 59--76 to define treatment status and lines 585--606 to perform regression in code file ``AEJreplicationfile\_LaborDrops.do'' using data in ``Sri-Lanka-Panel-Experiment-Paper.dta''. There are 12 rounds of experiments and the means of rounds 10-12 outcomes are treated as endline (year 3+) outcomes. I filter out observations whose outcomes are missing in any of these three rounds and impute the covariates with missing values using the median. I finally got 454 observations. 20\% of the original data are sampled with replacement and fixed throughout the replications to be used as the pilot data, which contains 90 observations. I consider Models 1--3 for data imputation. When drawing from the empirical distribution, 452 observations are sampled with replacement so the sample size is a multiple of four.

	dependent: allpaid\_trunc (number of all paid works, calculated using mean of round 10-12 grouped by the key variable sheno).

	baseline: baseallpaid

	covariates: basetotalscore, booster, raven, digitspan, baselK\_noland, baseedn, baseprofits. (Did not find descriptions. Variables baselowassets, basehighassets, baselowprofits, basehighprofits are dismissed because they make the matrix not invertible)

	treatment: voucheronly

	original stratification: there are 6 variables (strata1-strata6)

	\item \cite{lafortune2018role}:

	I re-implement the analysis of interactions in Panel A, Column ``Income per capita,'' Table 5 on p.242 of the original paper and estimate the effect of role models on income per capita. I run the provided ``.do'' files to generate the dataset ``Base\_Analisis\_SEG1.dta'' according to the ``Read-me.pdf''. I then perform data processing and regression analysis on this dataset following the codes in ``Interacciones.do''. After removing observations with missing outcomes and imputing missing gender and age values using the median, I finally get 979 observations. 20\% of the original data are sampled with replacement and fixed throughout the replications to be used as the pilot data, which contains 195 observations. I consider Models 1--3 for data imputation. When drawing from the empirical distribution, 976 observations are sampled with replacement so the sample size is a multiple of four. 

	dependent: IncomePC\_seg1 (income per capita)

	baseline: lb\_IncomePC\_seg1

	covariates: Edad (age), mujer (gender), Ed2, Ed3 (education dummy variables, ignore Ed1), NO\_info\_Educ (indicating missing values of these three education variables), NegocioBsico030\_fi, NegocioIntermedio300M1M\_fi (description written in foreign language, but are dummy variables indicating different amount of money, ignore NegocioDesarrollado1MM\_fi ), NO\_info\_Size (indicating missing values of  former three ``Nego\_'' variables

	treatment: asignado\_role\_Model
\end{enumerate}

\section{Matched-Pair Designs in the AEA RCT Registry} \label{sec:aea}
The following experiments in the AEA RCT Registry use matched-pair designs: AEARCTR-0000086, 0000171, 0000293, 0000443, 0000481, 0000550, 0000578, 0000587, 0000644, 0000688, 0000721, 0000983, 0000986, 0001034, 0001097, 0001218, 0001370, 0001591, 0001607, 0001712, 0001714, 0001778, 0001992, 0001995, 0002010, 0002125, 0002132, 0002282, 0002585, 0002622, 0002664, 0002750, 0002776, 0003056, 0003076, 0003524, 0003581, 0003629, 0003648, 0003779, 0003814, 0003933, 0003994, 0004024, 0004042, 0004022, 0006706.

\end{document}